\documentclass[letterpaper,onecolumn,11pt,unpublished]{quantumarticle}
\pdfoutput=1
\usepackage[utf8]{inputenc}
\usepackage[english]{babel}
\usepackage[T1]{fontenc}
\usepackage{qcircuit}
\usepackage{subcaption}
\usepackage{amsfonts}
\usepackage{amsmath,amssymb}
\usepackage{graphicx}           
\usepackage{relsize}
\usepackage{tcolorbox}
\usepackage{amsthm}
\usepackage{thm-restate}
\usepackage{braket}
\usepackage{csquotes}

\usepackage{soul}
\usepackage{mathtools}

\usepackage{cite}
\usepackage[figurename=Fig.,font=footnotesize]{caption}
\usepackage{multirow}
\usepackage{enumitem}

\theoremstyle{definition}
\usepackage{float}
\makeatletter
\newlength\min@xx

\makeatother

\usepackage{hyperref}
\hypersetup{
    colorlinks=true,
    linkcolor=blue,
    citecolor=magenta,
    filecolor=magenta,      
    urlcolor=blue,
    pdftitle={},
    pdfpagemode=FullScreen,
    }

\usepackage{lipsum}
\usepackage[capitalize]{cleveref}
\usepackage{bbold}

\newtheorem{theorem}{Theorem}
\newtheorem{lemma}{Lemma}
\newtheorem{corollary}{Corollary}

\newtheorem{algorithm}{Algorithm}

\def\ket#1{\left| #1 \right\rangle}

\newcommand{\codepar}[1]{\ensuremath{[\![#1]\!]}}

\DeclareMathOperator{\wt}{wt}

\title{Synchronizable hybrid subsystem codes}

\author{Theerapat Tansuwannont}
\affiliation{Center for Quantum Information and Quantum Biology, Osaka University, Toyonaka, Osaka 560-0043, Japan}
\email{t.tansuwannont.qiqb@osaka-u.ac.jp}

\author{Andrew Nemec}
\affiliation{Department of Computer Science, University of Texas at Dallas, Richardson, TX 75080, USA}
\affiliation{Duke Quantum Center, Duke University, Durham, NC 27701, USA}
\affiliation{Department of Electrical and Computer Engineering, Duke University, Durham, NC 27708, USA}
\email{andrew.nemec@utdallas.edu}

\begin{document}

\maketitle

\begin{abstract}
Quantum synchronizable codes are quantum error correcting codes that can correct not only Pauli errors but also errors in block synchronization. The code can be constructed from two classical cyclic codes $\mathcal{C}$, $\mathcal{D}$ satisfying $\mathcal{C}^{\perp} \subset \mathcal{C} \subset \mathcal{D}$ through the Calderbank-Shor-Steane (CSS) code construction. In this work, we establish connections between quantum synchronizable codes, subsystem codes, and hybrid codes constructed from the same pair of classical cyclic codes. We also propose a method to construct a synchronizable hybrid subsystem code which can correct both Pauli and synchronization errors, is resilient to gauge errors by virtue of the subsystem structure, and can transmit both classical and quantum information, all at the same time. The trade-offs between the number of synchronization errors that the code can correct, the number of gauge qubits, and the number of logical classical bits of the code are also established. In addition, we propose general methods to construct hybrid and hybrid subsystem codes of CSS type from classical codes, which cover relevant codes from our main construction.
\end{abstract}

\section{Introduction} \label{sec:intro}

Large-scale quantum computing, which requires quantum error correction, fault-tolerant quantum computation, and quantum communications as key ingredients, has advanced rapidly in the past few decades. Recent breakthroughs include good quantum low-density parity-check (LDPC) codes \cite{BE21,HHJO21,EKZ22,LZ22,PK22a,PK22b,DHLV23}, fault-tolerant computation with low overheads \cite{FGL20,CKBB22,TDB22,BCGM+24,XBPR+24,YK24,ZZCB+24}, and experimental realizations of small-sized codes with logical error rates close to the break-even point \cite{EDNR+21,EPMP+21,RBLG+21,PHPR+22,RBAA+22,google23,SERS+23,BEGL+24,DRBC+24,Google2024}. Modular architectures for various platforms \cite{MRRB+14,HIVC+15,BKM16,BCRW+19,BKLN+21,GPSR+21,BDGG22,BBBC+23,NZLQ+23,SB23,RSBV24} have also been proposed to scale-up the system size. To connect between modules, communication through noisy quantum channels is thus required.

In addition to the effects of decoherence and faulty circuits, quantum communications can be affected by errors in block synchronization. In both classical and quantum communications, information is often encoded into an error-correcting code block, and multiple blocks are often sent consecutively. To decode the information correctly, it is important that the receiver correctly identifies the boundaries between each block of information; that is, the first bit (or qubit) of the block defined by the sender and the first bit (or qubit) of the block marked by the receiver must be the same. If the blocks become misaligned and the boundary is misidentified, the information may be read or decoded incorrectly. The process of block synchronization can be challenging if the sender and the receiver are far apart and an external mechanism to correct block alignment cannot be simply applied.

One solution to the block synchronization problem in the classical setting is to insert marker bits into the code block or use some bit pattern as a preamble (see \cite{sklar2021,bregni2002} for example). However, these procedures cannot be directly applied in the quantum setting as they require qubit measurement, which could potentially destroy quantum information in the code block. One solution in the case of quantum information is to make use of quantum synchronizable codes (QSCs), first introduced by Fujiwara in \cite{Fujiwara2013}. Here, the author provides a construction of QSCs from a nested pair of self-orthogonal classical cyclic codes via the Calderbank-Shor-Steane (CSS) construction of quantum codes \cite{Calderbank1996, Steane1996}. Since then, multiple families of QSCs have been studied \cite{Fujiwara2013-Jul, Fujiwara2013-Nov, Xie2014, Xie2016, Guenda2017, Luo2018}.

Other variants of quantum error-correcting codes that specialize in different aspects of error correction have also received recent attention. Subsystem codes \cite{Kribs2005, Poulin2005, Bacon2006} allow for greater flexibility in the decoding of quantum codes due to a gauge subsystem, while hybrid quantum-classical codes \cite{Kremsky2008, Grassl2017, Nemec2018, Nemec2021} simultaneously encode both quantum and classical information together in a single code. Recently, these two code families have been unified into hybrid subsystem codes \cite{Nemec2022, Dauphinais2024}.


In this work, we generalize the notion of QSCs in stabilizer formalism and show that from the same pair of classical cyclic codes that can be used to construct a QSC, it is possible to construct hybrid codes, subsystem codes, and hybrid subsystem codes that can also protect against block synchronization errors; our main results are stated in \cref{thm:QSC_unified}. These codes are unified into a family of synchronizable hybrid subsystem codes. It is possible to show that QSCs developed in \cite{Fujiwara2013} can be described in terms of hybrid codes, where no classical information is initially transmitted but the subspace contains classical information about synchronization errors upon decoding. We also establish trade-offs between the number of synchronization errors the code can correct, the number of gauge qubits of the code, and the amount of encoded classical information in \cref{cor:QSC_tradeoff}. 

In addition to the construction of synchronizable hybrid subsystem codes, we also propose general methods to construct hybrid and hybrid subsystem codes of CSS type from classical codes in \cref{thm:CSS_hybrid,thm:CSS_hybrid_subsystem}. The classical codes that can be used to construct hybrid and hybrid subsystem codes by these theorems are not limited to classical cyclic codes, although the resulting quantum codes might not be able to correct synchronization errors. We can show that \cref{thm:CSS_hybrid,thm:CSS_hybrid_subsystem} also cover the constructions of non-synchronizable codes obtained by \cref{thm:QSC_unified}.

This paper is organized as follows. \cref{sec:prelim} covers background information in classical and quantum error correction, including the notions of classical cyclic codes, the CSS code construction, and hybrid subsystem codes. \cref{sec:QSC_stabilizer} reviews the QSC construction from \cite{Fujiwara2013} as well as their encoding and decoding in stabilizer formalism. In \cref{sec:cyclic_to_Pauli}, we construct anticommuting pairs of Pauli operators from the generators of classical cyclic codes, which are useful in our main construction. In \cref{sec:SHSC_main}, we construct a family of synchronizable hybrid subsystem codes and show trade-offs between the number of correctable synchronization errors, the number of gauge qubits, and the amount of encoded classical information. General methods to construct subsystem codes, hybrid codes, and hybrid subsystem codes of CSS type are provided in \cref{sec:hybrid_subsystem_construction}. Finally, we discuss these results and give some concluding remarks in \cref{sec:discussion}.

\section{Preliminaries} \label{sec:prelim}

In this section, we provide a brief overview of certain families of classical and quantum codes, which is necessary for understanding the main construction in this work.

\subsection{Classical cyclic codes} \label{subsec:pre_cyclic}

Consider the binary field $\mathbb{Z}_2=\{0,1\}$. We denote the set of polynomials in $x$ with binary coefficients by $\mathbb{Z}_2[x]$, and let $R_n=\mathbb{Z}_2[x]/(x^n-1)$ be a polynomial ring over $\mathbb{Z}_2$. Any polynomial $c(x)\in R_n$ is of the form $c(x)=c_0+c_1x+\dots+c_{n-1}x^{n-1}$, where $c_0,c_1,\dots,c_{n-1} \in \mathbb{Z}_2$. We define $\mathcal{V}:R_n \rightarrow \mathbb{Z}_2^n$ to be the vectorization function such that $\mathcal{V}\left(c(x)\right)=(c_0,c_1,\dots,c_{n-1})$.

Let $\mathbf{c} = (c_0,c_1,\dots,c_{n-1}) \in \mathbb{Z}_2^n$ be an $n$-bit binary row vector. We define $\mathcal{O}: \mathbb{Z}_2^n \times \mathbb{Z}_n \rightarrow \mathbb{Z}_2^n$ to be the right cyclic shift function such that $\mathcal{O}(\mathbf{c},\alpha)\!=\!(c_{n-\alpha},c_{n-\alpha+1},\dots,c_{n-1},c_0,\dots,c_{n-\alpha-1})$, i.e., $\mathcal{O}(\mathbf{c},\alpha)$ is the vector obtained by cyclically shifting $\mathbf{c}$ to the right by $\alpha$ positions. For convenience, we also write $\mathcal{O}(\mathbf{c},-\alpha)=\mathcal{O}(\mathbf{c},n-\alpha)$, which refers to the vector obtained by cyclically shifting $\mathbf{c}$ to the left by $\alpha$ positions.

Let $\mathcal{C}$ be an $[n,k_c,d_c]$ classical binary linear code that encodes $k_c$ bits into $n$ bits and has code distance $d_c$ (where the distance of a classical code is defined by the minimum Hamming weight of non-zero code words). $\mathcal{C}$ is called a \emph{cyclic code} if for all $\mathbf{c} \in \mathcal{C}$ and for all $\alpha \in \{0,\dots,n-1\}$, $\mathcal{O}(\mathbf{c},\alpha) \in \mathcal{C}$; i.e., a cyclic shift of any code word in $\mathcal{C}$ is also a code word in $\mathcal{C}$. For each cyclic code $\mathcal{C}$ of length $n$, there is a principal ideal $I_{\mathcal{C}}$ of $R_n$ such that $I_{\mathcal{C}}$ and $\mathcal{C}$ are isomorphic, denoted by $I_{\mathcal{C}}\cong\mathcal{C}$. Note that $I_{\mathcal{C}}$ is a set of polynomials in $x$ with binary coefficients while $\mathcal{C}$ is a set of $n$-bit binary vectors, and there is a one-to-one correspondence between a polynomial $c(x)$ in $I_{\mathcal{C}}$ and a binary vector $\mathbf{c} =\mathcal{V}\left(c(x)\right)$ in $\mathcal{C}$. Observe that $\mathcal{V}\left(x^\alpha c(x)\right)=\mathcal{O}(\mathcal{V}\left(c(x)\right),\alpha)$; i.e., multiplying a polynomial $c(x) \in I_{\mathcal{C}}$ by $x^\alpha$ corresponds to cyclically shifting its corresponding code word $\mathbf{c} =\mathcal{V}\left(c(x)\right) \in \mathcal{C}$ to the right by $\alpha$ positions.

For each $I_{\mathcal{C}}$, there is a unique monic polynomial $p(x)$ of minimal degree in $I_{\mathcal{C}}$. $p(x)$ is called a \emph{generator polynomial} of $I_{\mathcal{C}}$, and any 
polynomial $c(x)\in I_{\mathcal{C}}$ can be written as $c(x)=i(x)p(x)$ for some polynomial $i(x)$ of degree $<k_c$. Suppose that $I_{\mathcal{C}}$ corresponds to a cyclic code $\mathcal{C}$ with parameter $[n,k_c,d_c]$. A generator polynomial $p(x)$ of $I_{\mathcal{C}}$ has degree $n-k_c$ and can be written as $p(x)=p_0+p_1x+\dots+p_{n-k_c}x^{n-k_c}$ for some $p_0,\dots,p_{n-k_c}$. We define $\mathbf{p}_i= \mathcal{V}\left(x^{i-1}p(x)\right)=\mathcal{O}(\mathcal{V}\left(p(x)\right),i-1)$ for $i \in \{1,\dots,k_c\}$, which are generators (or basis vectors) of $\mathcal{C}$. Therefore, we can write $\mathcal{C}=\langle \mathbf{p}_i \rangle$, $i \in \{1,\dots,k_c\}$. A generator matrix of $\mathcal{C}$ is
\begin{equation}
G_\mathcal{C} =
    \begin{pmatrix}
        \mathbf{p}_1 \\
        \mathbf{p}_2 \\
        \dots \\
        \mathbf{p}_{k_c}
    \end{pmatrix} =
    \begin{pmatrix}
        p_0 & p_1 & p_2 & \cdots & p_{n-k_c} & 0 & \dots & 0 \\
        0 & p_0 & p_1 & \cdots & p_{n-k_c-1} & p_{n-k_c} & \dots & 0 \\
        \dots &   &   &   &   &   &   & \dots \\
        0 & \dots & 0 & p_0 & p_1 & p_2 & \dots  & p_{n-k_c}
    \end{pmatrix}.
\end{equation}

We can define $\tilde{p}(x)=(x^n-1)/p(x)$ to be the check polynomial of $I_{\mathcal{C}}$. This comes from the fact that for any polynomial $c(x)\in I_{\mathcal{C}}$ of the form $c(x)=i(x)p(x)$, $\tilde{p}(x)c(x)=\tilde{p}(x)i(x)p(x)=0 \;\left(\mathrm{mod}\;(x^n-1)\right)$. $\tilde{p}(x)$ has degree $k_c$ and can be written as $\tilde{p}(x)=\tilde{p}_0+\tilde{p}_1x+\dots+\tilde{p}_{k_c}x^{k_c}$ for some $\tilde{p}_0,\dots,\tilde{p}_{k_c}$. Here we define a polynomial $\tilde{p}_\mathcal{R}(x)=\tilde{p}_{k_c}+\tilde{p}_{k_c-1}x+\dots+\tilde{p}_{0}x^{k_c}$ by sorting the coefficients of the 0-th degree term to the $k_c$-th degree term of $\tilde{p}(x)$ in reverse order. Note that $\tilde{p}_\mathcal{R}(x)$ also has degree $k_c$\footnote{$\tilde{p}_\mathcal{R}(x)$ has the same degree as $\tilde{p}(x)$ since $\tilde{p}_{0}$ are not zero; if $\tilde{p}_{0}$ is zero, both $\tilde{p}(x)$ and $x^{-1}\tilde{p}(x)$ generate the same ideal $I_{\mathcal{R}(\mathcal{C}^\perp)}$ (the principal ideal corresponds to the cyclic code $\mathcal{R}(\mathcal{C}^\perp)$ obtained by reversing the bit order of $\mathcal{C}^\perp$), which contradicts the fact that $\tilde{p}(x)$ is of minimal degree in $I_{\mathcal{R}(\mathcal{C}^\perp)}$.}. We also define $\tilde{\mathbf{p}}_i= \mathcal{V}\left(x^{i-1}\tilde{p}_\mathcal{R}(x)\right)=\mathcal{O}(\mathcal{V}\left(\tilde{p}_\mathcal{R}(x)\right),i-1)$ for $i \in \{1,\dots,n-k_c\}$. A check matrix of $\mathcal{C}$ is 
\begin{equation}
H_\mathcal{C} =
    \begin{pmatrix}
        \tilde{\mathbf{p}}_1 \\
        \tilde{\mathbf{p}}_2 \\
        \dots \\
        \tilde{\mathbf{p}}_{n-k_c}
    \end{pmatrix} =
    \begin{pmatrix}
        \tilde{p}_{k_c} & \tilde{p}_{k_c-1} & \tilde{p}_{k_c-2} & \cdots & \tilde{p}_{0} & 0 & \dots & 0 \\
        0 & \tilde{p}_{k_c} & \tilde{p}_{k_c-1} & \cdots & \tilde{p}_{1} & \tilde{p}_{0} & \dots & 0 \\
        \dots &   &   &   &   &   &   & \dots \\
        0 & \dots & 0 & \tilde{p}_{k_c} & \tilde{p}_{k_c-1} & \tilde{p}_{k_c-2} & \dots  & \tilde{p}_{0}
    \end{pmatrix}.
\end{equation}

A check matrix of $\mathcal{C}$ is also a generator matrix of the dual code $\mathcal{C}^\perp=\{\mathbf{w} \in \mathbb{Z}_2^n\,|\,\mathbf{c}\cdot\mathbf{w}=0 \;\forall \mathbf{c}\in \mathcal{C}\}$ where $\mathbf{c} \cdot \mathbf{w}=\sum_{i=0}^{n-1} c_i w_i$. We can write $\mathcal{C}^\perp=\langle \tilde{\mathbf{p}}_i \rangle$, $i \in \{1,\dots,n-k_c\}$. We also define $I_{\mathcal{C}^\perp}$ to be the principal ideal of $R_n$ that corresponds to $\mathcal{C}^\perp$, whose generator polynomial is $\tilde{p}_\mathcal{R}(x)$. Similar to $I_{\mathcal{C}}$, multiplying a polynomial $w(x) \in I_{\mathcal{C}^\perp}$ by $x^\alpha$ corresponds to cyclically shifting its corresponding code word $\mathbf{w} =\mathcal{V}\left(w(x)\right) \in \mathcal{C}^\perp$ to the right by $\alpha$ positions.

Let $\mathcal{D}$ be another cyclic code with parameter $[n,k_d,d_d]$. We denote a generator polynomial of the corresponding principal ideal $I_\mathcal{D}$ by $q(x)$. The notations of mathematical objects related to cyclic codes $\mathcal{C}, \mathcal{C}^\perp, \mathcal{D}, \mathcal{D}^\perp$ that will be used throughout this work are summarized in \cref{tab:code_notation}. 

\begin{table}[htbp]
    \centering
    \begin{tabular}{|c|c|c|c|c|}
        \hline
         Cyclic code & Parameters & Generators & Principal ideal & Generator polynomial \\
         \hline
         \multirow{2}{*}{$\mathcal{D}^\perp$} & \multirow{2}{*}{$[n,n-k_d]$} & $\tilde{\mathbf{q}}_i$ & \multirow{2}{*}{$I_{\mathcal{D}^\perp}$} & $\tilde{q}_\mathcal{R}(x)$ \\
            &   & $i \in \{1,\dots,n-k_d\}$  &   & of degree $k_d$  \\
         
         \hline
         \multirow{2}{*}{$\mathcal{C}^\perp$} & \multirow{2}{*}{$[n,n-k_c]$} & $\tilde{\mathbf{p}}_i$ & \multirow{2}{*}{$I_{\mathcal{C}^\perp}$} & $\tilde{p}_\mathcal{R}(x)$ \\
            &   & $i \in \{1,\dots,n-k_c\}$  &   & of degree $k_c$  \\
         \hline
         \multirow{2}{*}{$\mathcal{C}$} & \multirow{2}{*}{$[n,k_c,d_c]$} & $\mathbf{p}_i$ & \multirow{2}{*}{$I_\mathcal{C}$} & $p(x)$ \\
            &   & $i \in \{1,\dots,k_c\}$  &   & of degree $n-k_c$  \\
         \hline
         \multirow{2}{*}{$\mathcal{D}$} & \multirow{2}{*}{$[n,k_d,d_d]$} & $\mathbf{q}_i$ & \multirow{2}{*}{$I_\mathcal{D}$} & $q(x)$ \\
            &   & $i \in \{1,\dots,k_d\}$  &   & of degree $n-k_d$  \\
         \hline
    \end{tabular}
    \caption{Summary of notations related to cyclic codes that will be used throughout this work.}
    \label{tab:code_notation}
\end{table}

Suppose that $\mathcal{C} \subset \mathcal{D}$. Then, the generator polynomial $q(x)$ of $I_{\mathcal{D}}$ divides every polynomial in $I_\mathcal{C}$. Therefore, $p(x)=f_1(x)q(x)$ for some polynomial $f_1(x)$ of degree $k_d-k_c$. In addition, $\mathcal{C}^\perp \subseteq \mathcal{C}$ implies that $\tilde{p}_\mathcal{R}(x)=f_2(x)p(x)$ for some polynomial $f_2(x)$ of degree $2k_c-n$. Also, $\mathcal{C} \subset \mathcal{D}$ implies that $\mathcal{D}^\perp \subset \mathcal{C}^\perp$. Consequently, we have that $\tilde{q}_\mathcal{R}(x)=f_3(x)\tilde{p}_\mathcal{R}(x)$ for some polynomial $f_3(x)$ of degree $k_d-k_c$.

\subsection{Stabilizer codes and the CSS construction} \label{subsec:pre_CSS}

An \codepar{n,k,d} \emph{stabilizer code} \cite{Gottesman97} is a quantum code that uses $n$ physical qubits to encode $k$ logical qubits and has code distance $d$. The code can correct up to $\lfloor(d-1)/2\rfloor$ Pauli errors. A stabilizer code is determined by an Abelian group called the \emph{stabilizer group} which is generated by $n-k$ commuting Pauli operators and does not contain scalar multiples of the identity matrix $I^{\otimes n}$. The code space is the simultaneous $+1$ eigenspace of all \emph{stabilizers}, the elements of the stabilizer group.

Next, we consider the transformation of the single-qubit Pauli operators $X_i,Z_i$ ($i=1,\dots,n$) under the encoder $U$ of the stabilizer code. These new operators $\bar{X}_i=UX_iU^\dagger$ and $\bar{Z}_i=UZ_iU^\dagger$ satisfy the normal commutation and anticommutation relations of single-qubit Pauli operators; i.e., $\bar{X}_i$ and $\bar{Z}_i$ anticommute (denoted by $\{\bar{X}_i,\bar{Z}_i\}=0$) for all $i$, and $\bar{X}_i$ and $\bar{Z}_j$ commute (denoted by $[\bar{X}_i,\bar{Z}_j]=0$) for $i\neq j$ ($i,j=1,\dots,n$). These operators operate on the ``virtual'' qubits of the code. In particular, the $\bar{Z}_i$ operators corresponding to the virtual ancilla qubits indexed by $k+1,\dots, n$ generate the stabilizer group,
\begin{equation}
    \mathcal{S}=\left\langle \bar{Z}_{k+1}, \dots, \bar{Z}_n \right\rangle,
\end{equation}
while the corresponding $\bar{X}_i$ operators on the virtual ancilla qubits are often referred to as \emph{destabilizers}. These operators generate the pure errors of the code, which map the code to its orthogonal subspaces.

The centralizer $C(\mathcal{S})$ of the stabilizer group $\mathcal{S}$ are all Pauli operators that commute with the entire stabilizer group, and is generated as follows:
\begin{equation}
    C(\mathcal{S})=\left\langle iI^{\otimes n}, \bar{Z}_1, \dots, \bar{Z}_n, \bar{X}_1, \dots, \bar{X}_k\right\rangle.
\end{equation}
Removing those Pauli operators in the stabilizer group (up to some phase), we get the set $C(\mathcal{S})\setminus\left\langle iI^{\otimes n}, \mathcal{S}\right\rangle$, which are all the logical operators that act nontrivially on the code. The minimum distance of the code is defined as the minimum weight of the elements of this set (where the weight of a Pauli operator is the number of its non-identity tensor factors):
\begin{equation}
    d=\min\wt\left(C(\mathcal{S})\setminus\left\langle iI^{\otimes n}, \mathcal{S}\right\rangle\right).
\end{equation}

If we associate the qubits indexed by $1, \dots, k$ with the input state $\ket{\psi}$, then the set $\{\bar{Z}_1, \dots, \bar{Z}_k, \bar{X}_1, \dots, \bar{X}_k\}$ is a particularly nice choice of generators for the logical operators of the code. The logical group is defined as,
\begin{equation}
    \mathcal{L} = \left\langle iI^{\otimes n}, \bar{Z}_1, \dots, \bar{Z}_k, \bar{X}_1, \dots, \bar{X}_k\right\rangle = C(\mathcal{S})/\mathcal{S}.
\end{equation}

A Calderbank-Shor-Steane (CSS) code \cite{CS96,Steane96b} is a stabilizer code in which its stabilizer generators can be chosen to be purely $X$ or purely $Z$ type. To ease our explanation, we use the following notations throughout this work: Let $\mathbf{v}=(v_0,v_1,\dots,v_{n-1})$ be an $n$-bit binary row vector. We define $X(\mathbf{v})$ to be an $n$-qubit Pauli-$X$ operator of the form $X^{v_0}\otimes\cdots\otimes X^{v_{n-1}}$. We also define an $n$-qubit Pauli-$Z$ operator $Z(\mathbf{v})$ in a similar fashion. Let $\mathbf{u} \cdot \mathbf{v} = u_0 v_0+\dots+u_{n-1} v_{n-1} \;\left(\mathrm{mod}\;2\right)$ be the dot product between two binary vectors $\mathbf{u}$ and $\mathbf{v}$. We have that $X(\mathbf{u})$ and $Z(\mathbf{v})$ anticommute if and only if $\mathbf{u} \cdot \mathbf{v}=0$. Let $\mathcal{C}$ be a classical code with generators $\{\mathbf{v}_i\}$. We let $X(\mathcal{C})$ denote the set $\{X(\mathbf{v})\;|\;\mathbf{v} \in \mathcal{C}\}$ and let $\left\langle X(\mathcal{C})\right\rangle$ denote the group generated by $\{X(\mathbf{v}_i)\}$. $Z(\mathcal{C})$ and $\left\langle Z(\mathcal{C})\right\rangle$ are defined similarly. 

The CSS construction gives a systematic way to construct stabilizer codes from a pair of classical linear codes $\mathcal{C}_x$ and $\mathcal{C}_z$. The construction is provided below.
\begin{theorem} \cite{CS96,Steane96b} \label{thm:CSS}
    Let $\mathcal{C}_x$ and $\mathcal{C}_z$ be classical linear codes with parameters $\left[n,k_x\right]$ and $\left[n,k_z\right]$ respectively. Let $\mathcal{C}_z^\perp\subseteq \mathcal{C}_x$. Then there exists a quantum stabilizer code with stabilizer group $\left\langle X(\mathcal{C}_z^\perp),Z(\mathcal{C}_x^\perp)\right\rangle$ and parameters $\codepar{n,k,d}$, where 
    \begin{align*}
        k & = k_x+k_z-n, \\
        d_x & = \min\wt\left(\mathcal{C}_x\setminus \mathcal{C}_z^\perp\right), \\
        d_z & = \min\wt\left(\mathcal{C}_z\setminus \mathcal{C}_x^\perp\right), \\
        d & =\min\left\{d_x, d_z\right\}.
    \end{align*}
\end{theorem}

In this work, we let $\mathcal{Q}_\mathcal{C}$ denote the CSS code constructed from classical codes $\mathcal{C}_x = \mathcal{C}_z = \mathcal{C}$. The stabilizer generators of $\mathcal{Q}_\mathcal{C}$ correspond to the rows of the parity check matrix of $\mathcal{C}$. If $\mathcal{C}$ is the classical cyclic code described in \cref{subsec:pre_cyclic}, then the stabilizer group of $\mathcal{Q}_\mathcal{C}$ is $\left\langle X(\tilde{\mathbf{p}}_i),Z\left(\tilde{\mathbf{p}}_i\right) \right\rangle$, $i \in \{1,\dots,n-k_c\}$. Throughout this work, we also assume that $\mathcal{C}^\perp$ is a proper subset of $\mathcal{C}$, denoted by $\mathcal{C}^\perp \subset \mathcal{C}$, since we are interested only in the case where $k = 2k_c-n \geq 1$. Similarly, $\mathcal{Q}_\mathcal{D}$ is the CSS code with stabilizer group $\left\langle X(\tilde{\mathbf{q}}_i),Z(\tilde{\mathbf{q}}_i) \right\rangle$, $i \in \{1,\dots,n-k_d\}$.

\subsection{Subsystem codes, hybrid codes, and hybrid subsystem codes} \label{subsec:pre_hybrid}

In this section, we give some background on two generalizations of stabilizer codes\footnote{Here we restrict ourselves to subsystem, hybrid, and hybrid subsystem codes within the stabilizer formalism.}, namely subsystem codes and hybrid quantum-classical codes, as well as their unification, the hybrid subsystem codes.

\subsubsection{Subsystem codes}\label{subsubsec:subsystem_codes}
Subsystem codes \cite{Kribs2005, Poulin2005, Bacon2006}, also known as operator quantum error-correcting codes, are stabilizer codes with unused logical qubits known as gauge qubits. These additional gauge degrees of freedom can be useful during decoding, as errors only need to be corrected up to a gauge operator. Using our description of the virtual operators of stabilizer codes from the previous section, an \codepar{n,k,r,d} subsystem code with $k$ (used) logical qubits and $r$ (unused) gauge qubits is determined by the following groups:
\begin{align}
    \mathcal{L} & = \left\langle iI^{\otimes n}, \bar{Z}_1, \dots, \bar{Z}_k, \bar{X}_1, \dots, \bar{X}_k\right\rangle, \\
    \mathcal{G} & = \left\langle iI^{\otimes n}, \bar{Z}_{k+1}, \dots, \bar{Z}_{n}, \bar{X}_{k+1}, \dots, \bar{X}_{k+r}\right\rangle, \\
    \mathcal{S} & = \left\langle \bar{Z}_{k+r+1}, \dots, \bar{Z}_n \right\rangle,
\end{align}
where the gauge group $\mathcal{G}$ is generated by the stabilizer generators in addition to the anticommuting pairs of gauge operator generators $\bar{Z}_j,\bar{X}_j$ for $k+1\leq j\leq k+r$. As mentioned above, we can view this as moving $r$ pairs of logical operators from $\mathcal{L}$ to our new gauge group $\mathcal{G}$. An alternative way of viewing this is that we have taken $r$ stabilizers and ``unfixed'' them, turning them and their corresponding destabilizers into gauge operators and potentially lowering the minimum distance of the code.

These three groups are connected by,
\begin{equation}
    C(\mathcal{G})=\left\langle\mathcal{L}, \mathcal{S}\right\rangle,
\end{equation}
and in particular we have,
\begin{align}
    \left\langle iI^{\otimes n}, \mathcal{S}\right\rangle & = \mathcal{G}\cap C(\mathcal{G}), \\
    C(\mathcal{S}) & = \left\langle\mathcal{G},C(\mathcal{G})\right\rangle.
\end{align}
The logical operators in $C(\mathcal{G})$ are referred to as the bare logical operators, which act nontrivially on the logical qubits and trivially on the gauge qubits, while the logical operators in $C(\mathcal{S})$ are referred to as the dressed logical operators, which act nontrivially on both the logical and gauge qubits. The minimum distance $d$ of a subsystem code is given by the minimum weight of a nontrivial dressed logical operator, that is
\begin{equation}
    d=\min\wt\left(C(\mathcal{S})\setminus\mathcal{G}\right). \label{eq:dist_subsys}
\end{equation}

A CSS construction for subsystem codes has been studied in \cite{Aly2006, Liu2024}, which we discuss in detail in \cref{sec:hybrid_subsystem_construction}.

\subsubsection{Hybrid codes}\label{subsubsec:hybrid_codes}

Hybrid quantum-classical codes \cite{Kremsky2008, Grassl2017, Nemec2018, Nemec2021}, or more briefly hybrid codes, are those that simultaneously encode both quantum and classical information in a single code. A hybrid code with parameters $\codepar{n,k\!:\!m,d}$ uses $n$ physical qubits to transmit $k$ logical qubits and $m$ logical classical bits simultaneously. One way to view a hybrid code $\mathcal{Q}$ is as a collection of $2^m$ mutually orthogonal quantum codes $Q_i$ of dimension $2^k$ indexed by the classical information. The encoding process involves encoding the quantum message into the quantum code indexed by the classical message.

In terms of the operators on the virtual qubits, a hybrid code can be defined by the following groups:
\begin{align}
    \mathcal{L} & = \left\langle iI^{\otimes n},\bar{Z}_1,\dots,\bar{Z}_k,\bar{X}_1,\dots,\bar{X}_k\right\rangle, \\
    \mathcal{T} & = \left\langle iI^{\otimes n},\bar{X}_{k+1},\dots,\bar{X}_{k+m}\right\rangle, \\
    \mathcal{S}_C & = \left\langle \bar{Z}_{k+1},\dots,\bar{Z}_{k+m}\right\rangle, \\
    \mathcal{S}_Q & = \left\langle \bar{Z}_{k+m+1},\dots,\bar{Z}_n\right\rangle,
\end{align}
where $\mathcal{T}$ is the group of translation operators, which are classical logical operators that translate code words between the constituent quantum codes $Q_i$ of the hybrid code $\mathcal{Q}$, $\mathcal{S}_C$ is the group of classical stabilizers, and $\mathcal{S}_Q$ is the group of quantum stabilizers.

There are multiple ways this hybrid code construction can be interpreted. First, we can understand this as taking pairs of logical operators from $\mathcal{L}$ and dividing them into classical logical operators $\mathcal{T}$ and additional stabilizers $\mathcal{S}_C$. A second way this can be viewed is as picking destabilizers to be classical logical operators, while removing the corresponding stabilizer from the stabilizer group of the code. A third way hybrid codes can be understood is as a gauge-fixing of a subsystem code, where one element in each pair of gauge operators is fixed to become a stabilizer, while its anticommuting partner becomes a translation operator.

When the hybrid code $\mathcal{Q}$ is viewed as a stabilizer code, its stabilizer group corresponds to $\mathcal{S}_Q$, which we also refer to as the outer stabilizer group $\mathcal{S}$. We refer to the group $\mathcal{S}_0=\left\langle\mathcal{S}_C,\mathcal{S}_Q\right\rangle$ as the inner stabilizer group of the code, which stabilizes one of the constituent quantum codes (we can always choose this code to be $Q_0$). Each of the constituent quantum codes can be generated by adding phases to the stabilizer generators in $\mathcal{S}_C$.

The centralizers of the inner and outer codes are given by
\begin{align}
    C(\mathcal{S}) & = \left\langle \mathcal{L}, \mathcal{T}, \mathcal{S}_C, \mathcal{S}_Q\right\rangle, \\
    C(\mathcal{S}_0) & = \left\langle \mathcal{L}, \mathcal{S}_C, \mathcal{S}_Q\right\rangle.
\end{align}
The minimum distance $d$ of the hybrid code is the minimum weight of a nontrivial logical operator (either quantum, classical, or a product of both), that is
\begin{equation}
    d = \min\wt\left(C(\mathcal{S})\setminus\left\langle iI^{\otimes n}, \mathcal{S}_0\right\rangle\right). \label{eq:dist_hybrid}
\end{equation}

In \cref{sec:hybrid_subsystem_construction}, we present a new CSS construction of hybrid codes.

\subsubsection{Hybrid subsystem codes}\label{subsubsec:hybrid_subsystem_codes}

Hybrid subsystem codes were introduced by Dauphinais et al. \cite{Dauphinais2024} to unify subsystem and hybrid codes. A hybrid subsystem code with parameters $\codepar{n,k\!:\!m,r,d}$ uses $n$ physical qubits to transmit $k$ logical qubits and $m$ logical classical bits of information, while having $r$ gauge qubits. As with the previous codes, we can define a hybrid subsystem code by the following groups:
\begin{align}
    \mathcal{L} & = \left\langle iI^{\otimes n},\bar{Z}_1,\dots,\bar{Z}_k,\bar{X}_1,\dots,\bar{X}_k\right\rangle, \\
    \mathcal{T} & = \left\langle iI^{\otimes n},\bar{X}_{k+1},\dots,\bar{X}_{k+m}\right\rangle, \\
    \mathcal{G} & = \left\langle iI^{\otimes n},\bar{Z}_{k+m+1},\dots,\bar{Z}_{n},\bar{X}_{k+m+1},\dots,\bar{X}_{k+m+r}\right\rangle, \\
    \mathcal{S}_C & = \left\langle \bar{Z}_{k+1},\dots,\bar{Z}_{k+m}\right\rangle, \\
    \mathcal{S}_Q & = \left\langle \bar{Z}_{k+m+r+1},\dots,\bar{Z}_{n}\right\rangle,
\end{align}
where similar to the subsystem case in \cref{subsubsec:subsystem_codes}, the gauge group $\mathcal{G}$ is generated by the pairs of gauge operators $\bar{Z}_j,\bar{X}_j$ for $k+m+1\leq j\leq k+m+r$, as well as the stabilizer generators from $\mathcal{S}_Q$.

As with the hybrid codes in \cref{subsubsec:hybrid_codes}, we can define an outer stabilizer group $\mathcal{S}=\mathcal{S}_Q$ and an inner stabilizer group $\mathcal{S}_0=\left\langle\mathcal{S}_C,\mathcal{S}_Q\right\rangle$. We can also define an outer gauge group to be $\mathcal{G}$, and define an inner gauge group $\mathcal{G}_0 = \left\langle \mathcal{G},\mathcal{S}_C\right\rangle$.\footnote{Our choice to denote the outer and inner gauge groups by $\mathcal{G}$ and $\mathcal{G}_0$ respectively is a straightforward extension of the notation used in \cite{Grassl2017, Nemec2022}. In the notation of Dauphinais et al. \cite{Dauphinais2024}, $\mathcal{G}$ refers to what we call the inner gauge group and $\mathcal{G}_0$ refers to the gauge operator generators excluding stabilizer generators.}

Similar to the subsystem codes, the centralizers of the outer and inner gauge groups are 
\begin{align}
    C(\mathcal{G})&=\left\langle\mathcal{L}, \mathcal{T}, \mathcal{S}_C, \mathcal{S}_Q\right\rangle,\\
    C(\mathcal{G}_0)&=\left\langle\mathcal{L}, \mathcal{S}_C, \mathcal{S}_Q\right\rangle,
\end{align}
and from this we can derive the relations
\begin{align}
    \left\langle iI^{\otimes n},\mathcal{S}\right\rangle & = \mathcal{G} \cap C(\mathcal{G}), \\
    C(\mathcal{S}) & = \left\langle\mathcal{G}, C(\mathcal{G})\right\rangle, \\
    \left\langle iI^{\otimes n},\mathcal{S}_0\right\rangle & = \mathcal{G}_0 \cap C(\mathcal{G}_0), \\
    C(\mathcal{S}_0) & = \left\langle\mathcal{G}_0, C(\mathcal{G}_0)\right\rangle.
\end{align}
We follow the terminology from subsystem codes, and refer to elements in $C(\mathcal{G})$ as bare logical operators, which act nontrivially on both the logical qubits and classical bits while acting trivially on the gauge qubits, and elements in $C(\mathcal{S})$ as dressed logical operators, which act nontrivially on the logical qubits, logical classical bits, and the gauge qubits. The minimum distance $d$ of a hybrid subsystem code\footnote{Our minimum distance is a significant simplification of the one given by Dauphinais et al. \cite{Dauphinais2024} (see \cref{eq:dauphinais_min_dist}), as we have restricted ourselves to hybrid subsystem codes in the stabilizer framework. See \cref{subsec:css_hybrid_subsystem_proof} for a proof.} is the minimum weight of a nontrivial dressed logical operator, that is
\begin{equation}
    d = \min\wt\left(C(\mathcal{S})\setminus\mathcal{G}_0\right). \label{eq:dist_hybrid_subsys}
\end{equation}

In \cref{sec:hybrid_subsystem_construction}, we present a new CSS construction of hybrid subsystem codes.

\section{Quantum synchronizable codes in stabilizer formalism} \label{sec:QSC_stabilizer}

In quantum communications, quantum information is often transferred through a noisy channel in the form of blocks of quantum error-correcting codes. It is important that the sender and the receiver mutually understand which qubit is the first qubit of each block, otherwise the intended message could be misread (assuming that there is no qubit insertion or deletion). However, in the case that both parties cannot rely on an external mechanism, block synchronization could be challenging. One solution to this problem is to use quantum synchronizable codes (QSCs) introduced by Fujiwara \cite{Fujiwara2013}. In this section, we briefly review the ideas behind the QSCs, including the encoding procedure and the procedure to correct Pauli and synchronization errors. The main difference between the description in \cite{Fujiwara2013} and our description is that the original work describes QSCs by an evolution of quantum states (i.e., the Schr\"{o}dinger picture), while this work adopts the stabilizer formalism \cite{Gottesman97} and describe QSCs by an evolution of stabilizer generators \cite{Gottesman98} (i.e., the Heisenberg picture). 

First, let us consider a situation where a block of $n$ qubits is sent through a quantum channel, and suppose that $a_{l}$ qubits are sent before the block and $a_{r}$ qubits are sent after the block. A sequence of qubits can be described by $(q_{1},\dots,q_{a_{l}},q_{a_{l}+1},\dots,q_{a_{l}+n},q_{a_{l}+n+1},\dots,q_{a_{l}+n+a_{r}})$. Ideally, the receiver would like to decode the sub-sequence $(q_{a_{l}+1},\dots,q_{a_{l}+n})$, or the \emph{main block}, to obtain the correct encoded information. That is, the first qubit of the main block should be qubit $q_{a_{l}+1}$. However, if the receiver marks the first qubit of the main block too early or too late, the information after decoding could be incorrect. If the first qubit of the main block marked by the receiver is qubit $q_{a_{l}+1-\alpha}$ (or $q_{a_{l}+1+\alpha}$), we say that the receiver marks the first qubit \emph{too early} (or \emph{too late}) by $\alpha$ positions.

By the notation from \cite{Fujiwara2013}, a quantum synchronizable $(a_l,a_r)$-\codepar{n+a_l+a_r,k,d} code encodes $k$ logical qubits into $n+a_l+a_r$ physical qubits and can correct misalignment by up to $a_l$ qubits to the left and up to $a_r$ qubits to the right. Comparing to the situation above and assuming that the leftmost qubit is sent first, the QSC can detect that misalignment occurs and find the original block of qubits intended to be decoded even if the receiver marks the first qubit of the main block too early by up to $a_l$ positions or too late by up to $a_r$ positions. For convenience, we define the \emph{synchronization distance} $d_\textrm{sync}$ to be the possible number of qubits in which, if marked as the first qubit of the main block by the receiver, the synchronization errors (or misalignment) can be corrected. That is, the synchronization distance of an $(a_l,a_r)$-\codepar{n+a_l+a_r,k,d} code is $d_\textrm{sync}=a_l+a_r+1$ (the actual first qubit of the main block is also counted). By this definition, a quantum code that cannot correct any synchronization errors, or a \emph{non-synchronizable code}, has synchronization distance 1.

The following theorem from \cite{Fujiwara2013} provides a construction of QSCs from nested pairs of self-orthogonal classical cyclic codes via the CSS construction.
\begin{theorem} \cite{Fujiwara2013} \label{thm:QSC_original}
    Let $\mathcal{C}$ and $\mathcal{D}$ be $[n,k_c,d_c]$ and $[n,k_d,d_d]$ classical cyclic codes, respectively. Suppose that $\mathcal{C}$ and $\mathcal{D}$ satisfy $\mathcal{C}^\perp \subset \mathcal{C} \subset \mathcal{D}$ and $k_c < k_d$. For any non-negative integers $a_l, a_r$ satisfying $a_l+a_r<k_d-k_c$, there exists an $(a_l,a_r)$-\codepar{n+a_l+a_r,2k_c-n,d_d} quantum synchronizable code that can correct up to $\lfloor(d_d-1)/2\rfloor$ bit-flip errors and up to $\lfloor(d_c-1)/2\rfloor$ phase-flip errors.
\end{theorem}
Note that the maximum synchronization distance of the QSC in \cref{thm:QSC_original} is $k_d-k_c$.

Next, we describe how the sender can encode quantum information into a QSC, and how the receiver can correct the bit-flip, phase-flip, and synchronization errors, as well as decoding the quantum information.

\subsection{Encoding procedure} \label{subsec:QSC_enc}

Step 1. Let $\mathcal{Q}_\mathcal{C}$ be the \codepar{n,2k_c-n,d_c} stabilizer code constructed from a check matrix of the classical cyclic code $\mathcal{C}$ through the CSS construction, and let $U_\mathcal{C}$ be a unitary encoder of $\mathcal{Q}_\mathcal{C}$. The encoding procedure of a QSC starts by applying $U_\mathcal{C}$ to a $(2k_c-n)$-qubit state $\ket{\psi}$ containing the quantum information plus $2n-2k_c$ ancilla qubits. At this stage, the quantum state is described by the stabilizer generators,
\begin{align}
    &X(\tilde{\mathbf{p}}_j), \\
    &Z(\tilde{\mathbf{p}}_j), 
\end{align}
where $i \in \{1,\dots,n-k_c\}$.

Step 2. The operator $X(\mathbf{q}_1)$ is applied to the $n$-qubit code block. Thus, the stabilizer generators are transformed to,
\begin{align}
    &X(\tilde{\mathbf{p}}_j), \\
    &Z(\tilde{\mathbf{p}}_j)(-1)^{\tilde{\mathbf{p}}_j \cdot \mathbf{q}_1}. 
\end{align}

Step 3. $a_l$ ancilla qubits are attached to left of the $n$-qubit code block and $a_r$ ancilla qubits are attached the right, where $a_l+a_r < k_d-k_c$. Here we label the qubits as follows: The left ancilla qubits, the $n$-qubit code block, and the right ancilla qubits are qubits 1 to $a_l$, qubits $a_l+1$ to $a_l+n$, and qubits $a_l+n+1$ to $a_l+n+a_r$, respectively. All ancilla qubits are initialized in the physical $\ket{0}$ state. With the stabilizer generators describing the fixed state of each ancilla qubit, the whole stabilizer group is described by the stabilizer generators, 
\begin{align}
    &X(\underbracket[0.5pt]{\vphantom{\tilde{\mathbf{p}}_j}\,\mathbf{0}\,}_{\mathclap{a_l}}|\underbracket[0.5pt]{\,\tilde{\mathbf{p}}_j\,}_{\mathclap{n}}|\underbracket[0.5pt]{\vphantom{\tilde{\mathbf{p}}_j}\,\mathbf{0}\,}_{\mathclap{a_r}}), \\
    &Z(\underbracket[0.5pt]{\vphantom{\tilde{\mathbf{p}}_j}\,\mathbf{0}\,}_{\mathclap{a_l}}|\underbracket[0.5pt]{\,\tilde{\mathbf{p}}_j\,}_{\mathclap{n}}|\underbracket[0.5pt]{\vphantom{\tilde{\mathbf{p}}_j}\,\mathbf{0}\,}_{\mathclap{a_r}})(-1)^{\tilde{\mathbf{p}}_j \cdot \mathbf{q}_1}, \\
    &Z\left( \begin{array}{c|ccc|c}
    \mathbb{1}_{a_l} & \mathbb{0}_{a_l,a_r} & \mathbb{0}_{a_l,(n-a_l-a_r)} & \mathbb{0}_{a_l,a_l} & \mathbb{0}_{a_l,a_r} \\
    \mathbb{0}_{a_r,a_l} & \mathbb{0}_{a_r,a_r} & \mathbb{0}_{a_r,(n-a_l-a_r)} & \mathbb{0}_{a_r,a_l} & \mathbb{1}_{a_r}
    \end{array}
    \right),
\end{align}
where $|\cdots|$ covers a block of $n$ qubits, $\underbracket[0.5pt]{\,\mathbf{v}\,}_{\mathclap{q}}$ indicates that the vector $\mathbf{v}$ has $q$ bits (which corresponds to $q$ qubits), $\mathbb{1}_{p}$ denotes $p \times p$ identity matrix, $\mathbb{0}_{p,q}$ denotes $p \times q$ zero matrix, and $Z(M)$ of a matrix $M$ describes $Z$-type operators constructed from the rows of $M$.

Step 4. CNOT gates with the following pairs of control and target qubits $(c,t)$ are applied: $(a_l+1,a_l+n+1),\dots,(a_l+a_r,a_l+n+a_r)$ and $(n+a_l,a_l),\dots,(n+1,1)$. With these coupling operations, the stabilizer generators are transformed to,
\begin{align}
    &X(\overbracket[0.5pt]{\underbracket[0.5pt]{\,\tilde{\mathbf{p}}_j\,}_{\mathclap{a_l}}}^{\text{last}}|\underbracket[0.5pt]{\,\tilde{\mathbf{p}}_j\,}_{\mathclap{n}}|\overbracket[0.5pt]{\underbracket[0.5pt]{\,\tilde{\mathbf{p}}_j\,}_{\mathclap{a_r}}}^{\text{first}}), \label{eq:QSC_stab1}\\
    &Z(\underbracket[0.5pt]{\vphantom{\tilde{\mathbf{p}}_j}\,\mathbf{0}\,}_{\mathclap{a_l}}|\underbracket[0.5pt]{\,\tilde{\mathbf{p}}_j\,}_{\mathclap{n}}|\underbracket[0.5pt]{\vphantom{\tilde{\mathbf{p}}_j}\,\mathbf{0}\,}_{\mathclap{a_r}})(-1)^{\tilde{\mathbf{p}}_j \cdot \mathbf{q}_1}, \label{eq:QSC_stab2}\\
    &Z\left( \begin{array}{c|ccc|c}
    \mathbb{1}_{a_l} & \mathbb{0}_{a_l,a_r} & \mathbb{0}_{a_l,(n-a_l-a_r)} & \mathbb{1}_{a_l} & \mathbb{0}_{a_l,a_r} \\
    \mathbb{0}_{a_r,a_l} & \mathbb{1}_{a_r} & \mathbb{0}_{a_r,(n-a_l-a_r)} & \mathbb{0}_{a_r,a_l} & \mathbb{1}_{a_r}
    \end{array}
    \right), \label{eq:QSC_stab3}
\end{align}
where $\overbracket[0.5pt]{\underbracket[0.5pt]{\,\mathbf{v}\,}_{\mathclap{q}}}^{\text{first}}$ (or $\overbracket[0.5pt]{\underbracket[0.5pt]{\,\mathbf{v}\,}_{\mathclap{q}}}^{\text{last}}$) denotes the vector constructed from the first $q$ bits (or the last $q$ bits) of $\mathbf{v}$. The stabilizer generators in \cref{eq:QSC_stab1,eq:QSC_stab2,eq:QSC_stab3} generate the stabilizer group of the $(a_l,a_r)$-\codepar{n+a_l+a_r,2k_c-n} QSC in \cref{thm:QSC_original}.

After these steps, the block of $n+a_l+a_r$ qubits is ready to be transmitted through a noisy quantum channel. The entire encoding procedure is illustrated in \cref{fig:encoder}.

\begin{figure}[htbp]
    \centering
    \resizebox{.48\linewidth}{!}{
    $
		\Qcircuit @C=1em @R=.7em {
		& \qw & \qw & \qw & \qw & \qw & \qw & \qw & \qw & \targ & \qw & \qw \\
            & \cdots &  &  &  &  &  &  & \cdots &  &  & \\
		& \qw & \qw & \qw & \qw & \qw & \qw & \targ & \qw & \qw & \qw & \qw
            \inputgroupv{1}{3}{.8em}{1em}{\ket{0}^{\otimes a_l}\quad} \\           
		& \qw & \multigate{8}{U_\mathcal{C}} & \multigate{8}{X(\mathbf{q}_1)} & \ctrl{9} & \qw & \qw & \qw & \qw & \qw & \qw & \qw \\
            & \cdots & \nghost{U_\mathcal{C}} & \ghost{X(\mathbf{q}_1)} &  & \cdots &  &  & \cdots &  &  & \\
            & \cdots & \nghost{U_\mathcal{C}} & \ghost{X(\mathbf{q}_1)} & \qw & \qw & \ctrl{9} & \qw & \qw & \qw & \qw & \qw \\
		& \qw & \ghost{U_\mathcal{C}} & \ghost{X(\mathbf{q}_1)} & \qw & \qw & \qw & \qw & \qw & \qw & \qw & \qw
            \inputgroupv{4}{7}{0.8em}{2.4em}{\ket{\psi}\quad} \\
		& \qw & \ghost{U_\mathcal{C}} & \ghost{X(\mathbf{q}_1)} &  & \cdots &  &  & \cdots &  &  & \\
            & \cdots & \nghost{U_\mathcal{C}} & \ghost{X(\mathbf{q}_1)} & \qw & \qw & \qw & \qw & \qw & \qw & \qw & \qw \\
            & \cdots & \nghost{U_\mathcal{C}} & \ghost{X(\mathbf{q}_1)} & \qw & \qw & \qw & \qw & \qw & \ctrl{-9} & \qw & \qw \\
            & \cdots & \nghost{U_\mathcal{C}} & \ghost{X(\mathbf{q}_1)} &  & \cdots &  &  & \cdots &  &  & \\
            & \qw & \ghost{U_\mathcal{C}} & \ghost{X(\mathbf{q}_1)} & \qw & \qw & \qw & \ctrl{-9} & \qw & \qw & \qw & \qw 
            \inputgroupv{8}{12}{0.8em}{3.2em}{\ket{0}^{\otimes (2n-2k_c)}\qquad\qquad} \\
		& \qw & \qw & \qw & \targ & \qw & \qw & \qw & \qw & \qw & \qw & \qw \\
		& \cdots &  &  &  & \cdots &  &  &  &  &  & \\
		& \qw & \qw & \qw & \qw & \qw & \targ & \qw & \qw & \qw & \qw & \qw
		\inputgroupv{13}{15}{.8em}{1em}{\ket{0}^{\otimes a_r}\quad}
		} \nonumber
    $}
    \caption{A circuit diagram describing the entire encoding procedure of a QSC.}
    \label{fig:encoder}
\end{figure}
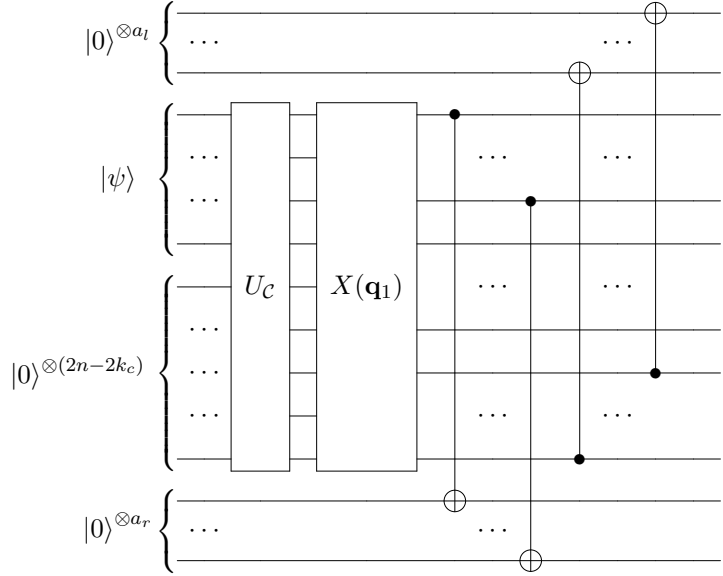

\subsection{Correcting bit-flip, phase-flip, and synchronization errors} \label{subsec:QSC_correction}

We start by considering how the synchronization error can be corrected in the absence of bit-flip and phase-flip errors. First, we present the following corollary.
\begin{corollary} \label{cor:equiv_stab_group_QSC}
Let $a_l,a_r$ be non-negative integers satisfying $a_l+a_r < k_d-k_c$. For all integers $\alpha$ such that $-a_l \leq \alpha \leq a_r$,
\begin{align}
    &X(\overbracket[0.5pt]{\underbracket[0.5pt]{\,\tilde{\mathbf{p}}_j\,}_{\mathclap{a_l+\alpha}}}^{\text{last}}|\underbracket[0.5pt]{\,\tilde{\mathbf{p}}_j\,}_{\mathclap{n}}|\overbracket[0.5pt]{\underbracket[0.5pt]{\,\tilde{\mathbf{p}}_j\,}_{\mathclap{a_r-\alpha}}}^{\text{first}}), \\
    &Z(\underbracket[0.5pt]{\vphantom{\tilde{\mathbf{p}}_j}\,\mathbf{0}\,}_{\mathclap{a_l+\alpha}}|\underbracket[0.5pt]{\,\tilde{\mathbf{p}}_j\,}_{\mathclap{n}}|\underbracket[0.5pt]{\vphantom{\tilde{\mathbf{p}}_j}\,\mathbf{0}\,}_{\mathclap{a_r-\alpha}})(-1)^{\tilde{\mathbf{p}}_j \cdot \mathcal{O}(\mathbf{q}_1,-\alpha)},  \\
    &Z\left( \begin{array}{c|ccc|c}
    \mathbb{1}_{a_l} & \mathbb{0}_{a_l,a_r} & \mathbb{0}_{a_l,(n-a_l-a_r)} & \mathbb{1}_{a_l} & \mathbb{0}_{a_l,a_r} \\
    \mathbb{0}_{a_r,a_l} & \mathbb{1}_{a_r} & \mathbb{0}_{a_r,(n-a_l-a_r)} & \mathbb{0}_{a_r,a_l} & \mathbb{1}_{a_r}
    \end{array}
    \right), 
\end{align}
generate the same stabilizer group, where $j \in \{1,\dots,n-k_c\}$.
\end{corollary} 
\cref{cor:equiv_stab_group_QSC} states that there are many choices of stabilizer generators that describe the stabilizer group of the same QSC, including the choice described by \cref{eq:QSC_stab1,eq:QSC_stab2,eq:QSC_stab3} (which corresponds to $\alpha=0$). \cref{cor:equiv_stab_group_QSC} is obtained from \cref{thm:sync_proof_subspace} (to be formulated later in \cref{sec:SHSC_main}) with $\mathbf{v}=\mathbf{0}$ and $\mathbf{w}=\mathbf{q}_1$. A proof of \cref{thm:sync_proof_subspace} is provided in \cref{subsec:sync_subspace_proof}.

Suppose that the receiver handles a certain subblock of $n$-consecutive qubits in the block of $n+a_l+a_r$ qubits. \cref{cor:equiv_stab_group_QSC} implies that there is some $\alpha$ such that $Z(\tilde{\mathbf{p}}_j)(-1)^{\tilde{\mathbf{p}}_j \cdot \mathcal{O}(\mathbf{q}_1,-\alpha)}$ defining on such $n$ qubits are stabilizer generators of the QSC. That is, if the receiver measures eigenvalues of $Z(\tilde{\mathbf{p}}_j)$ for all $j \in \{1,\dots,n-k_c\}$, they will obtain the values of $(-1)^{\tilde{\mathbf{p}}_j \cdot \mathcal{O}(\mathbf{q}_1,-\alpha)}$. These values can be used to determine $\alpha$, which provided information about misalignment. Since $\tilde{\mathbf{p}}_j$ are rows of the check matrix $H_\mathcal{C}$ of the code $\mathcal{C}$, we can also describe the eigenvalues of $Z(\tilde{\mathbf{p}}_j)$ by the column vector $H_\mathcal{C}\mathcal{O}(\mathbf{q}_1,-\alpha)^T$.

Next, we prove the possible values of $\alpha$ such that $H_\mathcal{C}\mathcal{O}(\mathbf{q}_1,-\alpha)^T$ are distinct.
\begin{lemma} \label{lem:sync_value}
    Suppose that $a_l$ and $a_r$ are non-negative integers satisfying $a_l+a_r < k_d-k_c$. For any integer $\alpha$ such that $-a_l \leq \alpha \leq a_r$, $H_\mathcal{C}\mathcal{O}(\mathbf{q}_1,-\alpha)^T$ is distinct. Or equivalently, for any integers $\beta,\gamma$ such that $\beta > \gamma$, if $\beta-\gamma < k_d-k_c$, then these exists $j \in \{1,\dots,n-k_c\}$ such that $\tilde{\mathbf{p}}_j\cdot \mathcal{O}(\mathbf{q}_1,\beta) \neq \tilde{\mathbf{p}}_j\cdot \mathcal{O}(\mathbf{q}_1,\gamma)$.
\end{lemma}
\begin{proof}
For any integers $\beta,\gamma$ such that $\beta > \gamma$, the following statements are equivalent.
\begin{align}
    &\tilde{\mathbf{p}}_j\cdot \mathcal{O}(\mathbf{q}_1,\gamma) =  \tilde{\mathbf{p}}_j\cdot \mathcal{O}(\mathbf{q}_1,\beta) \;\forall j\\
    \Leftrightarrow\; & \tilde{\mathbf{p}}_j\cdot \left(\mathcal{O}(\mathbf{q}_1,\gamma)+\mathcal{O}(\mathbf{q}_1,\beta)\right) = 0 \;\forall j\\
    \Leftrightarrow\;  & H_\mathcal{C}\left(\mathcal{O}(\mathbf{q}_1,\gamma)+\mathcal{O}(\mathbf{q}_1,\beta)\right)^T = \mathbf{0}^T \\
    \Leftrightarrow\;  & \mathcal{O}(\mathbf{q}_1,\gamma)+\mathcal{O}(\mathbf{q}_1,\beta) \in \mathcal{C} \\
    \Leftrightarrow\;  & \mathbf{q}_1+\mathcal{O}(\mathbf{q}_1,\beta-\gamma) \in \mathcal{C} \\
    \Leftrightarrow\;  & q(x)\left(1+x^{\beta-\gamma}\right) \in I_\mathcal{C}.
\end{align}
Observe that a nontrivial polynomial in $I_\mathcal{C}$ (which corresponds to a nonzero binary vector in $\mathcal{C}$) has degree at least $n-k_c$. Also, $q(x)$ has degree $n-k_d$. If $\beta-\gamma < k_d-k_c$, then $q(x)\left(1+x^{\beta-\gamma}\right)$ is not in $I_\mathcal{C}$. This implies that these exists $j \in \{1,\dots,n-k_c\}$ such that $\tilde{\mathbf{p}}_j\cdot \mathcal{O}(\mathbf{q}_1,\beta) \neq \tilde{\mathbf{p}}_j\cdot \mathcal{O}(\mathbf{q}_1,\gamma)$, or equivalently, $H_\mathcal{C}\mathcal{O}(\mathbf{q}_1,\beta)^T \neq H_\mathcal{C}\mathcal{O}(\mathbf{q}_1,\gamma)^T$. Note that the generators $\mathbf{q}_1,\dots,\mathbf{q}_{k_d-k_c}$ of $\mathcal{D}$ are linearly independent, thus $q(x) \neq x^{\beta-\gamma}q(x)$ when $\beta-\gamma < k_d-k_c$.

Let $a_l+a_r < k_d-k_c$. For any two integers $\alpha,\alpha'$ such that $-a_l \leq \alpha < \alpha' \leq a_r$, we have that $(-\alpha)-(-\alpha') = \alpha'-\alpha < k_d-k_c$ and thus $H_\mathcal{C}\mathcal{O}(\mathbf{q}_1,-\alpha)^T \neq H_\mathcal{C}\mathcal{O}(\mathbf{q}_1,-\alpha')^T$. Therefore, the value of $H_\mathcal{C}\mathcal{O}(\mathbf{q}_1,-\alpha)^T$ for each possible $\alpha$ is distinct.
\end{proof}

\cref{lem:sync_value} implies that a lookup table between $\alpha$ and $H_\mathcal{C}\mathcal{O}(\mathbf{q}_1,-\alpha)^T$ can be made whenever $a_l+a_r < k_d-k_c$.

In the case that there is no $X$-type (bit-flip) error, measuring $Z(\tilde{\mathbf{p}}_j)$ will give $H_\mathcal{C}\mathcal{O}(\mathbf{q}_1,-\alpha)^T$ which can uniquely determines $\alpha$. The sign and the magnitude of $\alpha$ tells us whether the first qubit of the main block is marked too early or too late and by how many positions (where the minus sign means too early and the plus sign means too late). With this information, the receiver can shift their qubit markers back to the original numbering as intended by the sender. Note that if the receiver marks the first qubit too early, they may have to wait for more qubits before the whole block of $n+a_l+a_r$ qubits that can be used to decode the information is fully obtained.

Next, let us consider the case in which $X$-type (bit-flip) and $Z$-type (phase-flip) errors are present. Since $X$-type errors could interfere with the measurement of $Z(\tilde{\mathbf{p}}_j)$, $X$-type error correction need to be done before $\alpha$ can be determined. Observe that for any $\alpha$, $Z(\underbracket[0.5pt]{\vphantom{\tilde{\mathbf{q}}_i}\,\mathbf{0}\,}_{\mathclap{a_l+\alpha}}|\underbracket[0.5pt]{\,\tilde{\mathbf{q}}_i\,}_{\mathclap{n}}|\underbracket[0.5pt]{\vphantom{\tilde{\mathbf{q}}_i}\,\mathbf{0}\,}_{\mathclap{a_r-\alpha}})$ commute with all $X$-type stabilizers for all $i \in \{1,\dots,n-k_d\}$; this can be shown by choosing a proper choice of $\alpha$ in \cref{cor:equiv_stab_group_QSC} and use that fact that $\tilde{\mathbf{q}}_i \cdot \tilde{\mathbf{p}}_j=0$, which is true because $\mathcal{D}^\perp \subset \mathcal{C} = (\mathcal{C}^\perp)^\perp$. 
Let $\mathcal{Q}_\mathcal{D}$ be the \codepar{n,2k_d-n,d_d} stabilizer code constructed from a check matrix of the classical cyclic code $\mathcal{D}$ through the CSS construction. The stabilizer code defined by $Z(\underbracket[0.5pt]{\vphantom{\tilde{\mathbf{q}}_i}\,\mathbf{0}\,}_{\mathclap{a_l+\alpha}}|\underbracket[0.5pt]{\,\tilde{\mathbf{q}}_i\,}_{\mathclap{n}}|\underbracket[0.5pt]{\vphantom{\tilde{\mathbf{q}}_i}\,\mathbf{0}\,}_{\mathclap{a_r-\alpha}})$ on any block of $n$ consecutive qubits is similar to the stabilizer code defined by the $Z$-type generators $Z(\tilde{\mathbf{q}}_i)$ of the code $\mathcal{Q}_\mathcal{D}$. Therefore, the receiver can measure $Z(\tilde{\mathbf{q}}_i)$ on the received $n$ qubits to obtain the error syndrome of an $X$-type error, and $X$-type error correction on this block of $n$ qubits can be done using an error decoder (such as a lookup-table decoder) for the code $\mathcal{Q}_\mathcal{D}$. Afterwards, $\alpha$ can be determined by measuring $Z(\tilde{\mathbf{p}}_j)$ as described earlier in this section. Note that $X$-type error correction is guaranteed only if there are no more than $\lfloor(d_d-1)/2\rfloor$ $X$-type errors on any $n$ consecutive qubits in the block of $n+a_l+a_r$ qubits.

After $\alpha$ is obtained and the correct qubit numbering is retrieved, the receiver can further correct $X$-type and $Z$-type errors using the following procedures: (1) $X$-type error correction on the first $n$ qubits and the last $n$ qubits of the block of $n+a_l+a_r$ qubits can be done by measuring $Z(\underbracket[0.5pt]{\vphantom{\tilde{\mathbf{q}}_i}\,\mathbf{0}\,}_{\mathclap{a_l+a_r}}|\underbracket[0.5pt]{\,\tilde{\mathbf{q}}_i\,}_{\mathclap{n}})$ and $Z(\underbracket[0.5pt]{\,\tilde{\mathbf{q}}_i\,}_{\mathclap{n}}|\underbracket[0.5pt]{\vphantom{\tilde{\mathbf{q}}_i}\,\mathbf{0}\,}_{\mathclap{a_l+a_r}})$ (where $i \in \{1,\dots,n-k_d\}$) and using an error decoder for $\mathcal{Q}_\mathcal{D}$. It is possible to show that $a_l+a_r < k_d-k_c \leq n/2$ using the facts that $2k_c-n \geq 0$ and $n-k_d \geq 0$. Therefore, the length of the entire block is at most $3n/2$ and the procedure is sufficient to correct $X$-type errors in the entire block of $n+a_l+a_r$ qubits (assuming that there are no more than $\lfloor(d_d-1)/2\rfloor$ $X$-type errors on any $n$ consecutive qubits). (2) $Z$-type error correction can be done by first measuring $X(\overbracket[0.5pt]{\underbracket[0.5pt]{\,\tilde{\mathbf{p}}_j\,}_{\mathclap{a_l}}}^{\text{last}}|\underbracket[0.5pt]{\,\tilde{\mathbf{p}}_j\,}_{\mathclap{n}}|\overbracket[0.5pt]{\underbracket[0.5pt]{\,\tilde{\mathbf{p}}_j\,}_{\mathclap{a_r}}}^{\text{first}})$ (where $j \in \{1,\dots,n-k_c\}$) to obtain the error syndrome of a $Z$-type error. Observe that a $Z$-type error on any first $a_l$ qubits (or on any last $a_r$ qubits) is equivalent to a $Z$-type error on the middle block of $n$ qubits up to a multiplication by some stabilizer in \cref{eq:QSC_stab3}. Therefore, the recovery operator for $Z$-type error correction can be obtained using an error decoder for the code $\mathcal{Q}_\mathcal{C}$. Note that $Z$-type error correction is guaranteed only if there are no more than $\lfloor(d_c-1)/2\rfloor$ $Z$-type errors on the entire block of $n+a_l+a_r$ qubits.

The full procedure for the receiver to correct bit-flip, phase-flip, and synchronization errors is summarized below.

\begin{enumerate}
    \item Measure $Z(\tilde{\mathbf{q}}_i)$ (where $i \in \{1,\dots,n-k_d\}$) on the received $n$ qubits and perform $X$-type error correction using an error decoder for the code $\mathcal{Q}_\mathcal{D}$.
    \item Measure $Z(\tilde{\mathbf{p}}_j)$ (where $j \in \{1,\dots,n-k_c\}$) on the received $n$ qubits to obtain the value $H_\mathcal{C}\mathcal{O}(\mathbf{q}_1,-\alpha)^T$.
    \item Determine $\alpha$ from $H_\mathcal{C}\mathcal{O}(\mathbf{q}_1,-\alpha)^T$ and retrieve the original qubit numbering intended by the sender.
    \item Measure $Z(\underbracket[0.5pt]{\vphantom{\tilde{\mathbf{q}}_i}\,\mathbf{0}\,}_{\mathclap{a_l+a_r}}|\underbracket[0.5pt]{\,\tilde{\mathbf{q}}_i\,}_{\mathclap{n}})$ and $Z(\underbracket[0.5pt]{\,\tilde{\mathbf{q}}_i\,}_{\mathclap{n}}|\underbracket[0.5pt]{\vphantom{\tilde{\mathbf{q}}_i}\,\mathbf{0}\,}_{\mathclap{a_l+a_r}})$ (where $i \in \{1,\dots,n-k_d\}$) and correct $X$-type errors on the entire block of $n+a_l+a_r$ qubits using an error decoder for the code $\mathcal{Q}_\mathcal{D}$.
    \item Measure $X(\overbracket[0.5pt]{\underbracket[0.5pt]{\,\tilde{\mathbf{p}}_j\,}_{\mathclap{a_l}}}^{\text{last}}|\underbracket[0.5pt]{\,\tilde{\mathbf{p}}_j\,}_{\mathclap{n}}|\overbracket[0.5pt]{\underbracket[0.5pt]{\,\tilde{\mathbf{p}}_j\,}_{\mathclap{a_r}}}^{\text{first}})$ (where $j \in \{1,\dots,n-k_c\}$) and correct $Z$-type errors on the entire block of $n+a_l+a_r$ qubits using an error decoder for the code $\mathcal{Q}_\mathcal{C}$.
\end{enumerate}

After the error correction procedure, the quantum state is an error-free code word of the QSC. To obtain the encoded quantum information, one can simply apply the inverse of the operations described in Steps 4 to 1 of \cref{subsec:QSC_enc}.

The descriptions of the encoding and the error correction procedures in \cref{subsec:QSC_enc,subsec:QSC_correction} provide an alternative proof of \cref{thm:QSC_original} in the stabilizer formalism. However, there is a subtle difference between the procedures in this work and the ones in the original work by Fujiwara \cite{Fujiwara2013}. In that work, misalignment is determined using polynomial division operations on $n$ qubits, which can be implemented through quantum shift registers \cite{GB03} (see also \cite{Wilde09} for an alternative implementation method). Such operations require $n$ steps with $n$ ancilla qubits. In this work, we find that misalignment (or $\alpha$) can be determined by measuring eigenvalues of operators $Z(\tilde{\mathbf{p}}_j)$ where $j \in \{1,\dots,n-k_c\}$ on the received qubits, which require $n-k_c$ steps with $n-k_c$ ancilla qubits. Later in \cref{sec:SHSC_main}, readers will find that the measurement results of $Z(\tilde{\mathbf{p}}_j)$ and $Z(\tilde{\mathbf{q}}_i)$ (for $X$-type error correction before synchronization recovery) are not independent, and thus $\alpha$ can be determined by measuring only $Z(\tilde{\mathbf{p}}_j)$ where $j \in \{1,\dots,k_d-k_c\}$. That is, our scheme for synchronization recovery requires only $k_d-k_c$ steps with $k_d-k_c$ ancilla qubits. Furthermore, readers will also see that QSCs similar to the ones proposed in \cite{Fujiwara2013} can also encode up to $k_d-k_c$ bits of classical information. 

\section{From generators of classical cyclic codes to anticommuting pairs of Pauli operators} \label{sec:cyclic_to_Pauli}

The main goal of this work is to generalize the notion of QSCs and construct other quantum codes with block synchronization property. In this section, we describe how anticommuting pairs of Pauli operators can be constructed from the generators of classical cyclic codes $\mathcal{C}$ and $\mathcal{D}$ satisfying $\mathcal{C}^\perp \subset \mathcal{C} \subset \mathcal{D}$. These anticommuting pairs are useful when we describe the code construction later in \cref{sec:SHSC_main}.

We start by observing that there are many possible choices of generators for each of the classical cyclic codes $\mathcal{C},\mathcal{C}^\perp,\mathcal{D},\mathcal{D}^\perp$. This fact provides flexibility to represent each classical code and is useful for the proofs of subsequent theorems.

\begin{theorem} \label{thm:gen_decomp}
Let $\mathcal{C}$ and $\mathcal{D}$ be $[n,k_c,d_c]$ and $[n,k_d,d_d]$ classical cyclic codes satisfying $\mathcal{C}^\perp \subset \mathcal{C} \subset \mathcal{D}$. Let $p(x)$ and $q(x)$ be the generator polynomials of $\mathcal{C}$ and $\mathcal{D}$, and let $\tilde{p}_\mathcal{R}(x)$, $\tilde{q}_\mathcal{R}(x)$, $\tilde{\mathbf{q}}_i$, $\tilde{\mathbf{p}}_i$, $\mathbf{p}_i$, and $\mathbf{q}_i$, be defined as in \cref{subsec:pre_cyclic}. Then the codes $\mathcal{D}^\perp$, $\mathcal{C}^\perp$, $\mathcal{C}$, and $\mathcal{D}$ can be represented in the following ways:
\begin{enumerate}[leftmargin=1cm]
    \item 
    \begin{equation}
        \mathcal{D}^\perp = \langle \tilde{\mathbf{q}}_i \rangle,\; \mathcal{C}^\perp = \langle \tilde{\mathbf{p}}_j \rangle,\; \mathcal{C} = \langle \mathbf{p}_l \rangle,\; \mathcal{D} = \langle \mathbf{q}_m \rangle, 
    \end{equation} where $i \in \{1,\dots,n-k_d\}$, $j \in \{1,\dots,n-k_c\}$, $l \in \{1,\dots,k_c\}$, and $m \in \{1,\dots,k_d\}$.
    \item 
    \begin{equation}
        \mathcal{D}^\perp = \langle \tilde{\mathbf{q}}_i \rangle,\; \mathcal{C}^\perp = \langle \tilde{\mathbf{q}}_i,\tilde{\mathbf{p}}_j \rangle,\; \mathcal{C} = \langle \tilde{\mathbf{q}}_i,\tilde{\mathbf{p}}_j,\mathbf{p}_l \rangle,\; \mathcal{D} = \langle \tilde{\mathbf{q}}_i,\tilde{\mathbf{p}}_j,\mathbf{p}_l,\mathbf{q}_m \rangle, 
    \end{equation} where $i \in \{1,\dots,n-k_d\}$, $j \in \{1,\dots,k_d-k_c\}$, $l \in \{1,\dots,2k_c-n\}$, and $m \in \{1,\dots,k_d-k_c\}$.
\end{enumerate}
\end{theorem}
A proof of \cref{thm:gen_decomp} is provided in \cref{subsec:gen_decomp_proof}.

Next, we observe that there exists a choice of classical code generators with good properties from which the anticommuting pairs of Pauli operators can be constructed.
\begin{theorem} \label{thm:op_pairing}
Let $\mathcal{C}$ and $\mathcal{D}$ be classical cyclic codes as defined in \cref{thm:gen_decomp}, and let $i,i' \in \{1,\dots,n-k_d\}$, $j,j' \in \{1,\dots,k_d-k_c\}$, $l,l' \in \{1,\dots,2k_c-n\}$, $m,m' \in \{1,\dots,k_d-k_c\}$. There exist binary vectors $\tilde{\mathbf{t}}_j,\mathbf{s}_l^x,\mathbf{s}_l^z,\mathbf{t}_m^x,\mathbf{t}_m^z$ with the following properties:
\begin{enumerate}[leftmargin=1cm]
    \item $\langle \mathbf{s}_l^x \rangle = \langle \mathbf{s}_l^z \rangle = \langle \mathbf{p}_l \rangle$, $\langle \tilde{\mathbf{t}}_j \rangle = \langle \tilde{\mathbf{p}}_j \rangle$, and $\langle\tilde{\mathbf{t}}_j,\mathbf{t}_m^x\rangle = \langle\tilde{\mathbf{t}}_j,\mathbf{t}_m^z\rangle = \langle\tilde{\mathbf{p}}_j,\mathbf{q}'_m\rangle$, where $\mathbf{q}'_m=\mathbf{q}_m+\sum_l (\mathbf{q}_m \cdot \mathbf{s}_l^x)\mathbf{s}_{l}^z = \mathbf{q}_m+\sum_l (\mathbf{q}_m \cdot \mathbf{s}_l^z)\mathbf{s}_{l}^x$.
    \item $\mathcal{C}^\perp = \langle \tilde{\mathbf{q}}_i,\tilde{\mathbf{t}}_j \rangle$, $\mathcal{C} = \langle \tilde{\mathbf{q}}_i,\tilde{\mathbf{t}}_j,\mathbf{s}_l^x \rangle = \langle \tilde{\mathbf{q}}_i,\tilde{\mathbf{t}}_j,\mathbf{s}_l^z \rangle$, and $\mathcal{D} = \langle \tilde{\mathbf{q}}_i,\tilde{\mathbf{t}}_j,\mathbf{s}_l^x,\mathbf{t}_m^x \rangle = \langle \tilde{\mathbf{q}}_i,\tilde{\mathbf{t}}_j,\mathbf{s}_l^z,\mathbf{t}_m^z \rangle$.
    \item $\tilde{\mathbf{q}}_i \cdot \tilde{\mathbf{q}}_{i'} = \tilde{\mathbf{q}}_i \cdot \tilde{\mathbf{t}}_{j} = \tilde{\mathbf{q}}_i \cdot \mathbf{s}_{l}^x = \tilde{\mathbf{q}}_i \cdot \mathbf{s}_{l}^z = \tilde{\mathbf{q}}_i \cdot \mathbf{t}_{m}^x = \tilde{\mathbf{q}}_i \cdot \mathbf{t}_{m}^z =0$ for all $i,i',j,l,m$.
    \item $\tilde{\mathbf{t}}_j \cdot \tilde{\mathbf{t}}_{j'} = \tilde{\mathbf{t}}_j \cdot \mathbf{s}_{l}^x=\tilde{\mathbf{t}}_j \cdot \mathbf{s}_{l}^z=0$ for all $j,j',l$.
    \item For each $l$, $\mathbf{s}_l^x \cdot \mathbf{s}_{l}^z=1$, and $\mathbf{s}_l^x \cdot \mathbf{s}_{l'}^z=\mathbf{s}_l^z \cdot \mathbf{s}_{l'}^x=0$ for all $l'\neq l$.
    \item For each $j$, $\tilde{\mathbf{t}}_{j} \cdot \mathbf{t}_{j}^x=\tilde{\mathbf{t}}_{j} \cdot \mathbf{t}_{j}^z=1$ and $\tilde{\mathbf{t}}_{j} \cdot \mathbf{t}_{j'}^x=\tilde{\mathbf{t}}_{j} \cdot \mathbf{t}_{j'}^z=0$ for all $j'\neq j$.
    \item $\mathbf{s}_l^x \cdot \mathbf{t}_{j}^z = \mathbf{s}_l^z \cdot \mathbf{t}_{j}^x =0$ for all $j,l$.
    \item $\mathbf{t}_{j}^x \cdot \mathbf{t}_{j'}^z=0$ for all $j,j'$.
\end{enumerate}
\end{theorem}
A proof of \cref{thm:op_pairing} is provided in \cref{subsec:op_pairing_proof}.

Observe that the Pauli operators $X(\mathbf{v}_1)$ and $Z\!\left(\mathbf{v}_2\right)$ commute if $\mathbf{v}_1 \cdot \mathbf{v}_2=0$ and they anticommute if $\mathbf{v}_1 \cdot \mathbf{v}_2=1$. From \cref{thm:op_pairing}, the following pairs of anticommuting Pauli operators can be constructed from the generators $\tilde{\mathbf{t}}_j,\mathbf{s}_l^x,\mathbf{s}_l^z,\mathbf{t}_m^x,\mathbf{t}_m^z$ of the classical codes: $\{X(\mathbf{s}_l^x),Z(\mathbf{s}_l^z)\}$, $\{X(\tilde{\mathbf{t}}_j),Z(\mathbf{t}_j^z)\}$, and $\{Z(\tilde{\mathbf{t}}_j),X(\mathbf{t}_j^x)\}$. 

Let $\mathcal{N}_0 = \{X(\tilde{\mathbf{q}}_i),Z(\tilde{\mathbf{q}}_i),X(\tilde{\mathbf{t}}_j),Z(\tilde{\mathbf{t}}_j),X(\mathbf{s}_l^x),Z(\mathbf{s}_l^z),X(\mathbf{t}_{m}^x),Z(\mathbf{t}_{m}^z)\}$. The Pauli operators in $\mathcal{N}_0$ have the following properties: 
\begin{enumerate}
    \item For each $i$, $X(\tilde{\mathbf{q}}_i)$ and $Z(\tilde{\mathbf{q}}_i)$ commute with all operators in $\mathcal{N}_0$.
    \item For each $l$, $X(\mathbf{s}_l^x)$ and $Z(\mathbf{s}_l^z)$ anticommute, and both operators commute with other operators in $\mathcal{N}_0$.
    \item For each $j$, $X(\tilde{\mathbf{t}}_j)$ and $Z(\mathbf{t}_{j}^z)$ (or $Z(\tilde{\mathbf{t}}_j)$ and $X(\mathbf{t}_{j}^x)$) anticommute, and both operators commute with other operators in $\mathcal{N}_0$.
\end{enumerate}
Consider the stabilizer code $\mathcal{Q}_\mathcal{D}$ which is constructed from a check matrix of the classical cyclic code $\mathcal{D}$ through the CSS construction. We can see that the stabilizer group $\mathcal{S}$ of $\mathcal{Q}_\mathcal{D}$ is generated by $\{X(\tilde{\mathbf{q}}_i),Z(\tilde{\mathbf{q}}_i)\}$, while the centralizer $C(\mathcal{S})$ is generated by $\mathcal{N}_0$ and $iI^{\otimes{n}}$.

In the next section, we will use the operators in $\mathcal{N}_0$ to define the stabilizer group, gauge group, and logical operators of each quantum code in the synchronizable hybrid subsystem code family.

\section{A family of synchronizable hybrid subsystem codes} \label{sec:SHSC_main}

Previously in \cref{sec:QSC_stabilizer}, we describe how the sender can encode quantum information into a block of QSC, and how the receiver can correct bit-flip, phase-flip, and synchronization errors then decode the quantum information. In this section, we generalize the ideas and provide constructions of hybrid codes, subsystem codes, and hybrid subsystem codes on which synchronization recovery can be done. The quantum codes provided in this section are unified into \emph{a family of synchronizable hybrid subsystem codes}.

We start by stating the main theorem of this work. Recall that we use the notation of $\codepar{n,k\!:\!m,r,d}$ hybrid subsystem code to refer to a code that encodes $k$ logical qubits and $m$ classical bits into $n$ physical qubits, has $r$ gauge qubits, and has code distance $d$. An $\codepar{n,k,r,d}$ subsystem code and an $\codepar{n,k\!:\!m,d}$ hybrid code are defined similarly. A synchronizable code with parameter $(a_l,a_r)$ can correct misalignment by up to $a_l$ qubits to the left and up to $a_r$ qubits to the right, or equivalently, the code can correct misalignment in the case that the first qubit of the main block is marked too early by up to $a_l$ positions or too late by up to $a_r$ positions (assuming that the leftmost qubit is sent first).

\begin{theorem} \label{thm:QSC_unified}
    Let $\mathcal{C}$ and $\mathcal{D}$ be $[n,k_c,d_c]$ and $[n,k_d,d_d]$ classical cyclic codes, respectively. Suppose that $\mathcal{C}$ and $\mathcal{D}$ satisfy $\mathcal{C}^\perp \subset \mathcal{C} \subset \mathcal{D}$ and $k_c < k_d$.
    \begin{enumerate}[leftmargin=1cm] 
        \item There exists an \codepar{n,2k_c-n,2(k_d-k_c),d_d} \emph{subsystem code} $\mathcal{Q}_1$ that has maximum synchronization distance $1$.
        \item For any non-negative integers $a_l, a_r$ satisfying $a_l+a_r<k_d-k_c$, there exists an $(a_l,a_r)$-\codepar{n+a_l+a_r,2k_c-n,k_d-k_c,d_d} \emph{synchronizable subsystem code} $\mathcal{Q}_2$ that has maximum synchronization distance $k_d-k_c$.
        \item For any non-negative integers $a_l, a_r$ satisfying $a_l+a_r<k_d-k_c$, there exists an $(a_l,a_r)$-\codepar{n+a_l+a_r,2k_c-n\!:\!k_d-k_c,d_d} \emph{synchronizable hybrid code} $\mathcal{Q}_3$ that has maximum synchronization distance $k_d-k_c$.
        \item For any integer $y \in \{1,\dots,k_d-k_c-2\}$ and for any non-negative integers $a_l, a_r$ satisfying $a_l+a_r<k_d-k_c-y$, there exists an $(a_l,a_r)$-\codepar{n+a_l+a_r,2k_c-n\!:\!k_d-k_c+y,d_d} \emph{synchronizable hybrid code} $\mathcal{Q}_4$ that has maximum synchronization distance $k_d-k_c-y$.
        \item There exists an \codepar{n,2k_c-n\!:\!2(k_d-k_c),d_d} \emph{hybrid code} $\mathcal{Q}_5$ that has maximum synchronization distance $1$.
        \item For any integer $y \in \{1,\dots,k_d-k_c-2\}$ and for any non-negative integers $a_l, a_r$ satisfying $a_l+a_r<k_d-k_c-y$, there exists an $(a_l,a_r)$-\codepar{n+a_l+a_r,2k_c-n\!:\!y,k_d-k_c,d_d} \emph{synchronizable hybrid subsystem code} $\mathcal{Q}_6$ that has maximum synchronization distance $k_d-k_c-y$.
        \item There exists an \codepar{n,2k_c-n\!:\!k_d-k_c,k_d-k_c,d_d} \emph{hybrid subsystem code} $\mathcal{Q}_7$ that has maximum synchronization distance $1$.
    \end{enumerate}
\end{theorem}

To encode any desired code $\mathcal{Q}_i$ from \cref{thm:QSC_unified}, we use a general encoding procedure described below. 
\begin{enumerate}
    \item The encoding procedure starts by applying a unitary encoder $U_i$ of some CSS code $\mathcal{Q}^0_i$, which will be called the \emph{initial code} of the desired code $\mathcal{Q}_i$, to $2k_c-n$ qubits containing quantum information plus $2n-2k_c$ ancilla qubits.
    \item If $\mathcal{Q}_i$ has hybrid structure (i.e., classical bits can be also encoded), the operators corresponding to the classical information are applied to the block of $n$ qubits.
    \item If $\mathcal{Q}_i$ is a synchronizable code (a code in which the maximum synchronization distance is $> 1$), the operator $X(\mathbf{q}_1)$ is applied to the block of $n$ qubits. Afterwards, $a_l$ ancilla qubits are attached to left of the $n$-qubit code block and $a_r$ ancilla qubits are attached the right (where $a_l+a_r$ satisfies the requirement for $\mathcal{Q}_i$). All ancilla qubits are initialized in the physical $\ket{0}$ state. After that, CNOT gates with the following pairs of control and target qubits $(c,t)$ are applied: $(a_l+1,a_l+n+1),\dots,(a_l+a_r,a_l+n+a_r)$ and $(n+a_l,a_l),\dots,(n+1,1)$.
\end{enumerate}
After these steps are complete, a code word of $\mathcal{Q}_i$ is obtained and ready to be transmitted through a quantum channel to the receiver.

Relationship between the properties of the codes from \cref{thm:QSC_unified} can be visualized by a Venn diagram in \cref{fig:Venn}. 
The constructions of these codes are related by three key processes: the encoding procedure (as described above), gauge fixing \cite{PR13,ADP14}, and the sacrifice of synchronization distance (to be explained later in this section). How each code can be constructed from one another is summarized in \cref{fig:con_diagram}.

\begin{figure}[htbp]
    \centering
    \includegraphics[width=0.32\textwidth]{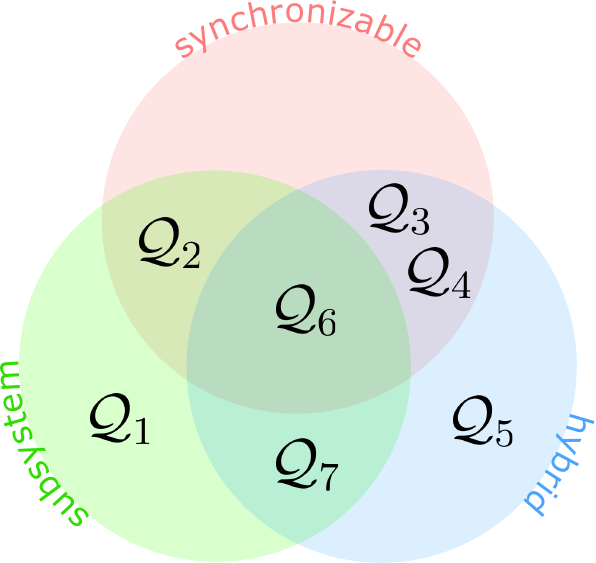}
    \caption{A Venn diagram displaying properties of the codes in a family of the synchronizable hybrid subsystem codes from \cref{thm:QSC_unified}.}
    \label{fig:Venn}
\end{figure}

\begin{figure}[htbp]
    \centering
    \hspace{0.1\textwidth}
    \includegraphics[width=0.55\textwidth]{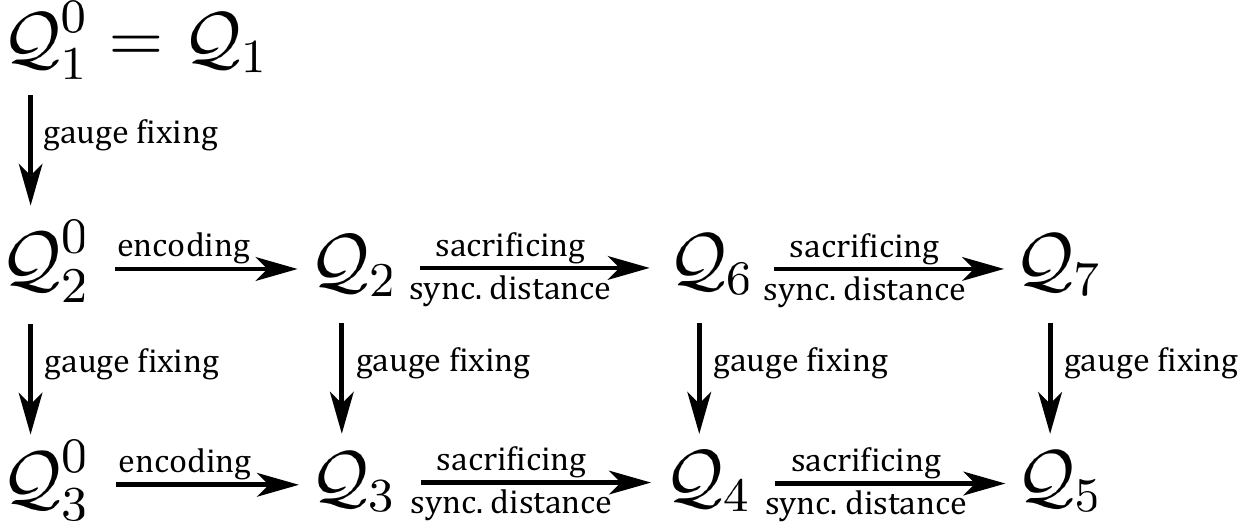}
    \caption{Relationship between each of the codes from \cref{thm:QSC_unified} and their initial codes. The constructions of these codes are related by three key processes: the encoding procedure involving ancilla attachment and CNOT operations, gauge fixing, and the sacrifice of synchronization distance.}
    \label{fig:con_diagram}
\end{figure}

The details of the initial code as well as the encoding and error correction procedures for each code from \cref{thm:QSC_unified} are provided below.

\subsection{Subsystem and synchronizable subsystem codes} \label{subsec:SHSC_subsystem}

\vspace*{0.2cm}
\textit{$\mathcal{Q}_1$: a subsystem code}
\vspace*{0.2cm}

\noindent The initial code $\mathcal{Q}^0_1$ of $\mathcal{Q}_1$ is described by the stabilizer group $\mathcal{S}=\left\langle X(\tilde{\mathbf{q}}_i),Z(\tilde{\mathbf{q}}_i)\right\rangle$ and the gauge group $\mathcal{G}=\left\langle iI^{\otimes n},X(\tilde{\mathbf{q}}_i),Z(\tilde{\mathbf{q}}_i),X(\tilde{\mathbf{t}}_j),Z(\tilde{\mathbf{t}}_j),X(\mathbf{t}_{m}^x),Z(\mathbf{t}_{m}^z)\right\rangle$ where $i \in \{1,\dots,n-k_d\}$, $j \in \{1,\dots,k_d-k_c\}$, and $m \in \{1,\dots,k_d-k_c\}$. Since $\mathcal{Q}_1$ is a non-synchronizable code and does not have hybrid structure, $\mathcal{Q}_1$ and $\mathcal{Q}^0_1$ are the same code. $\mathcal{Q}_1$ encodes no classical bits and has maximum synchronization distance 1.

The logical qubits of $\mathcal{Q}_1$ are defined by the anticommuting pairs $\{X(\mathbf{s}_l^x),Z(\mathbf{s}_l^z)\}$ where $l \in \{1,\dots,2k_c-n\}$, while the gauge qubits are defined by the pairs $\{X(\tilde{\mathbf{t}}_j),Z(\mathbf{t}_j^z)\}$, and $\{Z(\tilde{\mathbf{t}}_j),X(\mathbf{t}_j^x)\}$. Therefore, $\mathcal{Q}_1$ has $2k_c-n$ logical qubits and $2(k_d-k_c)$ gauge qubits. 

Error correction on $\mathcal{Q}_1$ can be done by measuring $X(\tilde{\mathbf{q}}_i)$ and $Z(\tilde{\mathbf{q}}_i)$, and using an error decoder for $\mathcal{Q}_\mathcal{D}$, the stabilizer code constructed from a check matrix of the classical code $\mathcal{D}$ through the CSS construction. 
The code distance of $\mathcal{Q}_1$ is the same as the code distance of $\mathcal{Q}_\mathcal{D}$, which is $d_d$.

\vspace*{0.2cm}
\noindent \textit{$\mathcal{Q}_2$: a synchronizable subsystem code}
\vspace*{0.2cm}

\noindent The initial code $\mathcal{Q}^0_2$ of $\mathcal{Q}_2$ is described by the stabilizer group $\mathcal{S}=\left\langle X(\tilde{\mathbf{q}}_i),Z(\tilde{\mathbf{q}}_i),Z(\tilde{\mathbf{t}}_j)\right\rangle$ and the gauge group $\mathcal{G}=\left\langle iI^{\otimes n},X(\tilde{\mathbf{q}}_i),Z(\tilde{\mathbf{q}}_i),X(\tilde{\mathbf{t}}_j),Z(\tilde{\mathbf{t}}_j),Z(\mathbf{t}_{m}^z)\right\rangle$ where $i \in \{1,\dots,n-k_d\}$, $j \in \{1,\dots,k_d-k_c\}$, and $m \in \{1,\dots,k_d-k_c\}$. By \cref{thm:op_pairing}, we can alternatively write the stabilizer and the gauge groups of $\mathcal{Q}^0_2$ as $\mathcal{S}=\left\langle X(\tilde{\mathbf{q}}_i),Z(\tilde{\mathbf{q}}_i),Z(\tilde{\mathbf{p}}_j)\right\rangle$ and $\mathcal{G}=\left\langle iI^{\otimes n},X(\tilde{\mathbf{q}}_i),Z(\tilde{\mathbf{q}}_i),X(\tilde{\mathbf{p}}_j),Z(\tilde{\mathbf{p}}_j),Z(\mathbf{q}'_m)\right\rangle$. $\mathcal{Q}^0_2$ can be obtained from $\mathcal{Q}^0_1$ through gauge fixing by measuring the operators $Z(\tilde{\mathbf{t}}_j)$ for all $j$. 

After applying the operator $X(\mathbf{q}_1)$, attaching $a_l+a_r$ ancilla qubits, and applying CNOT gates, we obtain the desired code $\mathcal{Q}_2$. The stabilizer group of $\mathcal{Q}_2$ can be described by the stabilizer generators,
\begin{align}
    &X(\overbracket[0.5pt]{\underbracket[0.5pt]{\,\tilde{\mathbf{q}}_i\,}_{\mathclap{a_l}}}^{\text{last}}|\underbracket[0.5pt]{\,\tilde{\mathbf{q}}_i\,}_{\mathclap{n}}|\overbracket[0.5pt]{\underbracket[0.5pt]{\,\tilde{\mathbf{q}}_i\,}_{\mathclap{a_r}}}^{\text{first}}), \\
    &Z(\underbracket[0.5pt]{\vphantom{\tilde{\mathbf{q}}_i}\,\mathbf{0}\,}_{\mathclap{a_l}}|\underbracket[0.5pt]{\,\tilde{\mathbf{q}}_i\,}_{\mathclap{n}}|\underbracket[0.5pt]{\vphantom{\tilde{\mathbf{q}}_i}\,\mathbf{0}\,}_{\mathclap{a_r}}),\\
    &Z(\underbracket[0.5pt]{\vphantom{\tilde{\mathbf{p}}_j}\,\mathbf{0}\,}_{\mathclap{a_l}}|\underbracket[0.5pt]{\,\tilde{\mathbf{p}}_j\,}_{\mathclap{n}}|\underbracket[0.5pt]{\vphantom{\tilde{\mathbf{p}}_j}\,\mathbf{0}\,}_{\mathclap{a_r}})(-1)^{\tilde{\mathbf{p}}_j \cdot \mathbf{q}_1},  \\
    &Z\left( \begin{array}{c|ccc|c}
    \mathbb{1}_{a_l} & \mathbb{0}_{a_l,a_r} & \mathbb{0}_{a_l,(n-a_l-a_r)} & \mathbb{1}_{a_l} & \mathbb{0}_{a_l,a_r} \\
    \mathbb{0}_{a_r,a_l} & \mathbb{1}_{a_r} & \mathbb{0}_{a_r,(n-a_l-a_r)} & \mathbb{0}_{a_r,a_l} & \mathbb{1}_{a_r}
    \end{array}
    \right), 
\end{align}
and the gauge group of $\mathcal{Q}_2$ can be described by the stabilizer generators, $iI^{\otimes n+a_l+a_r}$, and
\begin{align}
    &X(\overbracket[0.5pt]{\underbracket[0.5pt]{\,\tilde{\mathbf{p}}_j\,}_{\mathclap{a_l}}}^{\text{last}}|\underbracket[0.5pt]{\,\tilde{\mathbf{p}}_j\,}_{\mathclap{n}}|\overbracket[0.5pt]{\underbracket[0.5pt]{\,\tilde{\mathbf{p}}_j\,}_{\mathclap{a_r}}}^{\text{first}}), \\
    &Z(\underbracket[0.5pt]{\vphantom{\mathbf{q}'_m}\,\mathbf{0}\,}_{\mathclap{a_l}}|\underbracket[0.5pt]{\,\mathbf{q}'_m\,}_{\mathclap{n}}|\underbracket[0.5pt]{\vphantom{\mathbf{q}'_m}\,\mathbf{0}\,}_{\mathclap{a_r}})(-1)^{\mathbf{q}'_m \cdot \mathbf{q}_1}. 
\end{align}

The logical qubits of $\mathcal{Q}^0_2$ are defined by the anticommuting pairs $\{X(\mathbf{s}_l^x),Z(\mathbf{s}_l^z)\}$ where $l \in \{1,\dots,2k_c-n\}$, and the gauge qubits are defined by the pairs $\{X(\tilde{\mathbf{t}}_j),Z(\mathbf{t}_j^z)\}$ where $j \in \{1,\dots,k_d-k_c\}$. Transforming these operators by the operations in the encoding procedure, we can obtain the anticommuting pairs that define logical and gauge qubits of the code $\mathcal{Q}_2$. Thus, $\mathcal{Q}_2$ has $2k_c-n$ logical qubits and $k_d-k_c$ gauge qubits.

Next, we describe a procedure to correct Pauli and synchronization errors on $\mathcal{Q}_2$. We first state a theorem which is useful for analyzing synchronization recovery on subsystem and hybrid subsystem codes.

\begin{theorem} \label{thm:sync_proof_subsystem}
Let $\mathcal{D}^\perp=\langle \tilde{\mathbf{q}}_i \rangle$, $\mathcal{C}^\perp=\langle \tilde{\mathbf{p}}_j \rangle$, $\mathcal{C}=\langle \tilde{\mathbf{p}}_j,\mathbf{s}_l^x \rangle$, $\mathcal{D}=\langle \tilde{\mathbf{p}}_j,\mathbf{s}_l^x,\mathbf{q}'_m \rangle$, where $i \in \{1,\dots, n-k_d\}$, $j \in \{1,\dots, n-k_c\}$, $l \in \{1,\dots, 2k_c-n\}$, $m \in \{1,\dots, k_d-k_c\}$. Also, let $\mathbf{w} \in \mathbb{Z}_2^n$ be some binary vector. For all integers $\alpha$ such that $-a_l \leq \alpha \leq a_r$,
\begin{align}
    &X(\overbracket[0.5pt]{\underbracket[0.5pt]{\,\tilde{\mathbf{q}}_i\,}_{\mathclap{a_l+\alpha}}}^{\text{last}}|\underbracket[0.5pt]{\,\tilde{\mathbf{q}}_i\,}_{\mathclap{n}}|\overbracket[0.5pt]{\underbracket[0.5pt]{\,\tilde{\mathbf{q}}_i\,}_{\mathclap{a_r-\alpha}}}^{\text{first}}), \\
    &Z(\underbracket[0.5pt]{\vphantom{\tilde{\mathbf{p}}_j}\,\mathbf{0}\,}_{\mathclap{a_l+\alpha}}|\underbracket[0.5pt]{\,\tilde{\mathbf{p}}_j\,}_{\mathclap{n}}|\underbracket[0.5pt]{\vphantom{\tilde{\mathbf{p}}_j}\,\mathbf{0}\,}_{\mathclap{a_r-\alpha}})(-1)^{\tilde{\mathbf{p}}_j\cdot \mathcal{O}\left(\mathbf{w},-\alpha\right)},  \\
    &Z\left( \begin{array}{c|ccc|c}
    \mathbb{1}_{a_l} & \mathbb{0}_{a_l,a_r} & \mathbb{0}_{a_l,(n-a_l-a_r)} & \mathbb{1}_{a_l} & \mathbb{0}_{a_l,a_r} \\
    \mathbb{0}_{a_r,a_l} & \mathbb{1}_{a_r} & \mathbb{0}_{a_r,(n-a_l-a_r)} & \mathbb{0}_{a_r,a_l} & \mathbb{1}_{a_r}
    \end{array}
    \right),
\end{align}
generate the same stabilizer group, and
\begin{align}
    &X(\overbracket[0.5pt]{\underbracket[0.5pt]{\,\tilde{\mathbf{p}}_j\,}_{\mathclap{a_l+\alpha}}}^{\text{last}}|\underbracket[0.5pt]{\,\tilde{\mathbf{p}}_j\,}_{\mathclap{n}}|\overbracket[0.5pt]{\underbracket[0.5pt]{\,\tilde{\mathbf{p}}_j\,}_{\mathclap{a_r-\alpha}}}^{\text{first}}), \\
    &Z(\underbracket[0.5pt]{\vphantom{\tilde{\mathbf{p}}_j}\,\mathbf{0}\,}_{\mathclap{a_l+\alpha}}|\underbracket[0.5pt]{\,\tilde{\mathbf{p}}_j\,}_{\mathclap{n}}|\underbracket[0.5pt]{\vphantom{\tilde{\mathbf{p}}_j}\,\mathbf{0}\,}_{\mathclap{a_r-\alpha}})(-1)^{\tilde{\mathbf{p}}_j \cdot \mathcal{O}\left(\mathbf{w},-\alpha\right)},  \\
    &Z(\underbracket[0.5pt]{\vphantom{\mathbf{q}'_m}\,\mathbf{0}\,}_{\mathclap{a_l+\alpha}}|\underbracket[0.5pt]{\,\mathbf{q}'_m\,}_{\mathclap{n}}|\underbracket[0.5pt]{\vphantom{\mathbf{q}'_m}\,\mathbf{0}\,}_{\mathclap{a_r-\alpha}})(-1)^{\mathbf{q}'_m \cdot \mathcal{O}\left(\mathbf{w},-\alpha\right)},  \\
    &Z\left( \begin{array}{c|ccc|c}
    \mathbb{1}_{a_l} & \mathbb{0}_{a_l,a_r} & \mathbb{0}_{a_l,(n-a_l-a_r)} & \mathbb{1}_{a_l} & \mathbb{0}_{a_l,a_r} \\
    \mathbb{0}_{a_r,a_l} & \mathbb{1}_{a_r} & \mathbb{0}_{a_r,(n-a_l-a_r)} & \mathbb{0}_{a_r,a_l} & \mathbb{1}_{a_r}
    \end{array}
    \right),
\end{align}
and $iI^{\otimes n+a_l+a_r}$ generate the same gauge group.
\end{theorem}
A proof of \cref{thm:sync_proof_subsystem} is provided in \cref{subsec:sync_subsystem_proof}.

Using \cref{thm:sync_proof_subsystem} with $\mathbf{w}=\mathbf{q}_1$, the fact that $\mathcal{C}^\perp=\langle\tilde{\mathbf{q}}_i,\tilde{\mathbf{p}}_j\rangle$ where $i \in \{1,\dots,n-k_d\}$ and $j \in \{1,\dots,k_d-k_c\}$, and the fact that $\tilde{\mathbf{q}}_i \cdot \mathcal{O}\left(\mathbf{q}_1,-\alpha\right)=0$ for all $i,\alpha$, the stabilizer group of $\mathcal{Q}_2$ can also be described by the stabilizer generators,
\begin{align}
    &X(\overbracket[0.5pt]{\underbracket[0.5pt]{\,\tilde{\mathbf{q}}_i\,}_{\mathclap{a_l+\alpha}}}^{\text{last}}|\underbracket[0.5pt]{\,\tilde{\mathbf{q}}_i\,}_{\mathclap{n}}|\overbracket[0.5pt]{\underbracket[0.5pt]{\,\tilde{\mathbf{q}}_i\,}_{\mathclap{a_r-\alpha}}}^{\text{first}}), \\
    &Z(\underbracket[0.5pt]{\vphantom{\tilde{\mathbf{q}}_i}\,\mathbf{0}\,}_{\mathclap{a_l+\alpha}}|\underbracket[0.5pt]{\,\tilde{\mathbf{q}}_i\,}_{\mathclap{n}}|\underbracket[0.5pt]{\vphantom{\tilde{\mathbf{q}}_i}\,\mathbf{0}\,}_{\mathclap{a_r-\alpha}}),\\   
    &Z(\underbracket[0.5pt]{\vphantom{\tilde{\mathbf{p}}_j}\,\mathbf{0}\,}_{\mathclap{a_l+\alpha}}|\underbracket[0.5pt]{\,\tilde{\mathbf{p}}_j\,}_{\mathclap{n}}|\underbracket[0.5pt]{\vphantom{\tilde{\mathbf{p}}_j}\,\mathbf{0}\,}_{\mathclap{a_r-\alpha}})(-1)^{\tilde{\mathbf{p}}_j \cdot \mathcal{O}(\mathbf{q}_1,-\alpha)},  \\
    &Z\left( \begin{array}{c|ccc|c}
    \mathbb{1}_{a_l} & \mathbb{0}_{a_l,a_r} & \mathbb{0}_{a_l,(n-a_l-a_r)} & \mathbb{1}_{a_l} & \mathbb{0}_{a_l,a_r} \\
    \mathbb{0}_{a_r,a_l} & \mathbb{1}_{a_r} & \mathbb{0}_{a_r,(n-a_l-a_r)} & \mathbb{0}_{a_r,a_l} & \mathbb{1}_{a_r}
    \end{array}
    \right), 
\end{align}
and the gauge group of $\mathcal{Q}_2$ can also be described by the stabilizer generators, $iI^{\otimes n+a_l+a_r}$, and
\begin{align}
    &X(\overbracket[0.5pt]{\underbracket[0.5pt]{\,\tilde{\mathbf{p}}_j\,}_{\mathclap{a_l+\alpha}}}^{\text{last}}|\underbracket[0.5pt]{\,\tilde{\mathbf{p}}_j\,}_{\mathclap{n}}|\overbracket[0.5pt]{\underbracket[0.5pt]{\,\tilde{\mathbf{p}}_j\,}_{\mathclap{a_r-\alpha}}}^{\text{first}}), \\
    &Z(\underbracket[0.5pt]{\vphantom{\mathbf{q}'_m}\,\mathbf{0}\,}_{\mathclap{a_l+\alpha}}|\underbracket[0.5pt]{\,\mathbf{q}'_m\,}_{\mathclap{n}}|\underbracket[0.5pt]{\vphantom{\mathbf{q}'_m}\,\mathbf{0}\,}_{\mathclap{a_r-\alpha}})(-1)^{\mathbf{q}'_m \cdot \mathcal{O}(\mathbf{q}_1,-\alpha)}, 
\end{align}
for any integer $\alpha$ such that $-a_l \leq \alpha \leq a_r$. 

$X$-type error correction has to be done before $\alpha$ can be determined. Observe that the stabilizer generators $Z(\underbracket[0.5pt]{\vphantom{\tilde{\mathbf{q}}_i}\,\mathbf{0}\,}_{\mathclap{a_l+\alpha}}|\underbracket[0.5pt]{\,\tilde{\mathbf{q}}_i\,}_{\mathclap{n}}|\underbracket[0.5pt]{\vphantom{\tilde{\mathbf{q}}_i}\,\mathbf{0}\,}_{\mathclap{a_r-\alpha}})$ of $\mathcal{Q}_2$ define a stabilizer code on any block of $n$ consecutive qubits which is similar to the stabilizer code defined by the $Z$-type generators $Z(\tilde{\mathbf{q}}_i)$ of the code $\mathcal{Q}_\mathcal{D}$. Hence, $X$-type error correction on any block of $n$ consecutive qubits can be done by measuring $Z(\tilde{\mathbf{q}}_i)$ and using an error decoder for the code $\mathcal{Q}_\mathcal{D}$.

By \cref{lem:sync_value}, we know that if $a_l+a_r < k_d-k_c$, for any integer $\alpha$ such that $-a_l \leq \alpha \leq a_r$, $H_\mathcal{C}\mathcal{O}(\mathbf{q}_1,-\alpha)^T$ is distinct. Therefore, we can determine $\alpha$ by measuring $Z(\tilde{\mathbf{p}}_j)$ on the received $n$ qubits. Note that $j \in \{1,\dots,k_d-k_c\}$, so the measurements require only $k_d-k_c$ steps. After $\alpha$ is determined and the synchronization recovery is done, one can further correct $X$-type and $Z$-type errors on the entire block of $n+a_l+a_r$ qubits by measuring $Z(\underbracket[0.5pt]{\vphantom{\tilde{\mathbf{q}}_i}\,\mathbf{0}\,}_{\mathclap{a_l+a_r}}|\underbracket[0.5pt]{\,\tilde{\mathbf{q}}_i\,}_{\mathclap{n}})$, $Z(\underbracket[0.5pt]{\,\tilde{\mathbf{q}}_i\,}_{\mathclap{n}}|\underbracket[0.5pt]{\vphantom{\tilde{\mathbf{q}}_i}\,\mathbf{0}\,}_{\mathclap{a_l+a_r}})$, and $X(\overbracket[0.5pt]{\underbracket[0.5pt]{\,\tilde{\mathbf{q}}_i\,}_{\mathclap{a_l}}}^{\text{last}}|\underbracket[0.5pt]{\,\tilde{\mathbf{q}}_i\,}_{\mathclap{n}}|\overbracket[0.5pt]{\underbracket[0.5pt]{\,\tilde{\mathbf{q}}_i\,}_{\mathclap{a_r}}}^{\text{first}})$.

The full procedure to correct Pauli and synchronization errors on $\mathcal{Q}_2$ is summarized below.
\begin{enumerate}
    \item Measure $Z(\tilde{\mathbf{q}}_i)$ where $i \in \{1,\dots,n-k_d\}$ on the received $n$ qubits and perform $X$-type error correction using an error decoder for the code $\mathcal{Q}_\mathcal{D}$.
    \item Measure $Z(\tilde{\mathbf{p}}_j)$ where $j \in \{1,\dots,k_d-k_c\}$ on the received $n$ qubits to obtain the value $H_\mathcal{C}\mathcal{O}(\mathbf{q}_1,-\alpha)^T$.
    \item Determine $\alpha$ from $H_\mathcal{C}\mathcal{O}(\mathbf{q}_1,-\alpha)^T$ and perform synchronization recovery.
    \item Measure $Z(\underbracket[0.5pt]{\vphantom{\tilde{\mathbf{q}}_i}\,\mathbf{0}\,}_{\mathclap{a_l+a_r}}|\underbracket[0.5pt]{\,\tilde{\mathbf{q}}_i\,}_{\mathclap{n}})$ and $Z(\underbracket[0.5pt]{\,\tilde{\mathbf{q}}_i\,}_{\mathclap{n}}|\underbracket[0.5pt]{\vphantom{\tilde{\mathbf{q}}_i}\,\mathbf{0}\,}_{\mathclap{a_l+a_r}})$ where $i \in \{1,\dots,n-k_d\}$ and correct $X$-type errors on the entire block of $n+a_l+a_r$ qubits using an error decoder for the code $\mathcal{Q}_\mathcal{D}$.
    \item Measure $X(\overbracket[0.5pt]{\underbracket[0.5pt]{\,\tilde{\mathbf{q}}_i\,}_{\mathclap{a_l}}}^{\text{last}}|\underbracket[0.5pt]{\,\tilde{\mathbf{q}}_i\,}_{\mathclap{n}}|\overbracket[0.5pt]{\underbracket[0.5pt]{\,\tilde{\mathbf{q}}_i\,}_{\mathclap{a_r}}}^{\text{first}})$ where $i \in \{1,\dots,n-k_d\}$ and correct $Z$-type errors on the entire block of $n+a_l+a_r$ qubits using an error decoder for the code $\mathcal{Q}_\mathcal{D}$.
\end{enumerate}

$\mathcal{Q}_2$ has maximum synchronization distance $k_d-k_c$. Error correction on $\mathcal{Q}_2$ is guaranteed only if there are no more than $\lfloor(d_d-1)/2\rfloor$ $X$-type errors on any $n$ consecutive qubits and there are no more than $\lfloor(d_d-1)/2\rfloor$ $Z$-type errors on the entire block of $n+a_l+a_r$ qubits. Therefore, $\mathcal{Q}_2$ has code distance $d_d$.

\subsection{Hybrid and synchronizable hybrid codes} \label{subsec:SHSC_hybrid}

\vspace*{0.2cm}
\noindent \textit{$\mathcal{Q}_3$: a synchronizable hybrid code}
\vspace*{0.2cm}

\noindent The initial code $\mathcal{Q}^0_3$ of $\mathcal{Q}_3$ is described by the stabilizer group $\mathcal{S}=\left\langle X(\tilde{\mathbf{q}}_i),Z(\tilde{\mathbf{q}}_i),X(\tilde{\mathbf{t}}_j),Z(\tilde{\mathbf{t}}_j)\right\rangle$ where $i \in \{1,\dots,n-k_d\}$ and $j \in \{1,\dots,k_d-k_c\}$. Alternatively, we can write the stabilizer group as $\mathcal{S}=\left\langle X(\tilde{\mathbf{q}}_i),Z(\tilde{\mathbf{q}}_i),X(\tilde{\mathbf{p}}_j),Z(\tilde{\mathbf{p}}_j)\right\rangle$ by using \cref{thm:op_pairing}. This code is a subspace code (a subsystem code with zero gauge qubits), so the gauge group and the stabilizer group are the same up to $\pm 1,\pm i$ phases. $\mathcal{Q}^0_3$ can be obtained from $\mathcal{Q}^0_2$ through gauge fixing by measuring the operators $X(\tilde{\mathbf{t}}_j)$ for all $j$. 

$\mathcal{Q}_3$ is a hybrid code, and a translation operator (or a classical logical operator) of the form $Z(\sum_{m=1}^{k_d-k_c} b_m\mathbf{q}_m)$ for some $(b_1,\dots,b_{k_d-k_c})\in \mathbb{Z}_2^{k_d-k_c}$ will be applied in the second step of the general encoding procedure. This operator encodes $k_d-k_c$ classical bits of information into phases of $X$-type stabilizer generators. $\mathcal{Q}_3$ is also a synchronizable code, so $X(\mathbf{q}_1)$ will be applied in the third step of the encoding procedure.

The desired code $\mathcal{Q}_3$ can be obtained after completing the encoding procedure. The inner code of $\mathcal{Q}_3$ labeled by $(b_1,\dots,b_{k_d-k_c})$ can be described by the stabilizer generators,
\begin{align}
    &X(\overbracket[0.5pt]{\underbracket[0.5pt]{\,\tilde{\mathbf{q}}_i\,}_{\mathclap{a_l}}}^{\text{last}}|\underbracket[0.5pt]{\,\tilde{\mathbf{q}}_i\,}_{\mathclap{n}}|\overbracket[0.5pt]{\underbracket[0.5pt]{\,\tilde{\mathbf{q}}_i\,}_{\mathclap{a_r}}}^{\text{first}}), \\
    &X(\overbracket[0.5pt]{\underbracket[0.5pt]{\,\tilde{\mathbf{p}}_j\,}_{\mathclap{a_l}}}^{\text{last}}|\underbracket[0.5pt]{\,\tilde{\mathbf{p}}_j\,}_{\mathclap{n}}|\overbracket[0.5pt]{\underbracket[0.5pt]{\,\tilde{\mathbf{p}}_j\,}_{\mathclap{a_r}}}^{\text{first}})(-1)^{\tilde{\mathbf{p}}_j\cdot \sum_{m=1}^{k_d-k_c} b_m\mathbf{q}_m}, \\
    &Z(\underbracket[0.5pt]{\vphantom{\tilde{\mathbf{q}}_i}\,\mathbf{0}\,}_{\mathclap{a_l}}|\underbracket[0.5pt]{\,\tilde{\mathbf{q}}_i\,}_{\mathclap{n}}|\underbracket[0.5pt]{\vphantom{\tilde{\mathbf{q}}_i}\,\mathbf{0}\,}_{\mathclap{a_r}}),  \\
    &Z(\underbracket[0.5pt]{\vphantom{\tilde{\mathbf{p}}_j}\,\mathbf{0}\,}_{\mathclap{a_l}}|\underbracket[0.5pt]{\,\tilde{\mathbf{p}}_j\,}_{\mathclap{n}}|\underbracket[0.5pt]{\vphantom{\tilde{\mathbf{p}}_j}\,\mathbf{0}\,}_{\mathclap{a_r}})(-1)^{\tilde{\mathbf{p}}_j \cdot \mathbf{q}_1},  \\
    &Z\left( \begin{array}{c|ccc|c}
    \mathbb{1}_{a_l} & \mathbb{0}_{a_l,a_r} & \mathbb{0}_{a_l,(n-a_l-a_r)} & \mathbb{1}_{a_l} & \mathbb{0}_{a_l,a_r} \\
    \mathbb{0}_{a_r,a_l} & \mathbb{1}_{a_r} & \mathbb{0}_{a_r,(n-a_l-a_r)} & \mathbb{0}_{a_r,a_l} & \mathbb{1}_{a_r}
    \end{array}
    \right).
\end{align}
In fact, the QSC introduced by Fujiwara in \cite{Fujiwara2013} (which is reviewed in \cref{sec:QSC_stabilizer}) is the code $\mathcal{Q}_3$ in which classical bits are not initially encoded.

The logical qubits of $\mathcal{Q}^0_3$ are defined by the anticommuting pairs $\{X(\mathbf{s}_l^x),Z(\mathbf{s}_l^z)\}$ where $l \in \{1,\dots,2k_c-n\}$, and we can obtain logical $X$ and logical $Z$ operators of the code $\mathcal{Q}_3$ by transforming such pairs by the operations in the encoding procedure. Hence, $\mathcal{Q}_3$ has $2k_c-n$ logical qubits. 

Next, we describe a procedure to correct Pauli and synchronization errors on $\mathcal{Q}_3$. Here we state a theorem which is useful for analyzing synchronization recovery on subspace hybrid codes.

\begin{theorem} \label{thm:sync_proof_subspace}
Let $\mathcal{C}^\perp=\langle \tilde{\mathbf{p}}_j \rangle$, $j \in \{1,\dots, n-k_c\}$ and let $\mathbf{v},\mathbf{w} \in \mathbb{Z}_2^n$ be some binary vectors. For all integers $\alpha$ such that $-a_l \leq \alpha \leq a_r$,
\begin{align}
    &X(\overbracket[0.5pt]{\underbracket[0.5pt]{\,\tilde{\mathbf{p}}_j\,}_{\mathclap{a_l+\alpha}}}^{\text{last}}|\underbracket[0.5pt]{\,\tilde{\mathbf{p}}_j\,}_{\mathclap{n}}|\overbracket[0.5pt]{\underbracket[0.5pt]{\,\tilde{\mathbf{p}}_j\,}_{\mathclap{a_r-\alpha}}}^{\text{first}})(-1)^{\tilde{\mathbf{p}}_j\cdot \mathcal{O}\left(\mathbf{v},-\alpha\right)}, \\
    &Z(\underbracket[0.5pt]{\vphantom{\tilde{\mathbf{p}}_j}\,\mathbf{0}\,}_{\mathclap{a_l+\alpha}}|\underbracket[0.5pt]{\,\tilde{\mathbf{p}}_j\,}_{\mathclap{n}}|\underbracket[0.5pt]{\vphantom{\tilde{\mathbf{p}}_j}\,\mathbf{0}\,}_{\mathclap{a_r-\alpha}})(-1)^{\tilde{\mathbf{p}}_j\cdot \mathcal{O}\left(\mathbf{w},-\alpha\right)},   \\
    &Z\left( \begin{array}{c|ccc|c}
    \mathbb{1}_{a_l} & \mathbb{0}_{a_l,a_r} & \mathbb{0}_{a_l,(n-a_l-a_r)} & \mathbb{1}_{a_l} & \mathbb{0}_{a_l,a_r} \\
    \mathbb{0}_{a_r,a_l} & \mathbb{1}_{a_r} & \mathbb{0}_{a_r,(n-a_l-a_r)} & \mathbb{0}_{a_r,a_l} & \mathbb{1}_{a_r}
    \end{array}
    \right),
\end{align}
generate the same stabilizer group.
\end{theorem}
A proof of \cref{thm:sync_proof_subspace} is provided in \cref{subsec:sync_subspace_proof}.

Using \cref{thm:sync_proof_subspace} with $\mathbf{v}=\sum_{m=1}^{k_d-k_c} b_m\mathbf{q}_m$, $\mathbf{w}=\mathbf{q}_1$, the fact that $\mathcal{C}^\perp=\langle\tilde{\mathbf{q}}_i,\tilde{\mathbf{p}}_j\rangle$ where $i \in \{1,\dots,n-k_d\}$ and $j \in \{1,\dots,k_d-k_c\}$, and the fact that $\tilde{\mathbf{q}}_i \cdot \mathcal{O}\left(\mathbf{q}_m,-\alpha\right)=0$ for all $i,m,\alpha$, the stabilizer group of the inner code labeled by $(b_1,\dots,b_{k_d-k_c})$ can also be described by the stabilizer generators,
\begin{align}
    &X(\overbracket[0.5pt]{\underbracket[0.5pt]{\,\tilde{\mathbf{q}}_i\,}_{\mathclap{a_l+\alpha}}}^{\text{last}}|\underbracket[0.5pt]{\,\tilde{\mathbf{q}}_i\,}_{\mathclap{n}}|\overbracket[0.5pt]{\underbracket[0.5pt]{\,\tilde{\mathbf{q}}_i\,}_{\mathclap{a_r-\alpha}}}^{\text{first}}), \\
    &X(\overbracket[0.5pt]{\underbracket[0.5pt]{\,\tilde{\mathbf{p}}_j\,}_{\mathclap{a_l+\alpha}}}^{\text{last}}|\underbracket[0.5pt]{\,\tilde{\mathbf{p}}_j\,}_{\mathclap{n}}|\overbracket[0.5pt]{\underbracket[0.5pt]{\,\tilde{\mathbf{p}}_j\,}_{\mathclap{a_r-\alpha}}}^{\text{first}})(-1)^{\tilde{\mathbf{p}}_j\cdot \mathcal{O}\left(\sum_{m=1}^{k_d-k_c} b_m\mathbf{q}_m,-\alpha\right)}, \\
    &Z(\underbracket[0.5pt]{\vphantom{\tilde{\mathbf{q}}_i}\,\mathbf{0}\,}_{\mathclap{a_l+\alpha}}|\underbracket[0.5pt]{\,\tilde{\mathbf{q}}_i\,}_{\mathclap{n}}|\underbracket[0.5pt]{\vphantom{\tilde{\mathbf{q}}_i}\,\mathbf{0}\,}_{\mathclap{a_r-\alpha}}),  \\
    &Z(\underbracket[0.5pt]{\vphantom{\tilde{\mathbf{p}}_j}\,\mathbf{0}\,}_{\mathclap{a_l+\alpha}}|\underbracket[0.5pt]{\,\tilde{\mathbf{p}}_j\,}_{\mathclap{n}}|\underbracket[0.5pt]{\vphantom{\tilde{\mathbf{p}}_j}\,\mathbf{0}\,}_{\mathclap{a_r-\alpha}})(-1)^{\tilde{\mathbf{p}}_j \cdot \mathcal{O}(\mathbf{q}_1,-\alpha)},  \\
    &Z\left( \begin{array}{c|ccc|c}
    \mathbb{1}_{a_l} & \mathbb{0}_{a_l,a_r} & \mathbb{0}_{a_l,(n-a_l-a_r)} & \mathbb{1}_{a_l} & \mathbb{0}_{a_l,a_r} \\
    \mathbb{0}_{a_r,a_l} & \mathbb{1}_{a_r} & \mathbb{0}_{a_r,(n-a_l-a_r)} & \mathbb{0}_{a_r,a_l} & \mathbb{1}_{a_r},
    \end{array}
    \right), 
\end{align}
for any integer $\alpha$ such that $-a_l \leq \alpha \leq a_r$. Here we can see that $Z(\tilde{\mathbf{p}}_j)$ on any block of $n$ consecutive qubits provides information about $\alpha$ in terms of $H_\mathcal{C}\mathcal{O}(\mathbf{q}_1,-\alpha)^T$. By \cref{lem:sync_value}, we know that if $a_l+a_r < k_d-k_c$, for any integer $\alpha$ such that $-a_l \leq \alpha \leq a_r$, $H_\mathcal{C}\mathcal{O}(\mathbf{q}_1,-\alpha)^T$ is distinct. Therefore, $\alpha$ can be determined by the value of $H_\mathcal{C}\mathcal{O}(\mathbf{q}_1,-\alpha)^T$.

To correct Pauli and synchronization errors on $\mathcal{Q}_3$, the receiver starts by measuring $Z(\tilde{\mathbf{q}}_i)$ where $i \in \{1,\dots,n-k_d\}$ on a block of $n$ received qubits and using an error decoder for the code $\mathcal{Q}_\mathcal{D}$ to correct $X$-type errors. Next, $Z(\tilde{\mathbf{p}}_j)$ where $j \in \{1,\dots,k_d-k_c\}$ is measured and $\alpha$ can be found from $H_\mathcal{C}\mathcal{O}(\mathbf{q}_1,-\alpha)^T$. We point out the measurements of $Z(\tilde{\mathbf{p}}_{j'})$ where $j' \in \{k_d-k_c+1,\dots,n-k_c\}$, as in the error correction procedure in \cref{subsec:QSC_correction}, is not necessary since $\tilde{\mathbf{p}}_{j'}$ are not linearly independent from $\tilde{\mathbf{q}}_i$ and $\tilde{\mathbf{p}}_j$ (see the proof of \cref{thm:gen_decomp} in \cref{subsec:gen_decomp_proof}). 

After the synchronization recovery is done, $X$-type and $Z$-type error corrections on the entire block of $n+a_l+a_r$ qubits can be done by measuring $Z(\underbracket[0.5pt]{\vphantom{\tilde{\mathbf{q}}_i}\,\mathbf{0}\,}_{\mathclap{a_l+a_r}}|\underbracket[0.5pt]{\,\tilde{\mathbf{q}}_i\,}_{\mathclap{n}})$, $Z(\underbracket[0.5pt]{\,\tilde{\mathbf{q}}_i\,}_{\mathclap{n}}|\underbracket[0.5pt]{\vphantom{\tilde{\mathbf{q}}_i}\,\mathbf{0}\,}_{\mathclap{a_l+a_r}})$, and $X(\overbracket[0.5pt]{\underbracket[0.5pt]{\,\tilde{\mathbf{q}}_i\,}_{\mathclap{a_l}}}^{\text{last}}|\underbracket[0.5pt]{\,\tilde{\mathbf{q}}_i\,}_{\mathclap{n}}|\overbracket[0.5pt]{\underbracket[0.5pt]{\,\tilde{\mathbf{q}}_i\,}_{\mathclap{a_r}}}^{\text{first}})$, then using an error decoder for the code $\mathcal{Q}_\mathcal{D}$. Error correction is guaranteed only if there are no more than $\lfloor(d_d-1)/2\rfloor$ $X$-type errors on any $n$ consecutive qubits and there are no more than $\lfloor(d_d-1)/2\rfloor$ $Z$-type errors on the entire block of $n+a_l+a_r$ qubits. Note that in exchange for the encoding of classical information, the number of $Z$-type errors that $\mathcal{Q}_3$ can correct is less than or equal to the one that the QSC in \cref{sec:QSC_stabilizer} can correct (since $d_d \leq d_c$).

The full procedure to correct Pauli and synchronization errors on $\mathcal{Q}_3$ is the same as that of $\mathcal{Q}_2$.

After both Pauli and synchronization errors are corrected, one can obtain the encoded classical information by measuring $X(\overbracket[0.5pt]{\underbracket[0.5pt]{\,\tilde{\mathbf{p}}_j\,}_{\mathclap{a_l}}}^{\text{last}}|\underbracket[0.5pt]{\,\tilde{\mathbf{p}}_j\,}_{\mathclap{n}}|\overbracket[0.5pt]{\underbracket[0.5pt]{\,\tilde{\mathbf{p}}_j\,}_{\mathclap{a_r}}}^{\text{first}})$, which gives $H_\mathcal{C}\left(\sum_{m=1}^{k_d-k_c} b_m\mathbf{q}_m\right)^T$ as their eigenvalues. Using the lemma below, we can show that $H_\mathcal{C}\left(\sum_{m=1}^{k_d-k_c} b_m\mathbf{q}_m\right)^T$ is distinct for each $\mathbf{b}=(b_1,...,b_{k_d-k_c})$, and thus the encoded classical information $\mathbf{b}$ can be determined by the value of $H_\mathcal{C}\left(\sum_{m=1}^{k_d-k_c} b_m\mathbf{q}_m\right)^T$.

\begin{lemma} \label{lem:message_value}
    For each $\mathbf{c}=(c_1,...,c_{k_d-k_c}) \in \mathbb{Z}^{k_d-k_c}_2$, $H_\mathcal{C}\left(\sum_{m=1}^{k_d-k_c} c_m\mathbf{q}_m\right)^T$ is distinct.
\end{lemma}
\begin{proof}
Let $\mathbf{v}=(v_1,...,v_{k_d-k_c})$ and $\mathbf{w}=(w_1,...,w_{k_d-k_c})$ be some binary vectors in $\mathbb{Z}^{k_d-k_c}_2$. The following statements are equivalent.
\begin{align}
    &\tilde{\mathbf{p}}_j\cdot \sum_{m=1}^{k_d-k_c} v_m\mathbf{q}_m =  \tilde{\mathbf{p}}_j\cdot \sum_{m=1}^{k_d-k_c} w_m\mathbf{q}_m \;\forall j\\
    \Leftrightarrow\; & \tilde{\mathbf{p}}_j\cdot \sum_{m=1}^{k_d-k_c} (v_m+w_m)\mathbf{q}_m = 0 \;\forall j\\
    \Leftrightarrow\; & H_\mathcal{C}\left(\sum_{m=1}^{k_d-k_c} (v_m+w_m)\mathbf{q}_m\right)^T = \mathbf{0}^T \\
    \Leftrightarrow\; & \left(\sum_{m=1}^{k_d-k_c} (v_m+w_m)\mathbf{q}_m\right) \in \mathcal{C} \\
    \Leftrightarrow\; & q(x)\left(\sum_{m=1}^{k_d-k_c}(v_m+w_m)x^{m-1} \right) \in I_\mathcal{C}.
\end{align}
Observe that a nontrivial polynomial in $I_\mathcal{C}$ (which corresponds to a nonzero binary vector in $\mathcal{C}$) has degree at least $n-k_c$. Also, $q(x)$ has degree $n-k_d$. Suppose that $v_i \neq w_i$ for some $i$. The degree of $q(x)\left(\sum_{m=1}^{k_d-k_c}(v_m+w_m)x^{m-1} \right)$ is at most $k_d-k_c-1$. Therefore, $q(x)\left(\sum_{m=1}^{k_d-k_c}(v_m+w_m)x^{m-1} \right)$ is not in $I_\mathcal{C}$, which implies that $H_\mathcal{C}\left(\sum_{m=1}^{k_d-k_c} v_m\mathbf{q}_m\right)^T\neq H_\mathcal{C}\left(\sum_{m=1}^{k_d-k_c} w_m\mathbf{q}_m\right)^T$. That is, for each $\mathbf{c}=(c_1,...,c_{k_d-k_c})$, $H_\mathcal{C}\left(\sum_{m=1}^{k_d-k_c} c_m\mathbf{q}_m\right)^T$ is distinct.
\end{proof}

$\mathcal{Q}_3$ has maximum synchronization distance $k_d-k_c$. The code distance of the inner code of $\mathcal{Q}_3$, which is determined by the minimum weight of quantum logical operators, is $d_d$. To find the minimum weight of classical logical operators, we first observe that $\{\mathbf{q}_m\}$ where $m \in \{1,\dots,n-k_c\}$ are linearly independent. To modify a classical message $\mathbf{b}$ to another message $\mathbf{b'}$ is to modify the phases $(-1)^{\tilde{\mathbf{p}}_j\cdot \sum_{m=1}^{k_d-k_c} b_m\mathbf{q}_m}$ to $(-1)^{\tilde{\mathbf{p}}_j\cdot \sum_{m=1}^{k_d-k_c} b'_m\mathbf{q}_m}$, which requires a classical logical operator of the form $Z(\underbracket[0.5pt]{\vphantom{\tilde{\mathbf{q}}_i}\,\mathbf{0}\,}_{\mathclap{a_l}}|\underbracket[0.5pt]{\, \textstyle \sum_{m=1}^{k_d-k_c} c_m\mathbf{q}_m \,}_{\mathclap{n}}|\underbracket[0.5pt]{\vphantom{\tilde{\mathbf{q}}_i}\,\mathbf{0}\,}_{\mathclap{a_r}})$. Note that $\sum_{m=1}^{k_d-k_c} c_m\mathbf{q}_m$ is a code word in $\mathcal{D}$, so its minimum Hamming weight is $d_d$. A product of quantum and classical logical operators can be written as a product of $X$-type and $Z$-type operators associated with code words in $\mathcal{D}$. Therefore, the code distance of $\mathcal{Q}_3$ is $d_d$.

\vspace*{0.2cm}
\noindent \textit{$\mathcal{Q}_4$: a synchronizable hybrid code with more encoded classical bits}
\vspace*{0.2cm}

\noindent $\mathcal{Q}_4$ and $\mathcal{Q}_3$ are constructed from the same initial code $\mathcal{Q}^0_3$. The only difference between $\mathcal{Q}_4$ and $\mathcal{Q}_3$ is that the classical encoding of $\mathcal{Q}_4$ in the second step of the encoding procedure uses two types of operators: $Z(\sum_{m=1}^{k_d-k_c} b_m\mathbf{q}_m)$ for some binary vector $\mathbf{b}=(b_1,\dots,b_{k_d-k_c})\in \mathbb{Z}_2^{k_d-k_c}$ and $X(\sum_{m=2}^{y+1} c_m\mathbf{q}_m)$ for some $\mathbf{c}=(c_2,\dots,c_{y+1})\in \mathbb{Z}_2^{y}$ and for some $y \in \{1,\dots,k_d-k_c-2\}$. These operators encode $k_d-k_c+y$ classical bits of information into the phases of $X$-type and $Z$-type stabilizer generators. $\mathcal{Q}_4$ is also a synchronizable code, so $X(\mathbf{q}_1)$ is applied in the third step of the encoding procedure.

The desired code $\mathcal{Q}_4$ can be obtained after completing the encoding procedure. The inner code of $\mathcal{Q}_4$ labeled by $(b_1,\dots,b_{k_d-k_c})$ and $(c_2,\dots,c_{y+1})$ can be described by the stabilizer generators,
\begin{align}
    &X(\overbracket[0.5pt]{\underbracket[0.5pt]{\,\tilde{\mathbf{q}}_i\,}_{\mathclap{a_l}}}^{\text{last}}|\underbracket[0.5pt]{\,\tilde{\mathbf{q}}_i\,}_{\mathclap{n}}|\overbracket[0.5pt]{\underbracket[0.5pt]{\,\tilde{\mathbf{q}}_i\,}_{\mathclap{a_r}}}^{\text{first}}), \\
    &X(\overbracket[0.5pt]{\underbracket[0.5pt]{\,\tilde{\mathbf{p}}_j\,}_{\mathclap{a_l}}}^{\text{last}}|\underbracket[0.5pt]{\,\tilde{\mathbf{p}}_j\,}_{\mathclap{n}}|\overbracket[0.5pt]{\underbracket[0.5pt]{\,\tilde{\mathbf{p}}_j\,}_{\mathclap{a_r}}}^{\text{first}})(-1)^{\tilde{\mathbf{p}}_j\cdot \sum_{m=1}^{k_d-k_c} b_m\mathbf{q}_m}, \\
    &Z(\underbracket[0.5pt]{\vphantom{\tilde{\mathbf{q}}_i}\,\mathbf{0}\,}_{\mathclap{a_l}}|\underbracket[0.5pt]{\,\tilde{\mathbf{q}}_i\,}_{\mathclap{n}}|\underbracket[0.5pt]{\vphantom{\tilde{\mathbf{q}}_i}\,\mathbf{0}\,}_{\mathclap{a_r}}),  \\
    &Z(\underbracket[0.5pt]{\vphantom{\tilde{\mathbf{p}}_j}\,\mathbf{0}\,}_{\mathclap{a_l}}|\underbracket[0.5pt]{\,\tilde{\mathbf{p}}_j\,}_{\mathclap{n}}|\underbracket[0.5pt]{\vphantom{\tilde{\mathbf{p}}_j}\,\mathbf{0}\,}_{\mathclap{a_r}})(-1)^{\tilde{\mathbf{p}}_j\cdot \left(\mathbf{q}_1+\sum_{m=2}^{y+1} c_m\mathbf{q}_m\right)},   \\
    &Z\left( \begin{array}{c|ccc|c}
    \mathbb{1}_{a_l} & \mathbb{0}_{a_l,a_r} & \mathbb{0}_{a_l,(n-a_l-a_r)} & \mathbb{1}_{a_l} & \mathbb{0}_{a_l,a_r} \\
    \mathbb{0}_{a_r,a_l} & \mathbb{1}_{a_r} & \mathbb{0}_{a_r,(n-a_l-a_r)} & \mathbb{0}_{a_r,a_l} & \mathbb{1}_{a_r}
    \end{array}
    \right). 
\end{align}
Logical $X$ and logical $Z$ operators of the code $\mathcal{Q}_4$ can be obtained by transforming the anticommuting pairs $\{X(\mathbf{s}_l^x),Z(\mathbf{s}_l^z)\}$ where $l \in \{1,\dots,2k_c-n\}$ by the operations in the encoding procedure. Therefore, $\mathcal{Q}_4$ encodes $2k_c-n$ logical qubits.

Synchronization recovery on $\mathcal{Q}_4$ uses the fact that the stabilizer group of the inner code can be described by
\begin{align}
    &X(\overbracket[0.5pt]{\underbracket[0.5pt]{\,\tilde{\mathbf{q}}_i\,}_{\mathclap{a_l+\alpha}}}^{\text{last}}|\underbracket[0.5pt]{\,\tilde{\mathbf{q}}_i\,}_{\mathclap{n}}|\overbracket[0.5pt]{\underbracket[0.5pt]{\,\tilde{\mathbf{q}}_i\,}_{\mathclap{a_r-\alpha}}}^{\text{first}}),\\ 
    &X(\overbracket[0.5pt]{\underbracket[0.5pt]{\,\tilde{\mathbf{p}}_j\,}_{\mathclap{a_l+\alpha}}}^{\text{last}}|\underbracket[0.5pt]{\,\tilde{\mathbf{p}}_j\,}_{\mathclap{n}}|\overbracket[0.5pt]{\underbracket[0.5pt]{\,\tilde{\mathbf{p}}_j\,}_{\mathclap{a_r-\alpha}}}^{\text{first}})(-1)^{\tilde{\mathbf{p}}_j\cdot \mathcal{O}\left(\sum_{m=1}^{k_d-k_c} b_m\mathbf{q}_m,-\alpha\right)}, \\
    &Z(\underbracket[0.5pt]{\vphantom{\tilde{\mathbf{q}}_i}\,\mathbf{0}\,}_{\mathclap{a_l+\alpha}}|\underbracket[0.5pt]{\,\tilde{\mathbf{q}}_i\,}_{\mathclap{n}}|\underbracket[0.5pt]{\vphantom{\tilde{\mathbf{q}}_i}\,\mathbf{0}\,}_{\mathclap{a_r-\alpha}}),  \\
    &Z(\underbracket[0.5pt]{\vphantom{\tilde{\mathbf{p}}_j}\,\mathbf{0}\,}_{\mathclap{a_l+\alpha}}|\underbracket[0.5pt]{\,\tilde{\mathbf{p}}_j\,}_{\mathclap{n}}|\underbracket[0.5pt]{\vphantom{\tilde{\mathbf{p}}_j}\,\mathbf{0}\,}_{\mathclap{a_r-\alpha}})(-1)^{\tilde{\mathbf{p}}_j\cdot \mathcal{O}\left(\mathbf{q}_1+\sum_{m=2}^{y+1} c_m\mathbf{q}_m,-\alpha\right)},   \\
    &Z\left( \begin{array}{c|ccc|c}
    \mathbb{1}_{a_l} & \mathbb{0}_{a_l,a_r} & \mathbb{0}_{a_l,(n-a_l-a_r)} & \mathbb{1}_{a_l} & \mathbb{0}_{a_l,a_r} \\
    \mathbb{0}_{a_r,a_l} & \mathbb{1}_{a_r} & \mathbb{0}_{a_r,(n-a_l-a_r)} & \mathbb{0}_{a_r,a_l} & \mathbb{1}_{a_r}
    \end{array}
    \right), 
\end{align}
for any integer $\alpha$ such that $-a_l \leq \alpha \leq a_r$. This can be obtained by using \cref{thm:sync_proof_subspace} with $\mathbf{v}=\sum_{m=1}^{k_d-k_c} b_m\mathbf{q}_m$ and $\mathbf{w}=\mathbf{q}_1+\sum_{m=2}^{y+1} c_m\mathbf{q}_m$, the fact that $\mathcal{C}^\perp=\langle\tilde{\mathbf{q}}_i,\tilde{\mathbf{p}}_j\rangle$ where $i \in \{1,\dots,n-k_d\}$ and $j \in \{1,\dots,k_d-k_c\}$, and the fact that $\tilde{\mathbf{q}}_i \cdot \mathcal{O}\left(\mathbf{q}_m,-\alpha\right)=0$ for all $i,m,\alpha$.

Similar to other codes, $X$-type error correction must be done before synchronization recovery by measuring $Z(\tilde{\mathbf{q}}_i)$ where $i \in \{1,\dots,n-k_d\}$ on the block of $n$ received qubits and using an error decoder for the code $\mathcal{Q}_\mathcal{D}$. However, the synchronization recovery on $\mathcal{Q}_4$ is slightly different from the previous codes. To see how it can be done, we first present the following lemma.

\begin{lemma} \label{lem:sync_message_value}
    Let $y \in \{1,\dots,k_d-k_c-2\}$ and suppose that $a_l$ and $a_r$ are non-negative integers satisfying $a_l+a_r < k_d-k_c-y$. For any pair $(\mathbf{c},\alpha)$ of binary vector $\mathbf{c}=(c_2,...,c_{y+1}) \in \mathbb{Z}^{y}_2$ and integer $\alpha$ satisfying $-a_l \leq \alpha \leq a_r$, $H_\mathcal{C}\mathcal{O}(\mathbf{q}_1+\sum_{m=2}^{y+1} c_m\mathbf{q}_m,-\alpha)^T$ is distinct.
\end{lemma}
\begin{proof}
Let $\mathbf{v}=(v_2,...,v_{y+1})$ and $\mathbf{w}=(w_2,...,w_{y+1})$ be binary vectors in $\mathbb{Z}^{y}_2$, and let $\beta,\gamma$ be any integers such that $\beta > \gamma$. The following statements are equivalent.
\begin{align}
    &\tilde{\mathbf{p}}_j\cdot \mathcal{O}(\mathbf{q}_1+\sum_{m=2}^{y+1} v_m\mathbf{q}_m,\beta) =  \tilde{\mathbf{p}}_j\cdot \mathcal{O}(\mathbf{q}_1+\sum_{m=2}^{y+1} w_m\mathbf{q}_m,\gamma) \;\forall j\\
    \Leftrightarrow\; & \tilde{\mathbf{p}}_j\cdot \left(\mathcal{O}(\mathbf{q}_1+\sum_{m=2}^{y+1} w_m\mathbf{q}_m,\gamma) + \mathcal{O}(\mathbf{q}_1+\sum_{m=2}^{y+1} v_m\mathbf{q}_m,\beta) \right)=0 \;\forall j\\
    \Leftrightarrow\; & H_\mathcal{C}\left(\mathcal{O}(\mathbf{q}_1+\sum_{m=2}^{y+1} w_m\mathbf{q}_m,\gamma) + \mathcal{O}(\mathbf{q}_1+\sum_{m=2}^{y+1} v_m\mathbf{q}_m,\beta)\right)^T = \mathbf{0}^T\\
    \Leftrightarrow\; & \left(\mathcal{O}(\mathbf{q}_1+\sum_{m=2}^{y+1} w_m\mathbf{q}_m,\gamma) + \mathcal{O}(\mathbf{q}_1+\sum_{m=2}^{y+1} v_m\mathbf{q}_m,\beta)\right) \in \mathcal{C}\\
    \Leftrightarrow\; & \left(\mathbf{q}_1+\sum_{m=2}^{y+1} w_m\mathbf{q}_m + \mathcal{O}(\mathbf{q}_1+\sum_{m=2}^{y+1} v_m\mathbf{q}_m,\beta-\gamma)\right) \in \mathcal{C}
\end{align}

\vspace*{-1cm}
\begin{align}    
    \Leftrightarrow\; & q(x)\left(1+\sum_{m=2}^{y+1}w_m x^{m-1} + x^{\beta-\gamma}\left(1+\sum_{m=2}^{y+1}v_m x^{m-1}\right)\right) \in I_\mathcal{C}.
\end{align}

Case 1. Consider the case that $\beta=\gamma$ and $v_m \neq w_m$ for some $m$. The statements above are equivalent to, 
\begin{equation}
    q(x)\left(\sum_{m=2}^{y+1}(w_m+v_m) x^{m-1}\right) \in I_\mathcal{C}.
\end{equation}
Observe that a nontrivial polynomial in $I_\mathcal{C}$ has degree at least $n-k_c$. However, the degree of 
\begin{equation}
q(x)\left(\sum_{m=2}^{y+1}(w_i+v_i) x^{m-1}\right),
\end{equation}
is at most $n-k_d+y \leq n-k_c-2$. Therefore, the polynomial is not in $I_\mathcal{C}$. This implies that $H_\mathcal{C}\mathcal{O}(\mathbf{q}_1+\sum_{m=2}^{y+1} v_m\mathbf{q}_m,\beta)^T \neq H_\mathcal{C}\mathcal{O}(\mathbf{q}_1+\sum_{m=2}^{y+1} w_m\mathbf{q}_m,\beta)^T$ if $\mathbf{v} \neq \mathbf{w}$.

Case 2. Consider the case that $\beta \neq \gamma$, and suppose that $\beta-\gamma < k_d-k_c-y$. The degree of 
\begin{equation}
q(x)\left(1+\sum_{m=2}^{y+1}w_m x^{m-1} + x^{\beta-\gamma}\left(1+\sum_{m=2}^{y+1}v_m x^{m-1}\right)\right),
\end{equation}
is at most $n-k_d+\beta-\gamma+y \leq n-k_c-1$. Therefore, the polynomial is not in $I_\mathcal{C}$. This implies that $H_\mathcal{C}\mathcal{O}(\mathbf{q}_1+\sum_{m=2}^{y+1} v_m\mathbf{q}_m,\beta)^T \neq H_\mathcal{C}\mathcal{O}(\mathbf{q}_1+\sum_{m=2}^{y+1} w_m\mathbf{q}_m,\gamma)^T$ for any $\mathbf{v}$, $\mathbf{w}$ if $\beta \neq \gamma$.

Combining two cases, we have $H_\mathcal{C}\mathcal{O}(\mathbf{q}_1+\sum_{m=2}^{y+1} v_m\mathbf{q}_m,\beta)^T \neq H_\mathcal{C}\mathcal{O}(\mathbf{q}_1+\sum_{m=2}^{y+1} w_m\mathbf{q}_m,\gamma)^T$ whenever $(\mathbf{v},\beta) \neq (\mathbf{w},\gamma)$. This means that the value of $H_\mathcal{C}\mathcal{O}(\mathbf{q}_1+\sum_{m=2}^{y+1} c_m\mathbf{q}_m,-\alpha)^T$ for each $(\mathbf{c},\alpha)$ is distinct.
\end{proof}

By \cref{lem:sync_message_value}, it is possible to construct a lookup table between $(\mathbf{c},\alpha)$ and $H_\mathcal{C}\mathcal{O}(\mathbf{q}_1+\sum_{m=2}^{y+1} c_m\mathbf{q}_m,-\alpha)^T$ whenever $a_l+a_r < k_d-k_c-y$. When the receiver measures $Z(\tilde{\mathbf{p}}_j)$ where $j \in \{1,\dots,k_d-k_c\}$ on a block of $n$ received qubits, they obtain the value $H_\mathcal{C}\mathcal{O}(\mathbf{q}_1+\sum_{m=2}^{y+1} c_m\mathbf{q}_m,-\alpha)^T$. From this information, $\alpha$ (which tells the misalignment) and $\mathbf{c}$ (which provides a part of the encoded classical information) can be found at the same time.

After the synchronization recovery is done, $X$-type and $Z$-type error corrections on the entire block of $n+a_l+a_r$ qubits can be done by measuring $Z(\underbracket[0.5pt]{\vphantom{\tilde{\mathbf{q}}_i}\,\mathbf{0}\,}_{\mathclap{a_l+a_r}}|\underbracket[0.5pt]{\,\tilde{\mathbf{q}}_i\,}_{\mathclap{n}})$, $Z(\underbracket[0.5pt]{\,\tilde{\mathbf{q}}_i\,}_{\mathclap{n}}|\underbracket[0.5pt]{\vphantom{\tilde{\mathbf{q}}_i}\,\mathbf{0}\,}_{\mathclap{a_l+a_r}})$, and $X(\overbracket[0.5pt]{\underbracket[0.5pt]{\,\tilde{\mathbf{q}}_i\,}_{\mathclap{a_l}}}^{\text{last}}|\underbracket[0.5pt]{\,\tilde{\mathbf{q}}_i\,}_{\mathclap{n}}|\overbracket[0.5pt]{\underbracket[0.5pt]{\,\tilde{\mathbf{q}}_i\,}_{\mathclap{a_r}}}^{\text{first}})$, then using an error decoder for the code $\mathcal{Q}_\mathcal{D}$. Error correction is guaranteed only if there are no more than $\lfloor(d_d-1)/2\rfloor$ $X$-type errors on any $n$ consecutive qubits and there are no more than $\lfloor(d_d-1)/2\rfloor$ $Z$-type errors on the entire block of $n+a_l+a_r$ qubits. 

Finally, the remaining part of the encoded classical information $\mathbf{b}$ can be obtained by measuring $X(\overbracket[0.5pt]{\underbracket[0.5pt]{\,\tilde{\mathbf{p}}_j\,}_{\mathclap{a_l}}}^{\text{last}}|\underbracket[0.5pt]{\,\tilde{\mathbf{p}}_j\,}_{\mathclap{n}}|\overbracket[0.5pt]{\underbracket[0.5pt]{\,\tilde{\mathbf{p}}_j\,}_{\mathclap{a_r}}}^{\text{first}})$, which gives $H_\mathcal{C}\left(\sum_{m=1}^{k_d-k_c} b_m\mathbf{q}_m\right)^T$ as eigenvalues. By \cref{lem:message_value}, we know that $H_\mathcal{C}\left(\sum_{m=1}^{k_d-k_c} b_m\mathbf{q}_m\right)^T$ for each $\mathbf{b}=(b_1,...,b_{k_d-k_c})$ is distinct.

The full procedure to correct Pauli and synchronization errors and obtain the encoded classical information on $\mathcal{Q}_4$ is summarized below.
\begin{enumerate}
    \item Measure $Z(\tilde{\mathbf{q}}_i)$ (where $i \in \{1,\dots,n-k_d\}$) on the received $n$ qubits and perform $X$-type error correction using an error decoder for the code $\mathcal{Q}_\mathcal{D}$.
    \item Measure $Z(\tilde{\mathbf{p}}_j)$ (where $j \in \{1,\dots,k_d-k_c\}$) on the received $n$ qubits to obtain the eigenvalues $H_\mathcal{C}\mathcal{O}(\mathbf{q}_1+\sum_{m=2}^{y+1} c_m\mathbf{q}_m,-\alpha)^T$. Determine $\alpha$ and the first part of the encoded classical information $\mathbf{c}$ from $H_\mathcal{C}\mathcal{O}(\mathbf{q}_1+\sum_{m=2}^{y+1} c_m\mathbf{q}_m,-\alpha)^T$.
    \item Perform synchronization recovery using $\alpha$.
    \item Measure $Z(\underbracket[0.5pt]{\vphantom{\tilde{\mathbf{q}}_i}\,\mathbf{0}\,}_{\mathclap{a_l+a_r}}|\underbracket[0.5pt]{\,\tilde{\mathbf{q}}_i\,}_{\mathclap{n}})$ and $Z(\underbracket[0.5pt]{\,\tilde{\mathbf{q}}_i\,}_{\mathclap{n}}|\underbracket[0.5pt]{\vphantom{\tilde{\mathbf{q}}_i}\,\mathbf{0}\,}_{\mathclap{a_l+a_r}})$ (where $i \in \{1,\dots,n-k_d\}$) and correct $X$-type errors on the entire block of $n+a_l+a_r$ qubits using an error decoder for the code $\mathcal{Q}_\mathcal{D}$.
    \item Measure $X(\overbracket[0.5pt]{\underbracket[0.5pt]{\,\tilde{\mathbf{q}}_i\,}_{\mathclap{a_l}}}^{\text{last}}|\underbracket[0.5pt]{\,\tilde{\mathbf{q}}_i\,}_{\mathclap{n}}|\overbracket[0.5pt]{\underbracket[0.5pt]{\,\tilde{\mathbf{q}}_i\,}_{\mathclap{a_r}}}^{\text{first}})$ (where $i \in \{1,\dots,n-k_d\}$) and correct $Z$-type errors on the entire block of $n+a_l+a_r$ qubits using an error decoder for the code $\mathcal{Q}_\mathcal{D}$.
    \item Measure $X(\overbracket[0.5pt]{\underbracket[0.5pt]{\,\tilde{\mathbf{p}}_j\,}_{\mathclap{a_l}}}^{\text{last}}|\underbracket[0.5pt]{\,\tilde{\mathbf{p}}_j\,}_{\mathclap{n}}|\overbracket[0.5pt]{\underbracket[0.5pt]{\,\tilde{\mathbf{p}}_j\,}_{\mathclap{a_r}}}^{\text{first}})$ to obtain $H_\mathcal{C}\left(\sum_{m=1}^{k_d-k_c} b_m\mathbf{q}_m\right)^T$. Determine the second part of the encoded classical information $\mathbf{b}$.
\end{enumerate}

We can find that the code distance of $\mathcal{Q}_4$ is $d_d$ by an analysis similar to the one for $\mathcal{Q}_3$. $\mathcal{Q}_4$ that encodes $k_d-k_c+y$ classical bits has maximum synchronization distance $k_d-k_c-y$, where $y \in \{1,\dots,k_d-k_c-2\}$. A synchronizable hybrid code that encodes more classical bits can be obtained from a synchronizable hybrid code that encodes fewer classical bits by simply applying more $X(\mathbf{q}_m)$ and attaching fewer ancilla qubits in the encoding procedure. We call this process \emph{the sacrifice of synchronization distance}, as the maximum synchronization distance of the resulting code decreases.

We point out that it is also possible to construct $\mathcal{Q}_4$ with $y=k_d-k_c-1$, which encodes $2(k_d-k_c)-1$ classical bits and has maximum synchronization distance 1 since the only possible value of $(a_l,a_r)$ is $(0,0)$. However, the use of $X(\mathbf{q}_1)$ in the encoding procedure is meaningless because the synchronization recovery cannot be done on this code. Instead of doing so, $X(\mathbf{q}_1)$ can be used to encode one more classical bit, resulting in a non-synchronizable hybrid code $\mathcal{Q}_5$ that encodes $2(k_d-k_c)$ classical bits in total.

\vspace*{0.2cm}
\noindent \textit{$\mathcal{Q}_5$: a hybrid code}
\vspace*{0.2cm}

\noindent $\mathcal{Q}_5$ is constructed from the initial code $\mathcal{Q}^0_3$, similar to $\mathcal{Q}_3$ and $\mathcal{Q}_4$. In the encoding procedure of $\mathcal{Q}_5$, classical information is encoded by operators $Z(\sum_{m=1}^{k_d-k_c} b_m\mathbf{q}_m)$ and $X(\sum_{m=1}^{k_d-k_c} c_m\mathbf{q}_m)$ for some binary vectors $\mathbf{b}=(b_1,\dots,b_{k_d-k_c})$ and $\mathbf{c}=(c_1,\dots,c_{k_d-k_c})$ in $\mathbb{Z}^{k_d-k_c}_2$. With these operators, $2(k_d-k_c)$ classical bits are encoded into the phases of $X$-type and $Z$-type stabilizer generators. We point out this construction is similar to sacrificing all synchronization distance of $\mathcal{Q}_3$ and using $X(\mathbf{q}_1)$ to encode an additional classical bit.

The stabilizer group of the inner code of $\mathcal{Q}_5$ labeled by $(b_1,\dots,b_{k_d-k_c})$ and $(c_1,\dots,c_{k_d-k_c})$ can be described by the stabilizer generators, 
\begin{align}
    &X(\tilde{\mathbf{q}}_i), \\
    &X(\tilde{\mathbf{p}}_j)(-1)^{\tilde{\mathbf{p}}_j\cdot \left(\sum_{m=1}^{k_d-k_c} b_m\mathbf{q}_m\right)},\\
    &Z(\tilde{\mathbf{q}}_i), \\
    &Z(\tilde{\mathbf{p}}_j)(-1)^{\tilde{\mathbf{p}}_j\cdot \left(\sum_{m=1}^{k_d-k_c} c_m\mathbf{q}_m\right)}, 
\end{align}
where $i \in \{1,\dots,n-k_d\}$, $j \in \{1,\dots,k_d-k_c\}$, and $m \in \{1,\dots,k_d-k_c\}$. Logical $X$ and logical $Z$ operators can be obtained by transforming the anticommuting pairs $\{X(\mathbf{s}_l^x),Z(\mathbf{s}_l^z)\}$ where $l \in \{1,\dots,2k_c-n\}$ by the operations in the encoding procedure. Therefore, this code encodes $2k_c-n$ logical qubits.

Error correction on $\mathcal{Q}_5$ can be done by measuring $X(\tilde{\mathbf{q}}_i)$ and $Z(\tilde{\mathbf{q}}_i)$, and using an error decoder for $\mathcal{Q}_\mathcal{D}$. After error correction is done, the encoded classical information can be obtained by measuring $X(\tilde{\mathbf{p}}_j)$ and $Z(\tilde{\mathbf{p}}_j)$. The measurements provide $H_\mathcal{C}\left(\sum_{m=1}^{k_d-k_c} b_m\mathbf{q}_m\right)^T$ and $H_\mathcal{C}\left(\sum_{m=1}^{k_d-k_c} c_m\mathbf{q}_m\right)^T$. By \cref{lem:message_value}, we know that $\mathbf{b}$ and $\mathbf{c}$ can be uniquely determined by these eigenvalues. Similar to other codes, we can show that the code distance of $\mathcal{Q}_5$ is $d_d$.

\subsection{Hybrid subsystem and synchronizable hybrid subsystem codes} \label{subsec:SHSC_subsystem_hybrid}

\vspace*{0.2cm}
\noindent \textit{$\mathcal{Q}_6$: a synchronizable hybrid subsystem code}
\vspace*{0.2cm}

\noindent Starting from a synchronizable subsystem code $\mathcal{Q}_2$, it is possible to obtain a synchronizable hybrid subsystem code $\mathcal{Q}_6$ by sacrificing some synchronization distance. More precisely, the initial code of $\mathcal{Q}_6$ is $\mathcal{Q}^0_2$, and $X(\sum_{m=2}^{y+1} c_m\mathbf{q}_m)$ for some $\mathbf{c}=(c_2,\dots,c_{y+1})\in \mathbb{Z}_2^{y}$ where $y \in \{1,\dots,k_d-k_c-2\}$ is applied in the second step of the encoding procedure to encode $y$ classical bits into the phases of $Z$-type stabilizer generators. Furthermore, $X(\mathbf{q}_1)$ is applied in the third step of the encoding procedure as this code is a synchronizable code.

After completing the encoding procedure, the desired code $\mathcal{Q}_6$ is obtained. The inner (subsystem) code labeled by $(c_2,\dots,c_{y+1})$ can be described by the stabilizer generators,
\begin{align}
    &X(\overbracket[0.5pt]{\underbracket[0.5pt]{\,\tilde{\mathbf{q}}_i\,}_{\mathclap{a_l}}}^{\text{last}}|\underbracket[0.5pt]{\,\tilde{\mathbf{q}}_i\,}_{\mathclap{n}}|\overbracket[0.5pt]{\underbracket[0.5pt]{\,\tilde{\mathbf{q}}_i\,}_{\mathclap{a_r}}}^{\text{first}}), \\
    &Z(\underbracket[0.5pt]{\vphantom{\tilde{\mathbf{q}}_i}\,\mathbf{0}\,}_{\mathclap{a_l}}|\underbracket[0.5pt]{\,\tilde{\mathbf{q}}_i\,}_{\mathclap{n}}|\underbracket[0.5pt]{\vphantom{\tilde{\mathbf{q}}_i}\,\mathbf{0}\,}_{\mathclap{a_r}}),\\
    &Z(\underbracket[0.5pt]{\vphantom{\tilde{\mathbf{p}}_j}\,\mathbf{0}\,}_{\mathclap{a_l}}|\underbracket[0.5pt]{\,\tilde{\mathbf{p}}_j\,}_{\mathclap{n}}|\underbracket[0.5pt]{\vphantom{\tilde{\mathbf{p}}_j}\,\mathbf{0}\,}_{\mathclap{a_r}})(-1)^{\tilde{\mathbf{p}}_j \cdot \left(\mathbf{q}_1+\sum_{m=2}^{y+1} c_m\mathbf{q}_m\right)},   \\
    &Z\left( \begin{array}{c|ccc|c}
    \mathbb{1}_{a_l} & \mathbb{0}_{a_l,a_r} & \mathbb{0}_{a_l,(n-a_l-a_r)} & \mathbb{1}_{a_l} & \mathbb{0}_{a_l,a_r} \\
    \mathbb{0}_{a_r,a_l} & \mathbb{1}_{a_r} & \mathbb{0}_{a_r,(n-a_l-a_r)} & \mathbb{0}_{a_r,a_l} & \mathbb{1}_{a_r}
    \end{array}
    \right), 
\end{align}
and its gauge group is described by the stabilizer generators, $iI^{\otimes n+a_l+a_r}$, and
\begin{align}
    &X(\overbracket[0.5pt]{\underbracket[0.5pt]{\,\tilde{\mathbf{p}}_j\,}_{\mathclap{a_l}}}^{\text{last}}|\underbracket[0.5pt]{\,\tilde{\mathbf{p}}_j\,}_{\mathclap{n}}|\overbracket[0.5pt]{\underbracket[0.5pt]{\,\tilde{\mathbf{p}}_j\,}_{\mathclap{a_r}}}^{\text{first}}), \\
    &Z(\underbracket[0.5pt]{\vphantom{\mathbf{q}'_m}\,\mathbf{0}\,}_{\mathclap{a_l}}|\underbracket[0.5pt]{\,\mathbf{q}'_m\,}_{\mathclap{n}}|\underbracket[0.5pt]{\vphantom{\mathbf{q}'_m}\,\mathbf{0}\,}_{\mathclap{a_r}})(-1)^{\mathbf{q}'_m \cdot \left(\mathbf{q}_1+\sum_{m=2}^{y+1} c_m\mathbf{q}_m\right)},
\end{align}
where $i \in \{1,\dots,n-k_d\}$, $j \in \{1,\dots,k_d-k_c\}$, and $m \in \{1,\dots,k_d-k_c\}$. 

$\mathcal{Q}_6$ has $k_d-k_c$ gauge qubits similar to the code $\mathcal{Q}_2$. Logical $X$ and logical $Z$ operators of $\mathcal{Q}_6$ can be obtained by transforming the anticommuting pairs $\{X(\mathbf{s}_l^x),Z(\mathbf{s}_l^z)\}$ where $l \in \{1,\dots,2k_c-n\}$ by the operations in the encoding procedure. Therefore, $\mathcal{Q}_6$ encodes $2k_c-n$ logical qubits.

Synchronization recovery on $\mathcal{Q}_6$ uses the fact that the stabilizer group of its inner code labeled by $(c_2,\dots,c_{y+1})$ can also be described by the stabilizer generators,
\begin{align}
    &X(\overbracket[0.5pt]{\underbracket[0.5pt]{\,\tilde{\mathbf{q}}_i\,}_{\mathclap{a_l+\alpha}}}^{\text{last}}|\underbracket[0.5pt]{\,\tilde{\mathbf{q}}_i\,}_{\mathclap{n}}|\overbracket[0.5pt]{\underbracket[0.5pt]{\,\tilde{\mathbf{q}}_i\,}_{\mathclap{a_r-\alpha}}}^{\text{first}}), \\
    &Z(\underbracket[0.5pt]{\vphantom{\tilde{\mathbf{q}}_i}\,\mathbf{0}\,}_{\mathclap{a_l+\alpha}}|\underbracket[0.5pt]{\,\tilde{\mathbf{q}}_i\,}_{\mathclap{n}}|\underbracket[0.5pt]{\vphantom{\tilde{\mathbf{q}}_i}\,\mathbf{0}\,}_{\mathclap{a_r-\alpha}}),\\
    &Z(\underbracket[0.5pt]{\vphantom{\tilde{\mathbf{p}}_j}\,\mathbf{0}\,}_{\mathclap{a_l+\alpha}}|\underbracket[0.5pt]{\,\tilde{\mathbf{p}}_j\,}_{\mathclap{n}}|\underbracket[0.5pt]{\vphantom{\tilde{\mathbf{p}}_j}\,\mathbf{0}\,}_{\mathclap{a_r-\alpha}})(-1)^{\tilde{\mathbf{p}}_j \cdot \mathcal{O}\left(\mathbf{q}_1+\sum_{m=2}^{y+1} c_m\mathbf{q}_m,-\alpha\right)} ,  \\
    &Z\left( \begin{array}{c|ccc|c}
    \mathbb{1}_{a_l} & \mathbb{0}_{a_l,a_r} & \mathbb{0}_{a_l,(n-a_l-a_r)} & \mathbb{1}_{a_l} & \mathbb{0}_{a_l,a_r} \\
    \mathbb{0}_{a_r,a_l} & \mathbb{1}_{a_r} & \mathbb{0}_{a_r,(n-a_l-a_r)} & \mathbb{0}_{a_r,a_l} & \mathbb{1}_{a_r}
    \end{array}
    \right), 
\end{align}
and the gauge group is described by the stabilizer generators, $iI^{\otimes n+a_l+a_r}$, and
\begin{align}
    &X(\overbracket[0.5pt]{\underbracket[0.5pt]{\,\tilde{\mathbf{p}}_j\,}_{\mathclap{a_l+\alpha}}}^{\text{last}}|\underbracket[0.5pt]{\,\tilde{\mathbf{p}}_j\,}_{\mathclap{n}}|\overbracket[0.5pt]{\underbracket[0.5pt]{\,\tilde{\mathbf{p}}_j\,}_{\mathclap{a_r-\alpha}}}^{\text{first}}), \\
    &Z(\underbracket[0.5pt]{\vphantom{\mathbf{q}'_m}\,\mathbf{0}\,}_{\mathclap{a_l+\alpha}}|\underbracket[0.5pt]{\,\mathbf{q}'_m\,}_{\mathclap{n}}|\underbracket[0.5pt]{\vphantom{\mathbf{q}'_m}\,\mathbf{0}\,}_{\mathclap{a_r-\alpha}})(-1)^{\mathbf{q}'_m \cdot \mathcal{O}\left(\mathbf{q}_1+\sum_{m=2}^{y+1} c_m\mathbf{q}_m,-\alpha\right)},
\end{align}
for any integer $\alpha$ such that $-a_l \leq \alpha \leq a_r$. Here we use \cref{thm:sync_proof_subsystem} with $\mathbf{w}=\mathbf{q}_1+\sum_{m=2}^{y+1} c_m\mathbf{q}_m$, the fact that $\mathcal{C}^\perp=\langle\tilde{\mathbf{q}}_i,\tilde{\mathbf{p}}_j\rangle$ where $i \in \{1,\dots,n-k_d\}$ and $j \in \{1,\dots,k_d-k_c\}$, and the fact that $\tilde{\mathbf{q}}_i \cdot \mathcal{O}\left(\mathbf{q}_m,-\alpha\right)=0$ for all $i,m,\alpha$.

Similar to other codes, $X$-type error correction must be done before synchronization recovery by measuring $Z(\tilde{\mathbf{q}}_i)$ where $i \in \{1,\dots,n-k_d\}$ on the block of $n$ received qubits and using an error decoder for the code $\mathcal{Q}_\mathcal{D}$. Similar to $\mathcal{Q}_4$, The synchronization recovery on $\mathcal{Q}_6$ can be done by measuring $Z(\tilde{\mathbf{p}}_j)$ where $j \in \{1,\dots,k_d-k_c\}$ on the received qubits, which gives $H_\mathcal{C}\mathcal{O}(\mathbf{q}_1+\sum_{m=2}^{y+1} c_m\mathbf{q}_m,-\alpha)^T$. By \cref{lem:sync_message_value}, we can uniquely determine $\alpha$ and $\mathbf{c}$ by these measurement results if $a_l+a_r < k_d-k_c-y$. $\mathbf{c}$ provides the encoded classical information, and synchronization recovery can be done using $\alpha$.

After synchronization recovery is done, $X$-type and $Z$-type error corrections on the entire block of $n+a_l+a_r$ qubits can be done by measuring $Z(\underbracket[0.5pt]{\vphantom{\tilde{\mathbf{q}}_i}\,\mathbf{0}\,}_{\mathclap{a_l+a_r}}|\underbracket[0.5pt]{\,\tilde{\mathbf{q}}_i\,}_{\mathclap{n}})$, $Z(\underbracket[0.5pt]{\,\tilde{\mathbf{q}}_i\,}_{\mathclap{n}}|\underbracket[0.5pt]{\vphantom{\tilde{\mathbf{q}}_i}\,\mathbf{0}\,}_{\mathclap{a_l+a_r}})$, and $X(\overbracket[0.5pt]{\underbracket[0.5pt]{\,\tilde{\mathbf{q}}_i\,}_{\mathclap{a_l}}}^{\text{last}}|\underbracket[0.5pt]{\,\tilde{\mathbf{q}}_i\,}_{\mathclap{n}}|\overbracket[0.5pt]{\underbracket[0.5pt]{\,\tilde{\mathbf{q}}_i\,}_{\mathclap{a_r}}}^{\text{first}})$, then using an error decoder for the code $\mathcal{Q}_\mathcal{D}$. Error correction is guaranteed only if there are no more than $\lfloor(d_d-1)/2\rfloor$ $X$-type errors on any $n$ consecutive qubits and there are no more than $\lfloor(d_d-1)/2\rfloor$ $Z$-type errors on the entire block of $n+a_l+a_r$ qubits. 

The full procedure to correct Pauli and synchronization errors and obtain the encoded classical information on $\mathcal{Q}_6$ is summarized below.
\begin{enumerate}
    \item Measure $Z(\tilde{\mathbf{q}}_i)$ (where $i \in \{1,\dots,n-k_d\}$) on the received $n$ qubits and perform $X$-type error correction using an error decoder for the code $\mathcal{Q}_\mathcal{D}$.
    \item Measure $Z(\tilde{\mathbf{p}}_j)$ (where $j \in \{1,\dots,k_d-k_c\}$) on the received $n$ qubits to obtain the eigenvalues $H_\mathcal{C}\mathcal{O}(\mathbf{q}_1+\sum_{m=2}^{y+1} c_m\mathbf{q}_m,-\alpha)^T$. Determine $\alpha$ and the first part of the encoded classical information $\mathbf{c}$ from $H_\mathcal{C}\mathcal{O}(\mathbf{q}_1+\sum_{m=2}^{y+1} c_m\mathbf{q}_m,-\alpha)^T$.
    \item Perform synchronization recovery using $\alpha$.
    \item Measure $Z(\underbracket[0.5pt]{\vphantom{\tilde{\mathbf{q}}_i}\,\mathbf{0}\,}_{\mathclap{a_l+a_r}}|\underbracket[0.5pt]{\,\tilde{\mathbf{q}}_i\,}_{\mathclap{n}})$ and $Z(\underbracket[0.5pt]{\,\tilde{\mathbf{q}}_i\,}_{\mathclap{n}}|\underbracket[0.5pt]{\vphantom{\tilde{\mathbf{q}}_i}\,\mathbf{0}\,}_{\mathclap{a_l+a_r}})$ (where $i \in \{1,\dots,n-k_d\}$) and correct $X$-type errors on the entire block of $n+a_l+a_r$ qubits using an error decoder for the code $\mathcal{Q}_\mathcal{D}$.
    \item Measure $X(\overbracket[0.5pt]{\underbracket[0.5pt]{\,\tilde{\mathbf{q}}_i\,}_{\mathclap{a_l}}}^{\text{last}}|\underbracket[0.5pt]{\,\tilde{\mathbf{q}}_i\,}_{\mathclap{n}}|\overbracket[0.5pt]{\underbracket[0.5pt]{\,\tilde{\mathbf{q}}_i\,}_{\mathclap{a_r}}}^{\text{first}})$ (where $i \in \{1,\dots,n-k_d\}$) and correct $Z$-type errors on the entire block of $n+a_l+a_r$ qubits using an error decoder for the code $\mathcal{Q}_\mathcal{D}$.
\end{enumerate}

We can find that the code distance of $\mathcal{Q}_6$ is $d_d$ by an analysis similar to the one for $\mathcal{Q}_3$. $\mathcal{Q}_6$ that encodes $y$ classical bits has maximum synchronization distance $k_d-k_c-y$, where $y \in \{1,\dots,k_d-k_c-2\}$. 

It should be noted that a synchronizable hybrid subsystem code $\mathcal{Q}_6$ and a synchronizable hybrid code $\mathcal{Q}_4$ of the same $y$ are related by gauge fixing. Starting from $\mathcal{Q}_6$ of a particular $y$, measuring the gauge operators $X(\overbracket[0.5pt]{\underbracket[0.5pt]{\,\tilde{\mathbf{p}}_j\,}_{\mathclap{a_l}}}^{\text{last}}|\underbracket[0.5pt]{\,\tilde{\mathbf{p}}_j\,}_{\mathclap{n}}|\overbracket[0.5pt]{\underbracket[0.5pt]{\,\tilde{\mathbf{p}}_j\,}_{\mathclap{a_r}}}^{\text{first}})$ for all $j \in \{1,\dots,k_d-k_c\}$ results in $\mathcal{Q}_4$ with the same $y$ in which the classical information corresponding to $\mathbf{b}=(b_1,\dots,b_{k_d-k_c})$ is not encoded. The destabilizers obtained via gauge fixing (which are $Z$-type operators) can be used to encode this part of classical information.

If all synchronization distance of $\mathcal{Q}_2$ or $\mathcal{Q}_6$ is sacrificed (which corresponds to the case of $\mathcal{Q}_6$ with $y=k_d-k_c-1$), $X(\mathbf{q}_1)$ can be used to encode a classical bit since protection against synchronization errors is not required. In that case, we obtain a non-synchronizable hybrid subsystem code $\mathcal{Q}_7$ that has $k_d-k_c$ gauge qubits and encodes $k_d-k_c$ classical bits.

\vspace*{0.2cm}
\noindent \textit{$\mathcal{Q}_7$: a hybrid subsystem code}
\vspace*{0.2cm}

\noindent $\mathcal{Q}_7$ is constructed from the initial code $\mathcal{Q}^0_2$, similar to $\mathcal{Q}_2$ and $\mathcal{Q}_6$. In the second step of the encoding procedure of $\mathcal{Q}_7$, classical information is encoded by an operator $X(\sum_{m=1}^{k_d-k_c} c_m\mathbf{q}_m)$ for some binary vector $\mathbf{c}=(c_1,\dots,c_{k_d-k_c}) \in \mathbb{Z}^{k_d-k_c}_2$. This code has $k_d-k_c$ gauge qubits, and $k_d-k_c$ classical bits are encoded into the phases of $Z$-type stabilizer generators. The construction of $\mathcal{Q}_7$ is similar to sacrificing all synchronization distance of $\mathcal{Q}_2$ and using $X(\mathbf{q}_1)$ to encode an additional classical bit.

The stabilizer group of the inner code of $\mathcal{Q}_7$ labeled by $(c_1,\dots,c_{k_d-k_c})$ can be described by the stabilizer generators, 
\begin{align}
    &X(\tilde{\mathbf{q}}_i), \\
    &Z(\tilde{\mathbf{q}}_i), \\
    &Z(\tilde{\mathbf{p}}_j)(-1)^{\tilde{\mathbf{p}}_j\cdot \left(\sum_{m=1}^{k_d-k_c} c_m\mathbf{q}_m\right)}, 
\end{align}
and the gauge group is described by the stabilizer generators, $iI^{\otimes n+a_l+a_r}$, and
\begin{align}
    &X(\tilde{\mathbf{p}}_j), \\
    &Z(\mathbf{q}'_m)(-1)^{\mathbf{q}'_m \cdot \left(\sum_{m=1}^{k_d-k_c} c_m\mathbf{q}_m\right)}, 
\end{align}
where $i \in \{1,\dots,n-k_d\}$, $j \in \{1,\dots,k_d-k_c\}$, and $m \in \{1,\dots,k_d-k_c\}$. Logical $X$ and logical $Z$ operators can be obtained by transforming the anticommuting pairs $\{X(\mathbf{s}_l^x),Z(\mathbf{s}_l^z)\}$ where $l \in \{1,\dots,2k_c-n\}$ by the operations in the encoding procedure. Thus, this code encodes $2k_c-n$ logical qubits.

Error correction on $\mathcal{Q}_7$ can be done by measuring $X(\tilde{\mathbf{q}}_i)$ and $Z(\tilde{\mathbf{q}}_i)$, and using an error decoder for $\mathcal{Q}_\mathcal{D}$. After error correction is done, the encoded classical information can be obtained by measuring $Z(\tilde{\mathbf{p}}_j)$, which provides $H_\mathcal{C}\left(\sum_{m=1}^{k_d-k_c} c_m\mathbf{q}_m\right)^T$. By \cref{lem:message_value}, $\mathbf{c}$ can be uniquely determined. We can show that the code distance of $\mathcal{Q}_7$ is $d_d$.

Lastly, we point out that the hybrid subsystem code $\mathcal{Q}_7$ and the hybrid code $\mathcal{Q}_5$ are related by gauge fixing; $\mathcal{Q}_5$ in which the classical information corresponding to $\mathbf{b}=(b_1,\dots,b_{k_d-k_c})$ is not encoded can be obtained from $\mathcal{Q}_7$ by measuring the gauge operators $X(\tilde{\mathbf{p}}_j)$. The destabilizers obtained via gauge fixing can be used to further encode classical information.

\subsection{Trade-offs between the maximum synchronization distance and the numbers of gauge qubits and encoded classical bits}

In this section, we discuss parts of our construction which lead to the trade-offs between the maximum synchronization distance, the number of gauge qubits, and the number of encoded classical bits.

First, note that the initial code $\mathcal{Q}^0_2$ of a synchronizable subsystem code $\mathcal{Q}_2$ is constructed from a subsystem code $\mathcal{Q}_1=\mathcal{Q}^0_1$ by fixing the gauge operators $Z(\tilde{\mathbf{t}}_j)$ for all $j \in \{1,\dots,k_d-k_c\}$. By \cref{thm:op_pairing}, gauge fixing results in the stabilizer group that contains $Z(\mathbf{v})$ for all $\mathbf{v} \in \mathcal{C}^\perp$. Since $\mathcal{C}^\perp$ is a cyclic code, a cyclic shift of any $\mathbf{v} \in \mathcal{C}^\perp$ is also in $\mathcal{C}^\perp$. We can see that the subgroup generated by $Z(\tilde{\mathbf{p}}_j)$ has cyclic structure, and this is possible because $Z(\tilde{\mathbf{t}}_j)$ for all $j$ are fixed. The destabilizers of $X$-type obtained during gauge fixing can be used to encode phases to $Z(\tilde{\mathbf{p}}_j)$. 

To gain synchronization recovery property on $\mathcal{Q}_2$, $X(\mathbf{q}_1)$ is applied as a marker to ``break'' the cyclic symmetry; this is similar to marking a point on the circumference of a circle so that its orientation can be determined if rotation is performed. Misalignment corresponds to a transformation of the phase parameter $\mathbf{q}_1$ to some $\mathbf{q}_m \neq \mathbf{q}_1$. 
The number of linearly independent $\mathbf{q}_m$ which are also linearly independent from the generators of $\mathcal{C}$ is $k_d-k_c$ (see the proof of \cref{lem:sync_value}). Therefore, the maximum synchronization distance of $\mathcal{Q}_2$ is $k_d-k_c$, which is equal to the number of gauge qubits that are fixed. Similar argument is also applicable to the maximum synchronization distance of a synchronizable hybrid code $\mathcal{Q}_3$.

By further fixing $X(\tilde{\mathbf{t}}_j)$ for all $j \in \{1,\dots,k_d-k_c\}$, the initial code $\mathcal{Q}^0_3$ of $\mathcal{Q}_3$ is obtained from $\mathcal{Q}^0_2$. The subgroup generated by $X(\tilde{\mathbf{p}}_j)$ (or $X(\tilde{\mathbf{t}}_j)$) also has cyclic structure. Unfortunately, due to asymmetry caused by CNOT gates in the encoding procedure, the phases of $X$-type and $Z$-type stabilizer generators cannot be used for synchronization recovery at the same time, thus the synchronization distance cannot be further increased. Also, the number of linearly independent $\mathbf{q}_m$ is still the same. Nevertheless, we can use the operator $Z(\mathbf{q}_m)$ for all $m \in \{1,\dots,k_d-k_c\}$ to encode phases to $X(\tilde{\mathbf{p}}_j)$. Therefore, $\mathcal{Q}_3$ can encode $k_d-k_c$ classical bits, which is equal to the number of additional gauge qubits that are fixed.

For a synchronizable hybrid code $\mathcal{Q}_4$ that encodes more classical bits than $\mathcal{Q}_3$, the marker $\mathbf{q}_1$ for synchronization recovery and the vector $\sum_{m=2}^{y+1} c_m\mathbf{q}_m$ representing classical information are encoded in the phases of $Z(\tilde{\mathbf{p}}_j)$. Consider the case that $y=1$ in which the initial phase parameter is $\mathbf{q}_1+c_2\mathbf{q}_2$. Misalignment transforms this parameter to $\mathbf{q}_m+c_2\mathbf{q}_{m+1}$ for some $m \neq 1$. 
The number of cyclic shifts of $\mathbf{q}_1+c_2\mathbf{q}_2$ (including the original one) which are guaranteed to be linearly independent from one another and from the generators of $\mathcal{C}$ for any choice of $c_2$ is $k_d-k_c-1$ (see the proof of \cref{lem:sync_message_value}). Therefore, the maximum synchronization distance in this case is $k_d-k_c-1$. Similarly, we can show that encoding $y$ more classical bits results in decreasing the maximum synchronization distance by $y$. A similar argument is also applicable to the relationship between the maximum synchronization distance of a synchronizable subsystem code $\mathcal{Q}_2$ and a synchronizable hybrid subsystem code $\mathcal{Q}_6$.

As previously mentioned in the constructions of a hybrid code $\mathcal{Q}_5$ and a hybrid subsystem code $\mathcal{Q}_7$, when sacrificing all synchronization distance of a synchronizable code, $X(\mathbf{q}_1)$ can be used to encode an additional classical bit as protection against synchronization errors is not required.

Trade-offs between the maximum synchronization distance, the number of gauge qubits, and the number of encoded classical bits are summarized in the corollary below.

\begin{corollary} \label{cor:QSC_tradeoff}
Let $r,m,$ and $d_\mathrm{sync,max}$ denote the number of gauge qubits, the number of encoded classical bits, and the maximum synchronization distance, respectively.
\begin{enumerate}[leftmargin=1cm] 
    \item For any synchronizable code in \cref{thm:QSC_unified} (in which $d_\mathrm{sync,max}>1$),
    \begin{equation}
        r+m+d_\mathrm{sync,max} = 2(k_d-k_c).
    \end{equation}
    \item For any non-synchronizable code in \cref{thm:QSC_unified} (in which $d_\mathrm{sync,max}=1$),
    \begin{equation}
        r+m+d_\mathrm{sync,max} = 2(k_d-k_c)+1.
    \end{equation}
\end{enumerate}
\end{corollary}
The difference in the total numbers come from the fact that for any synchronizable code in our construction, $X(\mathbf{q}_1)$ is used in the encoding procedure to break the cyclic symmetry in order to gain the synchronization recovery property, and thus cannot be used to encode a classical bit of information.

\section{General methods to construct subsystem codes, hybrid codes, and hybrid subsystem codes of CSS type from classical codes} \label{sec:hybrid_subsystem_construction}

\cref{thm:QSC_unified} provides a method to construct quantum codes in the synchronizable hybrid subsystem code family, which also includes a non-synchronizable version of hybrid codes, subsystem codes, and hybrid subsystem codes. Putting aside the property to correct synchronizable errors, one may ask if there is any more general way to construct such codes. In this section, we state a theorem similar to the CSS construction that provides a method to construct subsystem codes of CSS type from classical codes, then propose new theorems that construct hybrid codes and hybrid subsystem codes of CSS type. Note that the classical codes required by these constructions are not limited to classical cyclic codes.  

We start by providing a construction of subsystem codes of CSS type, which has been studied in \cite{Aly2006, Liu2024}. Here, $\mathcal{C}_1+\mathcal{C}_2$ denotes the code $\{\mathbf{v}_1+\mathbf{v}_2\;|\;\mathbf{v}_1\in \mathcal{C}_1, \mathbf{v}_2\in \mathcal{C}_2\}$.

\begin{theorem} \label{thm:CSS_subsystem}
     Let $\mathcal{C}_x$ and $\mathcal{C}_z$ be classical linear codes with parameters $\left[n,k_x\right]$ and $\left[n,k_z\right]$ respectively. Suppose that $\mathcal{C}_x+\mathcal{C}_z^\perp$ and $\mathcal{C}_z+\mathcal{C}_x^\perp$ are classical linear codes with parameters $\left[n,r_x\right]$ and $\left[n,r_z\right]$ respectively. Then there exists a quantum subsystem code with gauge group $\left\langle iI^{\otimes n},X(\mathcal{C}_z^\perp),Z(\mathcal{C}_x^\perp)\right\rangle$, stabilizer group $\left\langle X(\mathcal{C}_x\cap \mathcal{C}_z^\perp),Z(\mathcal{C}_z\cap \mathcal{C}_x^\perp)\right\rangle$, and parameters $\codepar{n,k,r,d}$, where 
    \begin{align*}
        k & = r_x+k_z-n = r_z+k_x-n, \\
        r & = r_x-k_x = r_z-k_z, \\
        d_x & = \min\wt\left(\left(\mathcal{C}_x+\mathcal{C}_z^\perp\right)\setminus \mathcal{C}_z^\perp\right), \\
        d_z & = \min\wt\left(\left(\mathcal{C}_z+\mathcal{C}_x^\perp\right)\setminus \mathcal{C}_x^\perp\right), \\
        d & =\min\left\{d_x, d_z\right\}.
    \end{align*}
\end{theorem}

Next, we propose constructions of hybrid and hybrid subsystem codes of CSS type.

\begin{theorem} \label{thm:CSS_hybrid}
    Let $\mathcal{C}_x$ and $\mathcal{C}_z$ be classical linear codes with parameters $\left[n,k_x\right]$ and $\left[n,k_z\right]$ respectively. Let $\mathcal{D}_x$ and $\mathcal{D}_z$ be supercodes of $\mathcal{C}_x$ and $\mathcal{C}_z$ (i.e., $\mathcal{C}_x\subseteq \mathcal{D}_x$ and $\mathcal{C}_z\subseteq \mathcal{D}_z$) with parameters $\left[n,m_x\right]$ and $\left[n,m_z\right]$ respectively. Let $\mathcal{C}_z^\perp\subseteq \mathcal{C}_x$. Then there exists a hybrid stabilizer code with inner stabilizer group $\left\langle X(\mathcal{C}_z^\perp),Z(\mathcal{C}_x^\perp)\right\rangle$, outer stabilizer group $\left\langle X(\mathcal{D}_z^\perp),Z(\mathcal{D}_x^\perp)\right\rangle$, and parameters $\codepar{n,k\!:\!m,d}$, where 
    \begin{align*}
        k & = k_x+k_z-n, \\
        m & = m_x+m_z-k_x-k_z,
    \end{align*}
    \vspace*{-1cm}
    \begin{align*}
        d_x & = \min\wt\left(\mathcal{D}_x\setminus \mathcal{C}_z^\perp\right), \\
        d_z & = \min\wt\left(\mathcal{D}_z\setminus \mathcal{C}_x^\perp\right), \\
        d & =\min\left\{d_x, d_z\right\}.
    \end{align*}
\end{theorem}

\begin{theorem} \label{thm:CSS_hybrid_subsystem}
     Let $\mathcal{C}_x$ and $\mathcal{C}_z$ be classical linear codes with parameters $\left[n,k_x\right]$ and $\left[n,k_z\right]$ respectively. Let $\mathcal{D}_x$ and $\mathcal{D}_z$ be supercodes of $\mathcal{C}_x$ and $\mathcal{C}_z$ (i.e., $\mathcal{C}_x\subseteq \mathcal{D}_x$ and $\mathcal{C}_z\subseteq \mathcal{D}_z$) with parameters $\left[n,m_x\right]$ and $\left[n,m_z\right]$ respectively, such that $\left(\mathcal{D}_x\setminus \mathcal{C}_x\right)\cap \mathcal{C}_z^\perp=\left(\mathcal{D}_z\setminus \mathcal{C}_z\right)\cap \mathcal{C}_x^\perp=\emptyset$. Suppose that $\mathcal{C}_x+\mathcal{C}_z^\perp$ and $\mathcal{C}_z+\mathcal{C}_x^\perp$ are classical linear codes with parameters $\left[n,r_x\right]$ and $\left[n,r_z\right]$ respectively. Then there exists a hybrid subsystem code with inner gauge group $\left\langle iI^{\otimes n},X(\mathcal{C}_z^\perp),Z(\mathcal{C}_x^\perp)\right\rangle$, outer gauge group $\left\langle iI^{\otimes n},X(\mathcal{D}_z^\perp),Z(\mathcal{D}_x^\perp)\right\rangle$, inner stabilizer group $\left\langle X(\mathcal{C}_x\cap \mathcal{C}_z^\perp),Z(\mathcal{C}_z\cap \mathcal{C}_x^\perp)\right\rangle$, outer stabilizer group $\left\langle X(\mathcal{D}_x\cap \mathcal{D}_z^\perp),Z(\mathcal{D}_z\cap \mathcal{D}_x^\perp)\right\rangle$, and parameters $\codepar{n,k\!:\!m,r,d}$, where 
    \begin{align*}
        k & = r_x+k_z-n = r_z+k_x-n, \\
        m & = m_x+m_z-k_x-k_z, \\
        r & = r_x-k_x = r_z-k_z, \\
        d_x & = \min\wt\left(\left(\mathcal{D}_x+\mathcal{D}_z^\perp\right)\setminus \mathcal{C}_z^\perp\right), \\
        d_z & = \min\wt\left(\left(\mathcal{D}_z+\mathcal{D}_x^\perp\right)\setminus \mathcal{C}_x^\perp\right), \\
        d & =\min\left\{d_x, d_z\right\}.
    \end{align*}
\end{theorem}
Proofs of \cref{thm:CSS_hybrid,thm:CSS_hybrid_subsystem} are provided in \cref{subsec:CSS_hybrid_proof,subsec:css_hybrid_subsystem_proof}, respectively.

We can verify that the constructions of the codes $\mathcal{Q}_1$, $\mathcal{Q}_5$, and $\mathcal{Q}_7$ in \cref{thm:QSC_unified} are consistent with \cref{thm:CSS_subsystem,thm:CSS_hybrid,thm:CSS_hybrid_subsystem} by observing the following: Let $\tilde{\mathbf{q}}_i,\tilde{\mathbf{t}}_j,\mathbf{s}_l^x,\mathbf{s}_l^z,\mathbf{t}_m^x,\mathbf{t}_m^z$ be defined as in \cref{thm:op_pairing}, where $i \in \{1,\dots,n-k_d\}$, $j \in \{1,\dots,k_d-k_c\}$, $l \in \{1,\dots,2k_c-n\}$, $m \in \{1,\dots,k_d-k_c\}$.
\begin{enumerate}
    \item By applying \cref{thm:CSS_subsystem} to the codes $\mathcal{C}_x=\langle \tilde{\mathbf{q}}_i,\mathbf{s}_l^x \rangle$, $\mathcal{C}_x^\perp=\langle \tilde{\mathbf{q}}_i,\tilde{\mathbf{t}}_j,\mathbf{t}_m^z \rangle$, $\mathcal{C}_z=\langle \tilde{\mathbf{q}}_i,\mathbf{s}_l^z \rangle$, $\mathcal{C}_z^\perp=\langle \tilde{\mathbf{q}}_i,\tilde{\mathbf{t}}_j,\mathbf{t}_m^x \rangle$, the subsystem code $\mathcal{Q}_1$ in \cref{thm:QSC_unified} can be obtained.
    \item By applying \cref{thm:CSS_hybrid} to the codes $\mathcal{C}_x=\mathcal{C}_z=\mathcal{C}=\langle \tilde{\mathbf{q}}_i,\tilde{\mathbf{t}}_j,\mathbf{s}_l^x \rangle = \langle \tilde{\mathbf{q}}_i,\tilde{\mathbf{t}}_j,\mathbf{s}_l^z \rangle$, $\mathcal{C}_x^\perp=\mathcal{C}_z^\perp=\mathcal{C}^\perp=\langle \tilde{\mathbf{q}}_i,\tilde{\mathbf{t}}_j \rangle$, $\mathcal{D}_x = \mathcal{D}_z = \mathcal{D} = \langle \tilde{\mathbf{q}}_i,\tilde{\mathbf{t}}_j,\mathbf{s}_l^x,\mathbf{t}_m^x \rangle = \langle \tilde{\mathbf{q}}_i,\tilde{\mathbf{t}}_j,\mathbf{s}_l^z,\mathbf{t}_m^z \rangle$, and $\mathcal{D}_x^\perp = \mathcal{D}_z^\perp = \mathcal{D}^\perp = \langle \tilde{\mathbf{q}}_i \rangle$, the hybrid code $\mathcal{Q}_5$ in \cref{thm:QSC_unified} can be obtained.
    \item By applying \cref{thm:CSS_hybrid_subsystem} to the codes $\mathcal{C}_x = \langle \tilde{\mathbf{q}}_i,\mathbf{s}_l^x \rangle$, $\mathcal{C}_x^\perp = \langle \tilde{\mathbf{q}}_i,\tilde{\mathbf{t}}_j,\mathbf{t}_m^z \rangle$, $\mathcal{C}_z = \langle \tilde{\mathbf{q}}_i,\tilde{\mathbf{t}}_j,\mathbf{s}_l^z \rangle$, $\mathcal{C}_z^\perp = \langle \tilde{\mathbf{q}}_i,\tilde{\mathbf{t}}_j \rangle$, $\mathcal{D}_x = \langle \tilde{\mathbf{q}}_i,\mathbf{s}_l^x,\mathbf{t}_m^x \rangle$, $\mathcal{D}_x^\perp = \langle \tilde{\mathbf{q}}_i,\mathbf{t}_m^z \rangle$, $\mathcal{D}_z = \langle \tilde{\mathbf{q}}_i,\tilde{\mathbf{t}}_j,\mathbf{s}_l^z \rangle$, $\mathcal{D}_z^\perp = \langle \tilde{\mathbf{q}}_i,\tilde{\mathbf{t}}_j \rangle$, the hybrid subsystem code $\mathcal{Q}_7$ in \cref{thm:QSC_unified} can be obtained.
\end{enumerate}

\section{Discussion and conclusions} \label{sec:discussion}

In this work, we propose a method to construct quantum codes in the synchronizable hybrid subsystem code family in \cref{thm:QSC_unified}, which can be viewed as a generalization of QSCs proposed by Fujiwara in \cite{Fujiwara2013}. We also establish trade-offs between the number of gauge qubits $r$, the number of encoded classical bits $m$, and the maximum synchronization distance $d_\mathrm{sync,max}$ in \cref{cor:QSC_tradeoff}. The sum of these numbers is equal to $2(k_d-k_c)$ for a synchronizable code (where $d_\mathrm{sync,max}>1$), and is equal to $2(k_d-k_c)+1$ for a non-synchronizable code (where $d_\mathrm{sync,max}=1$). We conjecture that the sum of the numbers obtained by our code construction is optimal.

To construct a quantum code in the synchronizable hybrid subsystem code family with good parameters, one may want to choose classical cyclic codes $\mathcal{C}$ and $\mathcal{D}$ satisfying $\mathcal{C}^\perp \subset \mathcal{C} \subset \mathcal{D}$ with large value of $k_d-k_c$. If $k_d$ is fixed, choosing a classical code $\mathcal{C}$ with lower $k_c$ can lead to a quantum code with higher $k_d-k_c$. However, we note that the number of logical qubits of a quantum code is determined by $2k_c-n$, and choosing $\mathcal{C}$ with lower $k_c$ can lead to a quantum code with fewer logical qubits. On the other hand, suppose that the code $\mathcal{C}$ is fixed. Choosing a classical code $\mathcal{D}$ with higher $k_d$ could lead to a quantum code with higher $k_d-k_c$, while the number of logical qubits remains constant. However, a code $\mathcal{D}$ with higher $k_d$ tends to have lower code distance $d_d$, which could lead to a quantum code that corrects fewer Pauli errors. Therefore, it is important to properly choose classical codes $\mathcal{C}$ and $\mathcal{D}$ so that the number of logical qubits, the code distance, and the maximum synchronization distance are well-balanced.

We point out that our construction of synchronizable hybrid subsystem codes is based on the QSC construction developed in \cite{Fujiwara2013}. There are other methods to construct QSCs for qubit systems, such as the ones proposed in \cite{Fujiwara2013-Jul, Fujiwara2013-Nov, Xie2014, Xie2016, Guenda2017}. One possible research direction would be extending our construction to cover such families of QSCs. A construction of QSCs for qudit systems has also been proposed \cite{Luo2018}. Another interesting direction would be generalizing our construction of synchronizable hybrid subsystem codes to qudit systems.

\cref{thm:CSS_subsystem,thm:CSS_hybrid,thm:CSS_hybrid_subsystem} in \cref{sec:hybrid_subsystem_construction} provide other methods to construct subsystem, hybrid, and hybrid subsystem codes of CSS type from general classical codes, which are not necessarily cyclic codes. We hope that these theorems could provide a way to construct quantum codes with more complicated structure, as well as synchronizable hybrid subsystem codes with better parameters.

It should be noted that this work assume that Pauli errors can happen on the data qubits only, and measurement and gate faults are absent (i.e., the code capacity noise model is assumed). That is, our proposed procedures for encoding, error correction, and decoding are not fault-tolerant. To handle measurement and gate faults in the syndrome extraction, one needs to apply syndrome extraction circuits such as the ones in the fault-tolerant error correction (FTEC) schemes proposed by Shor \cite{Shor96,DA07,TPB23}, Steane \cite{Steane96b}, or Knill \cite{Knill05a}. A flag technique \cite{CR17a,CR20} provides another way extract syndromes with fewer ancilla qubits. Since the codes in synchronizable hybrid subsystem codes are CSS codes constructed from classical cyclic codes, one may utilize flag techniques designed specifically for such code families \cite{TCL20,DMLWD24}. A technique to connect several flag circuits together \cite{AM22} can also be applied to further reduce the number of required ancilla qubits.  

The usage of synchronizable hybrid subsystem codes in quantum computation could be protecting quantum information against Pauli and synchronization errors transferred between modules in the modular architectures. However, the main code for quantum computation might be another code with better properties for a given architecture, such as surface code \cite{Kitaev03,BK98,DKLP02}, color codes \cite{BM06}, or good quantum LDPC codes \cite{BE21,HHJO21,EKZ22,LZ22,PK22a,PK22b,DHLV23}. One interesting research direction would be finding an efficient way to convert between the main code for computation and a synchronizable hybrid subsystem code for communication, or constructing a new code from the main code and a synchronizable hybrid subsystem code so that the resulting code inherits properties of the main code and has ability to correct synchronization errors. We leave this for future work.

\section{Acknowledgements}

We thank Keisuke Fujii and members of Center for Quantum Information and Quantum Biology (QIQB) for helpful discussions. T.T. is supported by JST Moonshot R\&D Grant No. JPMJMS2061. A.N. was supported in part by the NSF QLCI program.

\appendix

\section{Proofs of theorems} \label{sec:all_proofs}

In this section, we provide proofs of \cref{thm:gen_decomp,thm:op_pairing,thm:sync_proof_subspace,thm:sync_proof_subsystem,thm:CSS_hybrid,thm:CSS_hybrid_subsystem}.

\subsection{Proof of \cref{thm:gen_decomp}} \label{subsec:gen_decomp_proof}

\begin{proof}
The first representations of the codes $\mathcal{D}^\perp$, $\mathcal{C}^\perp$, $\mathcal{C}$, $\mathcal{D}$ follow the definitions of $\mathbf{p}_i$, $\mathbf{q}_i$, $\tilde{\mathbf{p}}_i$, and $\tilde{\mathbf{q}}_i$ in \cref{subsec:pre_cyclic}, so we will only prove the second representations.

We start by showing that $\mathcal{C}^\perp = \langle \tilde{\mathbf{q}}_i,\tilde{\mathbf{p}}_j \rangle$ where $i \in \{1,\dots,n-k_d\}$ and $j \in \{1,\dots,k_d-k_c\}$. Recall that $\tilde{\mathbf{p}}_{i}=\mathcal{V}\left(x^{{i}-1}\tilde{p}_\mathcal{R}(x)\right)$ where $i \in \{1,\dots,n-k_c\}$ are generators of $\mathcal{C}^\perp$. Therefore, any polynomial $c(x) \in I_{\mathcal{C}^\perp}$ is of the form $g_1(x)\tilde{p}_\mathcal{R}(x)$ where $\mathrm{deg}\left(g_1(x)\right)<n-k_c$. Also, $\tilde{\mathbf{q}}_{i'}=\mathcal{V}\left(x^{{i'}-1}\tilde{q}_\mathcal{R}(x)\right)$ where $i' \in \{1,\dots,n-k_d\}$ are generators of $\mathcal{D}^\perp$. So any polynomial $d(x) \in I_{\mathcal{D}^\perp}$ is of the form $g_2(x)\tilde{q}_\mathcal{R}(x)$ where $\mathrm{deg}\left(g_2(x)\right)<n-k_d$. Since $I_{\mathcal{D}^\perp} \subset I_{\mathcal{C}^\perp}$, we can write $\tilde{q}_\mathcal{R}(x) = f(x)\tilde{p}_\mathcal{R}(x)$ for some polynomial $f(x)$ of degree $k_d-k_c$. 

Consider a polynomial of the form $g_2(x)\tilde{q}_\mathcal{R}(x)+\tilde{p}_\mathcal{R}(x) = \left(g_2(x)f(x)+1\right)\tilde{p}_\mathcal{R}(x)$. Since $\mathrm{deg}\left(\tilde{p}_\mathcal{R}(x)\right)=k_c < k_d = \mathrm{deg}\left(\tilde{q}_\mathcal{R}(x)\right)$, dividing $g_2(x)\tilde{q}_\mathcal{R}(x)+\tilde{p}_\mathcal{R}(x)$ by $\tilde{q}_\mathcal{R}(x)$ gives a remainder $\tilde{p}_\mathcal{R}(x) \neq 0$. That is, $g_2(x)\tilde{q}_\mathcal{R}(x)+\tilde{p}_\mathcal{R}(x)$ is a polynomial in $I_{\mathcal{C}^\perp}$ but not in $I_{\mathcal{D}^\perp}$. Using similar ideas, we can define
\begin{equation}
    \mathcal{A}_\mathbf{w} = \left\{g_2(x)\tilde{q}_\mathcal{R}(x)+\left(\sum_{j=1}^{k_d-k_c}w_{j-1}x^{j-1}\right)\tilde{p}_\mathcal{R}(x)\; \Bigg| \;\mathrm{deg}\left(g_2(x)\right)<n-k_d\right\},
\end{equation}
where $\mathbf{w}=(w_0,\dots,w_{k_d-k_c-1})$ is a binary vector in $\mathbb{Z}_2^{k_d-k_c}$. From this definition, $\mathcal{A}_\mathbf{0}=I_{\mathcal{D}^\perp}$, and $\mathcal{A}_\mathbf{w} \cap \mathcal{A}_\mathbf{w'} = \emptyset$ when $\mathbf{w} \neq \mathbf{w'}$. This comes from the fact that $\{\tilde{\mathbf{p}}_j\}$ (as well as $\{x^{j-1}\tilde{p}_\mathcal{R}(x)\}$) where $j \in \{1,\dots,k_d-k_c\}$ are linearly independent. We also have that 
\begin{equation}
    \bigcup_{\mathbf{w}\in \mathbb{Z}_2^{k_d-k_c}} \mathcal{A}_\mathbf{w} = \{g_2(x)\tilde{q}_\mathcal{R}(x)+g_3(x)\tilde{p}_\mathcal{R}(x)\;|\;\mathrm{deg}\left(g_2(x)\right)<n-k_d,\mathrm{deg}\left(g_3(x)\right)<k_d-k_c\},
\end{equation}
which corresponds to $\langle \tilde{\mathbf{q}}_i,\tilde{\mathbf{p}}_j \rangle$ where $i \in \{1,\dots,n-k_d\}$ and $j \in \{1,\dots,k_d-k_c\}$.

We aim to prove that $I_{\mathcal{C}^\perp} = \bigcup_{\mathbf{w}\in \mathbb{Z}_2^{k_d-k_c}} \mathcal{A}_\mathbf{w}$ by showing that $\bigcup_{\mathbf{w}\in \mathbb{Z}_2^{k_d-k_c}} \mathcal{A}_\mathbf{w} \subseteq I_{\mathcal{C}^\perp}$ and $I_{\mathcal{C}^\perp} \subseteq \bigcup_{\mathbf{w}\in \mathbb{Z}_2^{k_d-k_c}} \mathcal{A}_\mathbf{w}$. First, observe that any polynomial $g_2(x)\tilde{q}_\mathcal{R}(x)+g_3(x)\tilde{p}_\mathcal{R}(x)$ in which $\mathrm{deg}\left(g_2(x)\right)<n-k_d$ and $\mathrm{deg}\left(g_3(x)\right)<k_d-k_c$ can be written as $\left(g_2(x)f(x)+g_3(x)\right)\tilde{p}_\mathcal{R}(x)$. The polynomial $g_2(x)f(x)+g_3(x)$ has degree $< n-k_c$, meaning that $\left(g_2(x)f(x)+g_3(x)\right)\tilde{p}_\mathcal{R}(x)=g_1(x)\tilde{p}_\mathcal{R}(x)$ for some $g_1(x)$ of degree $<n-k_c$. Thus, $\bigcup_{\mathbf{w}\in \mathbb{Z}_2^{k_d-k_c}} \mathcal{A}_\mathbf{w} \subseteq I_{\mathcal{C}^\perp}$.

Next, consider any polynomial $g_1(x)\tilde{p}_\mathcal{R}(x)$ in which $\mathrm{deg}\left(g_1(x)\right)=D<n-k_c$. If $D < k_d-k_c$, it is obvious that $g_1(x)\tilde{p}_\mathcal{R}(x)$ is in $\bigcup_{\mathbf{w}\in \mathbb{Z}_2^{k_d-k_c}} \mathcal{A}_\mathbf{w}$. If $k_d-k_c \leq D \leq n-k_c-1$, we can write $g_1(x)=x^{D-k_d+k_c}f(x)+z_{D-k_d+k_c-1}x^{D-k_d+k_c-1}f(x)+\dots+z_0f(x)+g_3(x)$ for some $(z_0,\dots,z_{D-k_d+k_c-1})\in \mathbb{Z}_2^{(D-k_d+k_c)}$ and for some polynomial $g_3(x)$ of degree $<k_d-k_c$. Hence,
\begin{align}
    g_1(x)\tilde{p}_\mathcal{R}(x) =&x^{D-k_d+k_c}\tilde{q}_\mathcal{R}(x)+z_{D-k_d+k_c-1}x^{D-k_d+k_c-1}\tilde{q}_\mathcal{R}(x)+\dots+z_0\tilde{q}_\mathcal{R}(x) \nonumber \\ &+g_3(x)\tilde{p}_\mathcal{R}(x) \\
    =&g_2(x)\tilde{q}_\mathcal{R}(x)+g_3(x)\tilde{p}_\mathcal{R}(x),
\end{align}
for some $g_2(x)$ of degree $<n-k_d$ and for some $g_3(x)$ of degree $<k_d-k_c$. Therefore, $I_{\mathcal{C}^\perp} \subseteq \bigcup_{\mathbf{w}\in \mathbb{Z}_2^{k_d-k_c}} \mathcal{A}_\mathbf{w}$. Together with the previous argument, we find that $I_{\mathcal{C}^\perp} = \bigcup_{\mathbf{w}\in \mathbb{Z}_2^{(k_d-k_c)}} \mathcal{A}_\mathbf{w}$, or equivalently, $\mathcal{C}^\perp = \langle \tilde{\mathbf{q}}_i,\tilde{\mathbf{p}}_j \rangle$ where $i \in \{1,\dots,n-k_d\}$ and $j \in \{1,\dots,k_d-k_c\}$.

Using similar ideas, we can compare $\mathcal{C}^\perp$ and $\mathcal{C}$ to show that $\mathcal{C} = \langle \tilde{\mathbf{p}}_j,\mathbf{p}_l \rangle$ where $j \in \{1,\dots,n-k_c\}$ and $l \in \{1,\dots,2k_c-n\}$, and compare $\mathcal{C}$ and $\mathcal{D}$ to show that $\mathcal{D} = \langle \mathbf{p}_l,\mathbf{q}_m \rangle$ where $l \in \{1,\dots,k_c\}$ and $m \in \{1,\dots,k_d-k_c\}$. These imply that $\mathcal{C} = \langle \tilde{\mathbf{q}}_i,\tilde{\mathbf{p}}_j,\mathbf{p}_l \rangle$ and $\mathcal{D} = \langle \tilde{\mathbf{q}}_i,\tilde{\mathbf{p}}_j,\mathbf{p}_l,\mathbf{q}_m \rangle$ where $i \in \{1,\dots,n-k_d\}$, $j \in \{1,\dots,k_d-k_c\}$, $l \in \{1,\dots,2k_c-n\}$, and $m \in \{1,\dots,k_d-k_c\}$.
\end{proof}

\subsection{Proof of \cref{thm:op_pairing}} \label{subsec:op_pairing_proof}

\begin{proof}
The main goal of \cref{thm:op_pairing} is to show that we can pick new sets of generators of the classical codes $\mathcal{C}$ and $\mathcal{D}$ that correspond to stabilizer generators, gauge operators and logical operators of the stabilizer codes $\mathcal{Q}_\mathcal{C}$ and $\mathcal{Q}_\mathcal{D}$. Here we adopt Theorem 0.1 in Supplementary Material of \cite{BDH06}, which is also provided in this proof, to construct pairs of anticommuting Pauli operators (see also \cite{Wilde09b} for a similar construction).

We start by introducing the symplectic representation of Pauli operators. Any $n$-qubit Pauli operator in the Pauli group $\mathcal{P}_n$ can be represented by a $2n$-bit row vector $(\mathbf{x}|\mathbf{z}) \in \mathbb{Z}^{2n}_2$ where $\mathbf{x}$ and $\mathbf{z}$ correspond to the $X$-part and the $Z$-part of the Pauli operator (ignoring the phase factor). The symplectic product of $\mathbf{u}=(\mathbf{x}|\mathbf{z})$ and $\mathbf{v}=(\mathbf{x}'|\mathbf{z}')$ is given by,
\begin{equation}
    \mathbf{u} \odot \mathbf{v} = \mathbf{x} \cdot \mathbf{z}' + \mathbf{z} \cdot \mathbf{x}'.
\end{equation} 
$\mathbb{Z}^{2n}_2$ is generated by $n$ \emph{hyperbolic pairs} $(\mathbf{g}_i,\mathbf{h}_i)$ such that only the symplectic product between the hyperbolic partner is nonzero; that is, $\mathbf{g}_i \odot \mathbf{g}_j = 0$ for all $i,j$, $\mathbf{h}_i \odot \mathbf{h}_j = 0$ for all $i,j$, $\mathbf{g}_i \odot \mathbf{h}_i = 1$ for all $i$, and $\mathbf{g}_i \odot \mathbf{h}_j = 0$ for all $i\neq j$, where $i,j \in \{1,\dots,n\}$. These $n$ hyperbolic pairs correspond to $n$ pairs of anticommuting Pauli operators that together with $iI^{\otimes n}$ generate the Pauli group $\mathcal{P}_n$. Any basis of $\mathbb{Z}^{2n}_2$ that consists of $n$ hyperbolic pairs is called \emph{symplectic basis}. 

Consider a subspace $V$ of $\mathbb{Z}^{2n}_2$. $V$ is called \emph{isotropic} if for all $\mathbf{v} \in V$, $\mathbf{v} \odot \mathbf{u} =0$ for all $\mathbf{u} \in V$. $V$ is called \emph{symplectic} if there is no $\mathbf{v} \in V$ such that $\mathbf{v} \odot \mathbf{u} =0$ for all $\mathbf{u} \in V$. If $V$ is neither isotropic nor symplectic, the theorem below (which is Theorem 0.1 in Supplementary Material of \cite{BDH06}) can be applied.
\begin{theorem} \cite{BDH06} \label{thm:0.1}
    Let $V$ be an $r$-dimensional subspace of $\mathbb{Z}_2^{2n}$. There exists a symplectic basis of $\mathbb{Z}_2^{2n}$ consisting hyperbolic pair $(\mathbf{u}_i,\mathbf{v}_i), i\in\{1,\dots,n\}$ such that $\{\mathbf{u}_1,\dots,\mathbf{u}_{p+q},\mathbf{v}_1,\dots,\mathbf{v}_{p}\}$ are basis vectors of $V$, where $2p+q=r$. Equivalently, $V= \mathrm{symp}(V)\oplus\mathrm{iso}(V)$, where $\mathrm{symp}(V)=\mathrm{span}\{\mathbf{u}_1,\dots,\mathbf{u}_p,\mathbf{v}_1,\dots,\mathbf{v}_p\}$ is symplectic and $\mathrm{iso}(V)=\mathrm{span}\{\mathbf{u}_{p+1},\dots,\mathbf{u}_{p+q}\}$ is isotropic.
\end{theorem}

Consider the classical cyclic code $\mathcal{C}$ and let $V$ be the subspace of $(\mathbf{v}|\mathbf{0})$ and $(\mathbf{0}|\mathbf{v})$ for all $\mathbf{v} \in \mathcal{C}$. Since $\mathcal{C} = \langle \tilde{\mathbf{q}}_i,\tilde{\mathbf{p}}_j,\mathbf{p}_l \rangle$ by \cref{thm:gen_decomp}, we can write 
\begin{equation}
    V = \mathrm{span}\{(\tilde{\mathbf{q}}_i|\mathbf{0}),(\tilde{\mathbf{p}}_j|\mathbf{0}),(\mathbf{p}_l|\mathbf{0}),(\mathbf{0}|\tilde{\mathbf{q}}_i),(\mathbf{0}|\tilde{\mathbf{p}}_j),(\mathbf{0}|\mathbf{p}_l)\}.
\end{equation}
The dimension of $V$ is $r=2k_c$. 

Observe that for all $i,j$, $(\tilde{\mathbf{q}}_i|\mathbf{0}) \odot \mathbf{u}= (\mathbf{0}|\tilde{\mathbf{q}}_i) \odot \mathbf{u}=(\tilde{\mathbf{p}}_j|\mathbf{0}) \odot \mathbf{u} = (\mathbf{0}|\tilde{\mathbf{p}}_j) \odot \mathbf{u}=0$ for all $\mathbf{u} \in V$; this comes from the fact that $\mathcal{C}^\perp = \langle \tilde{\mathbf{q}}_i,\tilde{\mathbf{p}}_j\rangle$ is the dual code of $\mathcal{C}$. Also, there is no $\mathbf{p}_l$ such that $(\mathbf{p}_l|\mathbf{0}) \odot \mathbf{u} = 0$ or $ (\mathbf{0}|\mathbf{p}_l) \odot \mathbf{u}=0$ for all $\mathbf{u} \in V$. Therefore, $\mathrm{iso}(V)=\mathrm{span}\{(\tilde{\mathbf{q}}_i|\mathbf{0}),(\tilde{\mathbf{p}}_j|\mathbf{0}),(\mathbf{0}|\tilde{\mathbf{q}}_i),(\mathbf{0}|\tilde{\mathbf{p}}_j)\}$ and $\mathrm{symp}(V)=\mathrm{span}\{(\mathbf{p}_l|\mathbf{0}),(\mathbf{0}|\mathbf{p}_l)\}$. Here we have $p=2k_c-n$ and $q=2(n-k_c)$, which satisfy $2p+q=r$. $\mathrm{iso}(V)$ corresponds to the stabilizer group of $\mathcal{Q}_\mathcal{C}$, and $\mathrm{symp}(V)$ corresponds to the generators of logical Pauli operators of $\mathcal{Q}_\mathcal{C}$. 

In a special case where $\mathbf{p}_l \cdot \mathbf{p}_l=1$ for all $l$ and $\mathbf{p}_l \cdot \mathbf{p}_{l'}=0$ for all $l \neq l'$ where $l,l' \in \{1,\dots,2k_c-n\}$, $\left((\mathbf{p}_l|\mathbf{0}),(\mathbf{0}|\mathbf{p}_l)\right)$ are the hyperbolic pairs that describe $\mathrm{symp}(V)$. However, this is not always true. In general, hyperbolic pairs that describe $\mathrm{symp}(V)$ can be constructed using the algorithm provided in the proof of Theorem 0.1 in Supplementary Material of \cite{BDH06}. In this work, we slightly modify the algorithm so that (1) logical $X$ and logical $Z$ operators that correspond to the hyperbolic pairs describing $\mathrm{symp}(V)$ are purely $X$-type and purely $Z$-type, respectively, and (2) for each logical $X$ operator, its support is the same as the support of some logical $Z$ operator (and vice versa). The algorithm is presented below.

\begin{algorithm} \label{alg:for_logical}
The algorithm starts by letting $\{\mathbf{w}_1,\dots,\mathbf{w}_{2k_c-n}\} = \{\mathbf{p}_1,\dots,\mathbf{p}_{2k_c-n}\}$, $s=2k_c-n$, and $i=1$. At Round $i$, consider $\{\mathbf{w}_1,\dots,\mathbf{w}_s\}$ and find the smallest $j$ such that $\mathbf{w}_j \cdot \mathbf{w}_j = 1$.
\begin{enumerate}[leftmargin=1cm] 
    \item If the smallest $j$ such that $\mathbf{w}_j \cdot \mathbf{w}_j = 1$ can be found, then do the following.
    \begin{enumerate}
        \item Assign $\mathbf{s}_i^x:=\mathbf{w}_j$ and $\mathbf{s}_i^z:=\mathbf{w}_j$.
        \item For $k=1,\dots,j-1$, assign $\mathbf{w}'_{k}:=\mathbf{w}_k+(\mathbf{w}_k\cdot \mathbf{w}_j)\mathbf{w}_j$. \\
        For $k=j+1,\dots,s$, assign $\mathbf{w}'_{k-1}:=\mathbf{w}_k+(\mathbf{w}_k\cdot \mathbf{w}_j)\mathbf{w}_j$\\
        (This is to ensure that $\mathbf{w}'_k \cdot \mathbf{w}_j=0$ for all $k=1,\dots,s-1$.)
        \item Assign $s := s-1$, $i:=i+1$, and $\{\mathbf{w}_1,\dots,\mathbf{w}_{s}\}:= \{\mathbf{w}'_1,\dots,\mathbf{w}'_{s}\}$, then continue to the next round.
    \end{enumerate}
    \item If the smallest $j$ such that $\mathbf{w}_j \cdot \mathbf{w}_j = 1$ cannot be found, then do the following.
    \begin{enumerate}
        \item Find the smallest $j$ such that $\mathbf{w}_1 \cdot \mathbf{w}_j = 1$. 
        \item Assign $\mathbf{s}_i^x:=\mathbf{w}_1$, $\mathbf{s}_i^z:=\mathbf{w}_j$, $\mathbf{s}_{i+1}^x:=\mathbf{w}_j$, and $\mathbf{s}_{i+1}^z:=\mathbf{w}_1$.
        \item For $k=2,\dots,j-1$, assign $\mathbf{w}'_{k-1}:=\mathbf{w}_k+(\mathbf{w}_k\cdot \mathbf{w}_j)\mathbf{w}_1$. \\
            For $k=j+1,\dots,s$, assign $\mathbf{w}'_{k-2}:=\mathbf{w}_k+(\mathbf{w}_k\cdot \mathbf{w}_1)\mathbf{w}_j+(\mathbf{w}_k\cdot \mathbf{w}_j)\mathbf{w}_1$.\\
            (This is to ensure that $\mathbf{w}'_k \cdot \mathbf{w}_1=\mathbf{w}'_k \cdot \mathbf{w}_j=0$ for all $k=1,\dots,s-2$.)
        \item Assign $s := s-2$, $i:=i+2$, and $\{\mathbf{w}_1,\dots,\mathbf{w}_{s}\}:= \{\mathbf{w}'_1,\dots,\mathbf{w}'_{s}\}$, then continue to the next round.
    \end{enumerate}
\end{enumerate}
After the loop terminates at $s=0$, output $\{\mathbf{s}_l^x\}$ and $\{\mathbf{s}_l^z\}$ where $l=1,\dots,2k_c-n$.
\end{algorithm}

By the construction of \cref{alg:for_logical}, we have $\langle \mathbf{s}_l^x \rangle = \langle \mathbf{s}_l^z \rangle = \langle \mathbf{p}_l \rangle$. Also, note that there is no $\mathbf{w}_j$ such that $\mathbf{w}_j \cdot \mathbf{w}_k =0$ for all $k$; otherwise, such a vector belongs to $\mathcal{C}^\perp$. Therefore, \cref{alg:for_logical} outputs $2k_c-n$ pairs of $(\mathbf{s}_l^x,\mathbf{s}_l^z)$ such that $\mathbf{s}_l^x \cdot \mathbf{s}_{l}^z=1$ for all $l$, and $\mathbf{s}_l^x \cdot \mathbf{s}_{l'}^z=\mathbf{s}_l^z \cdot \mathbf{s}_{l'}^x=0$ for all $l'\neq l$. The hyperbolic pairs that describe $\mathrm{symp}(V)$ are $\left((\mathbf{s}_l^x|\mathbf{0}),(\mathbf{0}|\mathbf{s}_l^z)\right)$, or equivalently, the anticommuting pairs $\{X(\mathbf{s}_l^x),Z(\mathbf{s}_l^z)\}$ define logical Pauli operators of $\mathcal{Q}_\mathcal{C}$.

Next, we define $\mathbf{q}'_m=\mathbf{q}_m+\sum_l (\mathbf{q}_m \cdot \mathbf{s}_l^x)\mathbf{s}_{l}^z$ (which is equal to $\mathbf{q}_m+\sum_l (\mathbf{q}_m \cdot \mathbf{s}_l^z)\mathbf{s}_{l}^x$ since there is one-to-one correspondence between $\mathbf{s}_l^x$ and $\mathbf{s}_{l'}^z$ such that $\mathbf{s}_l^x=\mathbf{s}_{l'}^z$ where $l$ is not necessarily equal to $l'$). We have that $\langle\mathbf{s}_l^x,\mathbf{q}_m\rangle = \langle\mathbf{s}_l^z,\mathbf{q}_m\rangle=\langle\mathbf{s}_l^x,\mathbf{q}'_m\rangle=\langle\mathbf{s}_l^z,\mathbf{q}'_m\rangle$ and $\mathbf{s}_l^x\cdot \mathbf{q}'_m=\mathbf{s}_l^z\cdot \mathbf{q}'_m=0$ for all $l,m$. Now, let us consider the classical cyclic code $\mathcal{D}$ and let $V$ be the subspace of $(\mathbf{v}|\mathbf{0})$ and $(\mathbf{0}|\mathbf{v})$ for all $\mathbf{v} \in \mathcal{D}$. Using the facts that $\mathcal{D} = \langle \tilde{\mathbf{q}}_i,\tilde{\mathbf{p}}_j,\mathbf{p}_l,\mathbf{q}_m \rangle$ (by \cref{thm:gen_decomp}), $\langle \mathbf{s}_l^x \rangle = \langle \mathbf{s}_l^z \rangle = \langle \mathbf{p}_l \rangle$, and the expressions of $\mathbf{s}_l^x$, $\mathbf{s}_l^z$, $\mathbf{q}_m$, and $\mathbf{q}'_m$, we can write $V = \mathrm{span}\{(\tilde{\mathbf{q}}_i|\mathbf{0}),(\tilde{\mathbf{p}}_j|\mathbf{0}),(\mathbf{s}_l^x|\mathbf{0}),(\mathbf{q}'_m|\mathbf{0}),(\mathbf{0}|\tilde{\mathbf{q}}_i),(\mathbf{0}|\tilde{\mathbf{p}}_j),(\mathbf{0}|\mathbf{s}_l^z),(\mathbf{0}|\mathbf{q}'_m)\}$. The dimension of $V$ is $r=2k_d$. 

Observe that for all $i$, $(\tilde{\mathbf{q}}_i|\mathbf{0}) \odot \mathbf{u}= (\mathbf{0}|\tilde{\mathbf{q}}_i) \odot \mathbf{u}=0$ for all $\mathbf{u} \in V$, which comes from the fact that $\mathcal{D}^\perp = \langle \tilde{\mathbf{q}}_i\rangle$ is the dual code of $\mathcal{D}$. Thus, $\mathrm{iso}(V)=\mathrm{span}\{(\tilde{\mathbf{q}}_i|\mathbf{0}),(\mathbf{0}|\tilde{\mathbf{q}}_i)\}$. Also, there is no $\tilde{\mathbf{p}}_j$ such that $(\tilde{\mathbf{p}}_j|\mathbf{0}) \odot \mathbf{u} = 0$ or $ (\mathbf{0}|\tilde{\mathbf{p}}_j) \odot \mathbf{u}=0$ for all $\mathbf{u} \in V$ (otherwise, $\tilde{\mathbf{p}}_j$ is in $\mathcal{D}^\perp$, which is not true from \cref{thm:gen_decomp}). Similar arguments can be applied to $\mathbf{s}_l^x$, $\mathbf{s}_l^z$, and $\mathbf{q}'_m$. Therefore, $\mathrm{symp}(V)=\mathrm{span}\{(\tilde{\mathbf{p}}_j|\mathbf{0}),(\mathbf{s}_l^x|\mathbf{0}),(\mathbf{q}'_m|\mathbf{0}),(\mathbf{0}|\tilde{\mathbf{p}}_j),(\mathbf{0}|\mathbf{s}_l^z),(\mathbf{0}|\mathbf{q}'_m)\}$. Here we have $p=2k_d-n$ and $q=2(n-k_d)$, which satisfy $2p+q=r$. $\mathrm{iso}(V)$ corresponds to the stabilizer group of $\mathcal{Q}_\mathcal{D}$, and $\mathrm{symp}(V)$ corresponds to the generators of logical Pauli operators of $\mathcal{Q}_\mathcal{D}$. 

We aim to find hyperbolic pairs that describe $\mathrm{symp}(V)$. We already found the hyperbolic pairs $\left((\mathbf{s}_l^x|\mathbf{0}),(\mathbf{0}|\mathbf{s}_l^z)\right)$ from the previous part of the proof and want to find the rest. Again, we modify the algorithm in the proof of Theorem 0.1 in Supplementary Material of \cite{BDH06}. This time we modify it so that (1) anticommuting pairs of Pauli operators corresponding to the hyperbolic pairs describing $\mathrm{symp}(V)$ are purely $X$-type or purely $Z$-type, and (2) $\langle \tilde{\mathbf{p}}_j \rangle$ can be generated by a subset of binary vectors that describe $X$-type (or $Z$-type) Pauli operators in the anticommuting pairs. Note that the sets of $X$-type and $Z$-type Pauli operators might not be the same. The algorithm is presented below.

\begin{algorithm} \label{alg:for_gauge} 
We start by letting $\{\mathbf{w}_1,\dots,\mathbf{w}_{2(k_d-k_c)}\} = \{\tilde{\mathbf{p}}_1,\dots,\tilde{\mathbf{p}}_{k_d-k_c},\mathbf{q}'_1,\dots,\mathbf{q}'_{k_d-k_c}\}$, $s=2(k_d-k_c)$, and $i=1$. At Round $i$, consider $\{\mathbf{w}_1,\dots,\mathbf{w}_s\}$ and find the smallest $j$ such that $\mathbf{w}_1 \cdot \mathbf{w}_j = 1$.
\begin{enumerate}[leftmargin=1cm] 
    \item If $\mathbf{w}_j \cdot \mathbf{w}_j = 0$, do the following.
    \begin{enumerate}
        \item Assign $\tilde{\mathbf{t}}_i:=\mathbf{w}_1$, $\mathbf{t}_i^x:=\mathbf{w}_j$, and $\mathbf{t}_i^z:=\mathbf{w}_j$.
        \item For $k=2,\dots,j-1$, assign $\mathbf{w}'_{k-1}:=\mathbf{w}_k+(\mathbf{w}_k\cdot \mathbf{w}_j)\mathbf{w}_1$. \\
    For $k=j+1,\dots,s$, assign $\mathbf{w}'_{k-2}:=\mathbf{w}_k+(\mathbf{w}_k\cdot \mathbf{w}_1)\mathbf{w}_j+(\mathbf{w}_k\cdot \mathbf{w}_j)\mathbf{w}_1$.\\
        (This is to ensure that $\mathbf{w}'_k \cdot \mathbf{w}_1=\mathbf{w}'_k \cdot \mathbf{w}_j=0$ for all $k=1,\dots,s-2$.)
        \item Assign $s := s-2$, $i:=i+1$, and $\{\mathbf{w}_1,\dots,\mathbf{w}_{s}\}= \{\mathbf{w}'_1,\dots,\mathbf{w}'_{s}\}$, then continue to the next round.
    \end{enumerate}
    \item If $\mathbf{w}_j \cdot \mathbf{w}_j = 1$, do the following.
    \begin{enumerate}
        \item Assign $\tilde{\mathbf{t}}_i:=\mathbf{w}_1$, $\mathbf{t}_i^x:=\mathbf{w}_j$, and $\mathbf{t}_i^z:=\mathbf{w}_1+\mathbf{w}_j$
        \item For $k=2,\dots,j-1$, assign $\mathbf{w}'_{k-1}:=\mathbf{w}_k+(\mathbf{w}_k\cdot \mathbf{w}_j)\mathbf{w}_1+(\mathbf{w}_k\cdot \mathbf{w}_1)(\mathbf{w}_1+\mathbf{w}_j)$. \\
        For $k=j+1,\dots,s$, assign $\mathbf{w}'_{k-2}:=\mathbf{w}_k+(\mathbf{w}_k\cdot \mathbf{w}_j)\mathbf{w}_1+(\mathbf{w}_k\cdot \mathbf{w}_1)(\mathbf{w}_1+\mathbf{w}_j)$. \\
        (This is to ensure that $\mathbf{w}'_k \cdot \mathbf{w}_1=\mathbf{w}'_k \cdot \mathbf{w}_j=\mathbf{w}'_k \cdot (\mathbf{w}_1+\mathbf{w}_j)=0$ for all $k=1,\dots,s-2$.)
        \item Assign $s := s-2$, $i:=i+1$, and $\{\mathbf{w}_1,\dots,\mathbf{w}_{s}\}= \{\mathbf{w}'_1,\dots,\mathbf{w}'_{s}\}$, then continue to the next round.
    \end{enumerate}
\end{enumerate}
After the loop terminates at $s=0$, output $\{\tilde{\mathbf{t}}_j\}$, $\{\mathbf{t}_j^x\}$, and $\{\mathbf{t}_j^z\}$ where $j=1,\dots,k_d-k_c$.
\end{algorithm}

Consider the first round where $\mathbf{w}_1=\tilde{\mathbf{p}}_1$. $\mathbf{w}_j$ with the smallest $j$ such that $\mathbf{w}_1 \cdot \mathbf{w}_j = 1$ is some $\mathbf{q}'_{m'}$. By the vector modification rules, $\mathbf{w}_k$ that corresponds to some $\tilde{\mathbf{p}}_j$ is modified by $\tilde{\mathbf{p}}_1$ only, while $\mathbf{w}_k$ that corresponds to some $\mathbf{q}'_m$ is modified by $\tilde{\mathbf{p}}_1$ and $\mathbf{q}'_{m'}$. Similarly, we find that at any round, $\mathbf{w}_1$ is a linear combination of $\{\tilde{\mathbf{p}}_j\}$, and $\mathbf{w}_j$ is a linear combination of $\{\tilde{\mathbf{p}}_j,\mathbf{q}'_m\}$. At the end, we have that $\langle \tilde{\mathbf{t}}_j \rangle = \langle \tilde{\mathbf{p}}_j \rangle$ where $j \in \{1,\dots,k_d-k_c\}$, and $\langle\tilde{\mathbf{t}}_j,\mathbf{t}_m^x\rangle = \langle\tilde{\mathbf{t}}_j,\mathbf{t}_m^z\rangle = \langle\tilde{\mathbf{p}}_j,\mathbf{q}'_m\rangle$ where $j,m \in \{1,\dots,k_d-k_c\}$. The definition of $\mathbf{q}'_m$ ensures that $\mathbf{s}_l^x \cdot \mathbf{t}_{j}^z = \mathbf{s}_l^z \cdot \mathbf{t}_{j}^x =0$ for all $l,j$. The assignments of $\tilde{\mathbf{t}}_i$, $\mathbf{t}_i^x$, and $\mathbf{t}_i^z$ ensure that $\tilde{\mathbf{t}}_{j} \cdot \mathbf{t}_{j}^x=\tilde{\mathbf{t}}_{j} \cdot \mathbf{t}_{j}^z=1$ for all $j$, $\tilde{\mathbf{t}}_{j} \cdot \mathbf{t}_{j'}^x=\tilde{\mathbf{t}}_{j} \cdot \mathbf{t}_{j'}^z=0$ for all $j'\neq j$, and $\mathbf{t}_{j}^x \cdot \mathbf{t}_{j'}^z=0$ for all $j,j'$.

The hyperbolic pairs that describe $\mathrm{symp}(V)$ are $\left((\mathbf{s}_l^x|\mathbf{0}),(\mathbf{0}|\mathbf{s}_l^z)\right)$, $\left((\tilde{\mathbf{t}}_j|\mathbf{0}),(\mathbf{0}|\mathbf{t}_j^z)\right)$, and $\left((\mathbf{t}_j^x|\mathbf{0}),(\mathbf{0}|\tilde{\mathbf{t}}_j)\right)$. Equivalently, the anticommuting pairs $\{X(\mathbf{s}_l^x),Z(\mathbf{s}_l^z)\}$, $\{X(\tilde{\mathbf{t}}_j),Z(\mathbf{t}_j^z)\}$, and $\{X(\mathbf{t}_j^x),Z(\tilde{\mathbf{t}}_j)\}$ define logical Pauli operators of $\mathcal{Q}_\mathcal{D}$. As $\mathcal{D}^\perp \subset \mathcal{C}^\perp$ and thus $\mathcal{Q}_\mathcal{D} \subset \mathcal{Q}_\mathcal{C}$, $\{X(\tilde{\mathbf{t}}_j),Z(\mathbf{t}_j^z)\}$ and $\{X(\mathbf{t}_j^x),Z(\tilde{\mathbf{t}}_j)\}$ which are logical Pauli operators of $\mathcal{Q}_\mathcal{D}$ but not logical Pauli operators of $\mathcal{Q}_\mathcal{C}$ define gauge operators of $\mathcal{Q}_\mathcal{C}$.

The classical codes $\mathcal{C}$ and $\mathcal{D}$ can also be represented by $\tilde{\mathbf{t}}_j,\mathbf{s}_l^x,\mathbf{s}_l^z,\mathbf{t}_m^x,\mathbf{t}_m^z$ as $\mathcal{C}^\perp = \langle \tilde{\mathbf{q}}_i,\tilde{\mathbf{t}}_j \rangle$, $\mathcal{C} = \langle \tilde{\mathbf{q}}_i,\tilde{\mathbf{t}}_j,\mathbf{s}_l^x \rangle = \langle \tilde{\mathbf{q}}_i,\tilde{\mathbf{t}}_j,\mathbf{s}_l^z \rangle$, and $\mathcal{D} = \langle \tilde{\mathbf{q}}_i,\tilde{\mathbf{t}}_j,\mathbf{s}_l^x,\mathbf{t}_m^x \rangle = \langle \tilde{\mathbf{q}}_i,\tilde{\mathbf{t}}_j,\mathbf{s}_l^z,\mathbf{t}_m^z \rangle$. Using these representations and the fact that $\mathcal{D}^\perp$ (or $\mathcal{C}^\perp$) is the dual code of $\mathcal{D}$ (or $\mathcal{C}$), we find that $\tilde{\mathbf{q}}_i \cdot \tilde{\mathbf{q}}_{i'} = \tilde{\mathbf{q}}_i \cdot \tilde{\mathbf{t}}_{j} = \tilde{\mathbf{q}}_i \cdot \mathbf{s}_{l}^x = \tilde{\mathbf{q}}_i \cdot \mathbf{s}_{l}^z = \tilde{\mathbf{q}}_i \cdot \mathbf{t}_{m}^x = \tilde{\mathbf{q}}_i \cdot \mathbf{t}_{m}^z =0$ and
$\tilde{\mathbf{t}}_j \cdot \tilde{\mathbf{t}}_{j'} = \tilde{\mathbf{t}}_j \cdot \mathbf{s}_{l}^x=\tilde{\mathbf{t}}_j \cdot \mathbf{s}_{l}^z=0$
for all $i,i',j,j',l,m$. This completes the proof.
\end{proof}

\subsection{Proof of \cref{thm:sync_proof_subspace}} \label{subsec:sync_subspace_proof}

\begin{proof}
We will consider $X$-type and $Z$-type operators separately. First, observe that the $X$-type stabilizer generators $\{X(\overbracket[0.5pt]{\underbracket[0.5pt]{\,\tilde{\mathbf{p}}_j\,}_{\mathclap{a_l}}}^{\text{last}}|\underbracket[0.5pt]{\,\tilde{\mathbf{p}}_j\,}_{\mathclap{n}}|\overbracket[0.5pt]{\underbracket[0.5pt]{\,\tilde{\mathbf{p}}_j\,}_{\mathclap{a_r}}}^{\text{first}})(-1)^{\tilde{\mathbf{p}}_j\cdot \mathbf{v}}\}$ generate the set, 
\begin{equation}
    \mathcal{S}_1^x=\{X(\overbracket[0.5pt]{\underbracket[0.5pt]{\,\mathbf{c}\,}_{\mathclap{a_l}}}^{\text{last}}|\underbracket[0.5pt]{\,\mathbf{c}\,}_{\mathclap{n}}|\overbracket[0.5pt]{\underbracket[0.5pt]{\,\mathbf{c}\,}_{\mathclap{a_r}}}^{\text{first}})(-1)^{\mathbf{c}\cdot \mathbf{v},}\;|\;\mathbf{c} \in \mathcal{C}^\perp\}.
\end{equation}
For each element of $\mathcal{S}_1^x$, we have that
\begin{equation}
    X(\overbracket[0.5pt]{\underbracket[0.5pt]{\,\mathbf{c}\,}_{\mathclap{a_l}}}^{\text{last}}|\underbracket[0.5pt]{\,\mathbf{c}\,}_{\mathclap{n}}|\overbracket[0.5pt]{\underbracket[0.5pt]{\,\mathbf{c}\,}_{\mathclap{a_r}}}^{\text{first}})(-1)^{\mathbf{c}\cdot \mathbf{v}}
    =
    X(\overbracket[0.5pt]{\underbracket[0.5pt]{\,\mathcal{O}(\mathbf{c},-\alpha)\,}_{\mathclap{a_l+\alpha}}}^{\text{last}}|\underbracket[0.5pt]{\,\mathcal{O}(\mathbf{c},-\alpha)\,}_{\mathclap{n}}|\overbracket[0.5pt]{\underbracket[0.5pt]{\,\mathcal{O}(\mathbf{c},-\alpha)\,}_{\mathclap{a_r-\alpha}}}^{\text{first}})(-1)^{\mathcal{O}(\mathbf{c},-\alpha)\cdot \mathcal{O}(\mathbf{v},-\alpha)}.
\end{equation}
Therefore, $\mathcal{S}_1^x=\{X(\overbracket[0.5pt]{\underbracket[0.5pt]{\,\mathcal{O}(\mathbf{c},-\alpha)\,}_{\mathclap{a_l+\alpha}}}^{\text{last}}|\underbracket[0.5pt]{\,\mathcal{O}(\mathbf{c},-\alpha)\,}_{\mathclap{n}}|\overbracket[0.5pt]{\underbracket[0.5pt]{\,\mathcal{O}(\mathbf{c},-\alpha)\,}_{\mathclap{a_r-\alpha}}}^{\text{first}})(-1)^{\mathcal{O}(\mathbf{c},-\alpha)\cdot \mathcal{O}(\mathbf{v},-\alpha)}\;|\;\mathbf{c} \in \mathcal{C}^\perp\}$.

Next, consider a set $\mathcal{S}_2^x$ generated by $\{X(\overbracket[0.5pt]{\underbracket[0.5pt]{\,\tilde{\mathbf{p}}_j\,}_{\mathclap{a_l+\alpha}}}^{\text{last}}|\underbracket[0.5pt]{\,\tilde{\mathbf{p}}_j\,}_{\mathclap{n}}|\overbracket[0.5pt]{\underbracket[0.5pt]{\,\tilde{\mathbf{p}}_j\,}_{\mathclap{a_r-\alpha}}}^{\text{first}})(-1)^{\tilde{\mathbf{p}}_j\cdot \mathcal{O}\left(\mathbf{v},-\alpha\right)}\}$ for some integer $\alpha$. $\mathcal{S}_2^x$ is of the form,
\begin{align}
    \mathcal{S}_2^x&=\{X(\overbracket[0.5pt]{\underbracket[0.5pt]{\,\mathbf{c}'\,}_{\mathclap{a_l+\alpha}}}^{\text{last}}|\underbracket[0.5pt]{\,\mathbf{c}'\,}_{\mathclap{n}}|\overbracket[0.5pt]{\underbracket[0.5pt]{\,\mathbf{c}'\,}_{\mathclap{a_r-\alpha}}}^{\text{first}})(-1)^{\mathbf{c}'\cdot \mathcal{O}\left(\mathbf{v},-\alpha\right)}\;|\;\mathbf{c}' \in \mathcal{C}^\perp\}  \\
    &=\{X(\overbracket[0.5pt]{\underbracket[0.5pt]{\,\mathcal{O}(\mathbf{c},-\alpha)\,}_{\mathclap{a_l+\alpha}}}^{\text{last}}|\underbracket[0.5pt]{\,\mathcal{O}(\mathbf{c},-\alpha)\,}_{\mathclap{n}}|\overbracket[0.5pt]{\underbracket[0.5pt]{\,\mathcal{O}(\mathbf{c},-\alpha)\,}_{\mathclap{a_r-\alpha}}}^{\text{first}})(-1)^{\mathcal{O}(\mathbf{c},-\alpha)\cdot \mathcal{O}\left(\mathbf{v},-\alpha\right)}\;|\;\mathcal{O}(\mathbf{c},-\alpha) \in \mathcal{C}^\perp\}.
\end{align}

$\mathcal{S}_1^x$ and $\mathcal{S}_2^x$ are the same set since $\{\mathcal{O}\left(\mathbf{c},-\alpha\right) \in \mathcal{C}^\perp\} = \mathcal{O}\left(\mathcal{C}^\perp,\alpha\right) = \mathcal{C}^\perp$, where $\mathcal{O}\left(\mathcal{C}^\perp,\alpha\right) = \{\mathcal{O}\left(\mathbf{c},\alpha\right)\;|\;\mathbf{c}\in \mathcal{C}^\perp\}$.  Therefore, for any integer $\alpha$, $\{X(\overbracket[0.5pt]{\underbracket[0.5pt]{\,\tilde{\mathbf{p}}_j\,}_{\mathclap{a_l+\alpha}}}^{\text{last}}|\underbracket[0.5pt]{\,\tilde{\mathbf{p}}_j\,}_{\mathclap{n}}|\overbracket[0.5pt]{\underbracket[0.5pt]{\,\tilde{\mathbf{p}}_j\,}_{\mathclap{a_r-\alpha}}}^{\text{first}})(-1)^{\tilde{\mathbf{p}}_j\cdot \mathcal{O}\left(\mathbf{v},-\alpha\right)}\}$ generates the same group of $X$-type stabilizers.

Next, let us consider the $Z$-type stabilizer generators that correspond to the $a_l+a_r$ ancilla qubits, 
\begin{equation}
    Z\left( \begin{array}{c|ccc|c}
    \mathbb{1}_{a_l} & \mathbb{0}_{a_l,a_r} & \mathbb{0}_{a_l,(n-a_l-a_r)} & \mathbb{1}_{a_l} & \mathbb{0}_{a_l,a_r} \\
    \mathbb{0}_{a_r,a_l} & \mathbb{1}_{a_r} & \mathbb{0}_{a_r,(n-a_l-a_r)} & \mathbb{0}_{a_r,a_l} & \mathbb{1}_{a_r}
    \end{array}
    \right). \label{eq:stab_ancilla}
\end{equation}
We can multiply a $Z$-type generator $Z(\underbracket[0.5pt]{\vphantom{\tilde{\mathbf{p}}_j}\,\mathbf{0}}_{\mathclap{a_l}}|\underbracket[0.5pt]{\tilde{\mathbf{p}}_j}_{\mathclap{n}}|\underbracket[0.5pt]{\vphantom{\tilde{\mathbf{p}}_j}\mathbf{0}\,}_{\mathclap{a_r}})(-1)^{\tilde{\mathbf{p}}_j \cdot \mathbf{w}}$ with some generators from \cref{eq:stab_ancilla} to show that it is logically equivalent to the following.
\begin{align}
    Z(\underbracket[0.5pt]{\vphantom{\tilde{\mathbf{p}}_j}\,\mathbf{0}\,}_{\mathclap{a_l}}|\underbracket[0.5pt]{\,\tilde{\mathbf{p}}_j\,}_{\mathclap{n}}|\underbracket[0.5pt]{\vphantom{\tilde{\mathbf{p}}_j}\,\mathbf{0}\,}_{\mathclap{a_r}})(-1)^{\tilde{\mathbf{p}}_j \cdot \mathbf{w}} & \simeq
    Z(\underbracket[0.5pt]{\vphantom{\tilde{\mathbf{p}}_j}\,\mathbf{0}\,}_{\mathclap{a_l+\alpha}}|\underbracket[0.5pt]{\,\mathcal{O}(\tilde{\mathbf{p}}_j,-\alpha)\,}_{\mathclap{n}}|\underbracket[0.5pt]{\vphantom{\tilde{\mathbf{p}}_j}\,\mathbf{0}\,}_{\mathclap{a_r-\alpha}})(-1)^{\tilde{\mathbf{p}}_j \cdot \mathbf{w}} \\
    & =
    Z(\underbracket[0.5pt]{\vphantom{\tilde{\mathbf{p}}_j}\,\mathbf{0}\,}_{\mathclap{a_l+\alpha}}|\underbracket[0.5pt]{\,\mathcal{O}(\tilde{\mathbf{p}}_j,-\alpha)\,}_{\mathclap{n}}|\underbracket[0.5pt]{\vphantom{\tilde{\mathbf{p}}_j}\,\mathbf{0}\,}_{\mathclap{a_r-\alpha}})(-1)^{\mathcal{O}(\tilde{\mathbf{p}}_j,-\alpha) \cdot \mathcal{O}(\mathbf{w},-\alpha)}.
\end{align}

Consider a set $\mathcal{S}_1^z$ generated by 
$\{Z(\underbracket[0.5pt]{\vphantom{\tilde{\mathbf{p}}_j}\,\mathbf{0}\,}_{\mathclap{a_l+\alpha}}|\underbracket[0.5pt]{\,\mathcal{O}(\tilde{\mathbf{p}}_j,-\alpha)\,}_{\mathclap{n}}|\underbracket[0.5pt]{\vphantom{\tilde{\mathbf{p}}_j}\,\mathbf{0}\,}_{\mathclap{a_r-\alpha}})(-1)^{\mathcal{O}(\tilde{\mathbf{p}}_j,-\alpha) \cdot \mathcal{O}(\mathbf{w},-\alpha)}\}$. We find that,
\begin{align}
\mathcal{S}_1^z&=\{Z(\underbracket[0.5pt]{\vphantom{\mathbf{c}}\,\mathbf{0}\,}_{\mathclap{a_l+\alpha}}|\underbracket[0.5pt]{\,\mathcal{O}(\mathbf{c}',-\alpha)\,}_{\mathclap{n}}|\underbracket[0.5pt]{\vphantom{\mathbf{c}}\,\mathbf{0}\,}_{\mathclap{a_r-\alpha}})(-1)^{\mathcal{O}(\mathbf{c}',-\alpha) \cdot \mathcal{O}(\mathbf{w},-\alpha)}\}\;|\;\mathbf{c}' \in \mathcal{C}^\perp\}  \\
&=\{Z(\underbracket[0.5pt]{\vphantom{\mathbf{c}}\,\mathbf{0}\,}_{\mathclap{a_l+\alpha}}|\underbracket[0.5pt]{\,\mathbf{c}\,}_{\mathclap{n}}|\underbracket[0.5pt]{\vphantom{\mathbf{c}}\,\mathbf{0}\,}_{\mathclap{a_r-\alpha}})(-1)^{\mathbf{c} \cdot \mathcal{O}(\mathbf{w},-\alpha)}\}\;|\;\mathcal{O}(\mathbf{c},\alpha) \in \mathcal{C}^\perp\}.
\end{align}

Next, consider a set $\mathcal{S}_2^z$ generated by $\{Z(\underbracket[0.5pt]{\vphantom{\tilde{\mathbf{p}}_j}\,\mathbf{0}\,}_{\mathclap{a_l+\alpha}}|\underbracket[0.5pt]{\,\tilde{\mathbf{p}}_j\,}_{\mathclap{n}}|\underbracket[0.5pt]{\vphantom{\tilde{\mathbf{p}}_j}\,\mathbf{0}\,}_{\mathclap{a_r-\alpha}})(-1)^{\tilde{\mathbf{p}}_j \cdot \mathcal{O}(\mathbf{w},-\alpha)}\}$, which is the set,
\begin{align}
\mathcal{S}_2^z&=\{Z(\underbracket[0.5pt]{\vphantom{\mathbf{c}}\,\mathbf{0}\,}_{\mathclap{a_l+\alpha}}|\underbracket[0.5pt]{\,\mathbf{c}'\,}_{\mathclap{n}}|\underbracket[0.5pt]{\vphantom{\mathbf{c}}\,\mathbf{0}\,}_{\mathclap{a_r-\alpha}})(-1)^{\mathbf{c}' \cdot \mathcal{O}(\mathbf{w},-\alpha)}\}\;|\;\mathbf{c}' \in \mathcal{C}^\perp\}.
\end{align}

$\mathcal{S}_1^z$ and $\mathcal{S}_2^z$ are the same set since $\{\mathcal{O}\left(\mathbf{c},\alpha\right) \in \mathcal{C}^\perp\} \!=\! \mathcal{O}\left(\mathcal{C}^\perp,-\alpha\right) \!=\! \mathcal{C}^\perp$, where $\mathcal{O}\left(\mathcal{C}^\perp,-\alpha\right) \!=\! \{\mathcal{O}\left(\mathbf{c},-\alpha\right)\;|\;\mathbf{c}\in \mathcal{C}^\perp\}$. Therefore, for any integer $\alpha$, $\{Z(\underbracket[0.5pt]{\vphantom{\tilde{\mathbf{p}}_j}\,\mathbf{0}\,}_{\mathclap{a_l+\alpha}}|\underbracket[0.5pt]{\,\tilde{\mathbf{p}}_j\,}_{\mathclap{n}}|\underbracket[0.5pt]{\vphantom{\tilde{\mathbf{p}}_j}\,\mathbf{0}\,}_{\mathclap{a_r-\alpha}})(-1)^{\tilde{\mathbf{p}}_j \cdot \mathcal{O}(\mathbf{w},-\alpha)}\}$ together with the generators in \cref{eq:stab_ancilla} generate the same group of $Z$-type stabilizers.
\end{proof}

\subsection{Proof of \cref{thm:sync_proof_subsystem}} \label{subsec:sync_subsystem_proof}

\begin{proof}
To show that for any integer $\alpha$, $\{X(\overbracket[0.5pt]{\underbracket[0.5pt]{\,\tilde{\mathbf{q}}_i\,}_{\mathclap{a_l+\alpha}}}^{\text{last}}|\underbracket[0.5pt]{\,\tilde{\mathbf{q}}_i\,}_{\mathclap{n}}|\overbracket[0.5pt]{\underbracket[0.5pt]{\,\tilde{\mathbf{q}}_i\,}_{\mathclap{a_r-\alpha}}}^{\text{first}})\}$ generates the same group of $X$-type stabilizers, we use the fact that $\mathcal{D}^\perp=\langle \tilde{\mathbf{q}}_i \rangle$ is a cyclic code and follow the proof of \cref{thm:sync_proof_subspace}.

To show that for any integer $\alpha$,
\begin{align}
    &Z(\underbracket[0.5pt]{\vphantom{\tilde{\mathbf{p}}_j}\,\mathbf{0}\,}_{\mathclap{a_l+\alpha}}|\underbracket[0.5pt]{\,\tilde{\mathbf{p}}_j\,}_{\mathclap{n}}|\underbracket[0.5pt]{\vphantom{\tilde{\mathbf{p}}_j}\,\mathbf{0}\,}_{\mathclap{a_r-\alpha}})(-1)^{\tilde{\mathbf{p}}_j\cdot \mathcal{O}\left(\mathbf{w},-\alpha\right)},\\
    &Z\left( \begin{array}{c|ccc|c}
    \mathbb{1}_{a_l} & \mathbb{0}_{a_l,a_r} & \mathbb{0}_{a_l,(n-a_l-a_r)} & \mathbb{1}_{a_l} & \mathbb{0}_{a_l,a_r} \\
    \mathbb{0}_{a_r,a_l} & \mathbb{1}_{a_r} & \mathbb{0}_{a_r,(n-a_l-a_r)} & \mathbb{0}_{a_r,a_l} & \mathbb{1}_{a_r}
    \end{array}  
    \right), \label{eq:stab_ancilla2} 
\end{align}
generate the same group of $Z$-type stabilizers, we use the fact that $\mathcal{C}^\perp=\langle \tilde{\mathbf{p}}_j \rangle$ is a cyclic code and follow the proof of \cref{thm:sync_proof_subspace}.

To show that for any integer $\alpha$, $\{X(\overbracket[0.5pt]{\underbracket[0.5pt]{\,\tilde{\mathbf{p}}_j\,}_{\mathclap{a_l+\alpha}}}^{\text{last}}|\underbracket[0.5pt]{\,\tilde{\mathbf{p}}_j\,}_{\mathclap{n}}|\overbracket[0.5pt]{\underbracket[0.5pt]{\,\tilde{\mathbf{p}}_j\,}_{\mathclap{a_r-\alpha}}}^{\text{first}})\}$ generates the same group of $X$-type gauge operators, we use the fact that $\mathcal{C}^\perp=\langle \tilde{\mathbf{p}}_j \rangle$ is a cyclic code and follow the proof of \cref{thm:sync_proof_subspace}.

To show that for any integer $\alpha$,
\begin{align}
    &Z(\underbracket[0.5pt]{\vphantom{\tilde{\mathbf{p}}_j}\,\mathbf{0}\,}_{\mathclap{a_l+\alpha}}|\underbracket[0.5pt]{\,\tilde{\mathbf{p}}_j\,}_{\mathclap{n}}|\underbracket[0.5pt]{\vphantom{\tilde{\mathbf{p}}_j}\,\mathbf{0}\,}_{\mathclap{a_r-\alpha}})(-1)^{\tilde{\mathbf{p}}_j \cdot \mathcal{O}\left(\mathbf{w},-\alpha\right)},\\
    &Z(\underbracket[0.5pt]{\vphantom{\mathbf{q}'_m}\,\mathbf{0}\,}_{\mathclap{a_l+\alpha}}|\underbracket[0.5pt]{\,\mathbf{q}'_m\,}_{\mathclap{n}}|\underbracket[0.5pt]{\vphantom{\mathbf{q}'_m}\,\mathbf{0}\,}_{\mathclap{a_r-\alpha}})(-1)^{\mathbf{q}'_m \cdot \mathcal{O}\left(\mathbf{w},-\alpha\right)},
\end{align}
and the operators in \cref{eq:stab_ancilla2} generate the same group of $Z$-type gauge operators, we start by considering $\mathcal{C}=\langle \mathbf{p}_{l'} \rangle$ and $\mathcal{D}=\langle \mathbf{q}_{m'} \rangle$, where $l' \in \{1,\dots,k_c\}, m' \in \{1,\dots,k_d\}$. Let $\mathcal{S}^{z,\alpha}_\mathcal{D}$, $\mathcal{S}^{z,\alpha}_\mathcal{C}$, and $\mathcal{S}^{z,\alpha}_{\mathcal{C}^\perp}$ be the following groups, 
\begin{align}
    \mathcal{S}^{z,\alpha}_\mathcal{D}&=\langle Z(\underbracket[0.5pt]{\vphantom{\mathbf{q}_{m'}}\,\mathbf{0}\,}_{\mathclap{a_l+\alpha}}|\underbracket[0.5pt]{\,\mathbf{q}_{m'}\,}_{\mathclap{n}}|\underbracket[0.5pt]{\vphantom{\mathbf{q}_{m'}}\,\mathbf{0}\,}_{\mathclap{a_r-\alpha}})(-1)^{\mathbf{q}_{m'} \cdot \mathcal{O}\left(\mathbf{w},-\alpha\right)}\rangle,\\
    \mathcal{S}^{z,\alpha}_\mathcal{C}&=\langle Z(\underbracket[0.5pt]{\vphantom{\mathbf{p}_{l'}}\,\mathbf{0}\,}_{\mathclap{a_l+\alpha}}|\underbracket[0.5pt]{\,\mathbf{p}_{l'}\,}_{\mathclap{n}}|\underbracket[0.5pt]{\vphantom{\mathbf{p}_{l'}}\,\mathbf{0}\,}_{\mathclap{a_r-\alpha}})(-1)^{\mathbf{p}_{l'} \cdot \mathcal{O}\left(\mathbf{w},-\alpha\right)}\rangle,\\
    \mathcal{S}^{z,\alpha}_{\mathcal{C}^\perp}&=\langle Z(\underbracket[0.5pt]{\vphantom{\tilde{\mathbf{p}}_j}\,\mathbf{0}\,}_{\mathclap{a_l+\alpha}}|\underbracket[0.5pt]{\,\tilde{\mathbf{p}}_j\,}_{\mathclap{n}}|\underbracket[0.5pt]{\vphantom{\tilde{\mathbf{p}}_j}\,\mathbf{0}\,}_{\mathclap{a_r-\alpha}})(-1)^{\tilde{\mathbf{p}}_j \cdot \mathcal{O}\left(\mathbf{w},-\alpha\right)} \rangle.
\end{align}
We can show that $\mathcal{S}^{z,\alpha}_\mathcal{D}$ (or $\mathcal{S}^{z,\alpha}_\mathcal{C}$ or $\mathcal{S}^{z,\alpha}_{\mathcal{C}^\perp}$) are the same group for all $\alpha$ up to a multiplication of some operators in \cref{eq:stab_ancilla2} using the fact that $\mathcal{D}$ (or $\mathcal{C}$ or $\mathcal{C}^\perp$) is a cyclic code and follow the proof of \cref{thm:sync_proof_subspace}.

Consider $\mathcal{C}^\perp=\langle \tilde{\mathbf{p}}_j \rangle$, $\mathcal{C}=\langle \tilde{\mathbf{p}}_j,\mathbf{s}_l^x \rangle$, $\mathcal{D}=\langle \tilde{\mathbf{p}}_j,\mathbf{s}_l^x,\mathbf{q}'_m \rangle$, where $j \in \{1,\dots, n-k_c\}$, $l \in \{1,\dots, 2k_c-n\}$, $m \in \{1,\dots, k_d-k_c\}$ (which is possible by \cref{thm:op_pairing}). We find that $\langle \tilde{\mathbf{p}}_j,\mathbf{q}'_m \rangle$ is the group $\left\langle\mathcal{D}/\mathcal{C},\mathcal{C}^\perp\right\rangle$. Therefore, the group
\begin{equation}
    \left\langle Z(\underbracket[0.5pt]{\vphantom{\tilde{\mathbf{p}}_j}\,\mathbf{0}\,}_{\mathclap{a_l+\alpha}}|\underbracket[0.5pt]{\,\tilde{\mathbf{p}}_j\,}_{\mathclap{n}}|\underbracket[0.5pt]{\vphantom{\tilde{\mathbf{p}}_j}\,\mathbf{0}\,}_{\mathclap{a_r-\alpha}})(-1)^{\tilde{\mathbf{p}}_j \cdot \mathcal{O}\left(\mathbf{w},-\alpha\right)},Z(\underbracket[0.5pt]{\vphantom{\mathbf{q}'_m}\,\mathbf{0}\,}_{\mathclap{a_l+\alpha}}|\underbracket[0.5pt]{\,\mathbf{q}'_m\,}_{\mathclap{n}}|\underbracket[0.5pt]{\vphantom{\mathbf{q}'_m}\,\mathbf{0}\,}_{\mathclap{a_r-\alpha}})(-1)^{\mathbf{q}'_m \cdot \mathcal{O}\left(\mathbf{w},-\alpha\right)}\right\rangle,
\end{equation}
is $\left\langle\mathcal{S}^{z,\alpha}_\mathcal{D} / \mathcal{S}^{z,\alpha}_\mathcal{C} ,\mathcal{S}^{z,\alpha}_{\mathcal{C}^\perp}\right\rangle$.

Since $\mathcal{S}^{z,\alpha}_\mathcal{D}$ (or $\mathcal{S}^{z,\alpha}_\mathcal{C}$ or $\mathcal{S}^{z,\alpha}_{\mathcal{C}^\perp}$) are the same group for all $\alpha$ up to a multiplication of some operators in \cref{eq:stab_ancilla2}, we find that for any integer $\alpha$, 
\begin{align}
    &Z(\underbracket[0.5pt]{\vphantom{\tilde{\mathbf{p}}_j}\,\mathbf{0}\,}_{\mathclap{a_l+\alpha}}|\underbracket[0.5pt]{\,\tilde{\mathbf{p}}_j\,}_{\mathclap{n}}|\underbracket[0.5pt]{\vphantom{\tilde{\mathbf{p}}_j}\,\mathbf{0}\,}_{\mathclap{a_r-\alpha}})(-1)^{\tilde{\mathbf{p}}_j \cdot \mathcal{O}\left(\mathbf{w},-\alpha\right)}, \\
    &Z(\underbracket[0.5pt]{\vphantom{\mathbf{q}'_m}\,\mathbf{0}\,}_{\mathclap{a_l+\alpha}}|\underbracket[0.5pt]{\,\mathbf{q}'_m\,}_{\mathclap{n}}|\underbracket[0.5pt]{\vphantom{\mathbf{q}'_m}\,\mathbf{0}\,}_{\mathclap{a_r-\alpha}})(-1)^{\mathbf{q}'_m \cdot \mathcal{O}\left(\mathbf{w},-\alpha\right)},
\end{align}
and the operators in \cref{eq:stab_ancilla2} generate the same group of $Z$-type gauge operators.
\end{proof}

\subsection{Proof of \cref{thm:CSS_hybrid}} \label{subsec:CSS_hybrid_proof}

\begin{proof}
    Similar to the original CSS construction, both the inner stabilizer group 
    \begin{equation}
        \mathcal{S}_0=\left\langle X(\mathcal{C}_z^\perp),Z(\mathcal{C}_x^\perp)\right\rangle
    \end{equation}
    and outer stabilizer group 
    \begin{equation}
        \mathcal{S}=\left\langle X(\mathcal{D}_z^\perp),Z(\mathcal{D}_x^\perp)\right\rangle
    \end{equation}
    are Abelian groups, as the support of any $X$-type stabilizer (i.e., its non-identity tensor components) will intersect with the support of any $Z$-type stabilizer at an even number of locations, causing them to commute. Additionally we have $\mathcal{S}\subseteq\mathcal{S}_0$.

    We will show that the hybrid code defined by the two stabilizer groups generate a code with the desired properties. The centralizers of the two stabilizer groups can be written in terms of the classical codes:
    \begin{align*}
        C(\mathcal{S}_0) & = \left\langle iI^{\otimes n},X(\mathcal{C}_x),Z(\mathcal{C}_z)\right\rangle, \\
        C(\mathcal{S}) & = \left\langle iI^{\otimes n},X(\mathcal{D}_x),Z(\mathcal{D}_z)\right\rangle.
    \end{align*}
    The quantum information is transmitted by the inner code and its translations, so the number of logical qubits is $n-(n-k_x)-(n-k_z)=k_x+k_z-n$. The number of logical classical messages that can be sent using an $X$-type translation operator is given by $\left\lvert \mathcal{D}_x\right\rvert / \left\lvert \mathcal{C}_x\right\rvert=2^{m_x-k_x}$, or $m_x-k_x$ logical classical bits. The same holds for $Z$-type translation operators, so the total number of logical classical bits is $m=m_x+m_z-k_x-k_z$.

    The minimum distance of a hybrid stabilizer code is given by $d = \min\wt\left(C(\mathcal{S})\setminus\left\langle iI^{\otimes n}, \mathcal{S}_0\right\rangle\right)$. Since the product of $X$ and $Z$-type Pauli elements will always have a greater weight than the smaller of the two $X$ and $Z$-type Pauli elements, the minimum distance will be the minimum of $d_x$ and $d_z$, where $d_x$ (resp. $d_z$) is the minimum weight of a nontrivial $X$-type (resp. $Z$-type) logical operator. The $X$-type logical operators of $C(\mathcal{S})$ correspond to the codewords of $\mathcal{D}_x$, while the $X$-type trivial operators of $\mathcal{S}_0$ correspond to the codewords of $\mathcal{C}_z^\perp$, so we have 
    \begin{equation}
        d_x = \min\wt\left(\mathcal{D}_x\setminus \mathcal{C}_z^\perp\right).
    \end{equation}
    The same argument applied to the $Z$-type logical operators recovers $d_z$.
\end{proof}

\subsection{Proof of \cref{thm:CSS_hybrid_subsystem}}\label{subsec:css_hybrid_subsystem_proof}

\begin{proof}
    It is easy to check that the $X$ and $Z$-type operators in the inner stabilizer group $\mathcal{S}_0$ and the outer stabilizer group $\mathcal{S}$ commute with each other. Similar to the hybrid CSS code construction, we have $\mathcal{S}\subseteq\mathcal{S}_0$. Note that the constraints $\left(\mathcal{D}_x\setminus \mathcal{C}_x\right)\cap \mathcal{C}_z^\perp=\left(\mathcal{D}_z\setminus \mathcal{C}_z\right)\cap \mathcal{C}_x^\perp=\emptyset$ guarantee that the classical logical operators are not also gauge operators.

    As in the hybrid CSS code construction, the number of classical logical operators is $m = m_x+m_z-k_x-k_z$. The code $\mathcal{C}_x\cap \mathcal{C}_z^\perp$ is the dual of $\mathcal{C}_z + \mathcal{C}_x^\perp$, and thus $\lvert \mathcal{C}_x\cap \mathcal{C}_z^\perp\rvert = 2^{n-r_z}$. The number of $X$-type gauge generators that are not also stabilizers is given by $\log_2\left(\left\lvert \mathcal{C}_z^\perp \right\rvert / \left\lvert \mathcal{C}_x\cap \mathcal{C}_z^\perp\right\rvert\right)=r_z-k_z$. Since the same must hold for the $Z$-type operators, we obtain $r = r_x-k_x = r_z-k_z$. A similar argument tells us there are $\log_2\left( \left\lvert \mathcal{C}_x \right\rvert / \left\lvert \mathcal{C}_x\cap \mathcal{C}_z^\perp \right\rvert\right)=r_z+k_x-n$ generators of $X$-type quantum logical operators. As with the gauge case, the same must hold for the $Z$-type operators, so we obtain $k = r_x+k_z-n = r_z+k_x-n$.

    From Dauphinais et al. \cite{Dauphinais2024}, the minimum distance of a hybrid subsystem code is
    \begin{equation}\label{eq:dauphinais_min_dist}
        d=\min\wt\left(\left(C(\mathcal{S}_0)\setminus\mathcal{G}_0\right)\cup\left(\bigcup\limits_{i\neq j}t_iC(\mathcal{S}_0)t_j^\dagger\right)\right),
    \end{equation}
    where $t_i$ are the classical logical operators that represent the cosets of $C(\mathcal{S})/C(\mathcal{S}_0)$. Since we are restricted to the stabilizer case, we can rewrite this set as
    \begin{align}
        \left(C(\mathcal{S}_0) \setminus\mathcal{G}_0\right)\cup\left(\bigcup\limits_{i\neq j}t_iC(\mathcal{S}_0)t_j^\dagger\right) & = \left(C(\mathcal{S}_0) \setminus\mathcal{G}_0\right)\cup\left(\bigcup\limits_{i\neq j}t_it_j^\dagger C(\mathcal{S}_0)\right) \\
        & = C(\mathcal{S}) \setminus\mathcal{G}_0.
    \end{align}
    Therefore, as in the CSS code construction, the minimum distance will be the minimum of $d_x$ and $d_z$, which are the minimum weight $X$ and $Z$-type logical operators in $C(\mathcal{S})$ that are not gauge operators, that is
    \begin{align}
        d_x & = \min\wt\left(\left(\mathcal{D}_x+\mathcal{D}_z^\perp\right)\setminus \mathcal{C}_z^\perp\right), \\
        d_z & = \min\wt\left(\left(\mathcal{D}_z+\mathcal{D}_x^\perp\right)\setminus \mathcal{C}_x^\perp\right).
    \end{align}
\end{proof}

\bibliographystyle{quantum}

\bibliography{refs}

\end{document}